\documentclass[12pt]{article}
\usepackage{geometry}
\usepackage[round]{natbib}
\usepackage{hyperref}
\usepackage{color}
\usepackage{amsmath}
\usepackage{algorithm}
\usepackage{algorithmic}
\usepackage{graphicx}
\usepackage{float} 
\usepackage{color}
\usepackage{multirow}
\usepackage{cleveref}
\usepackage{authblk}
\usepackage{setspace}
\usepackage{mathrsfs}
\usepackage{subcaption}
\usepackage{ragged2e} 
\usepackage{enumitem}
\usepackage{xcolor}
\usepackage{stmaryrd}%空心方括号
\usepackage{amsmath,amssymb,amsthm}

\def\bse{\begin{eqnarray*}}
\def\ese{\end{eqnarray*}}
\def\be{\begin{eqnarray}}
\def\ee{\end{eqnarray}}
\def\bsq{\begin{equation*}}
\def\esq{\end{equation*}}
\def\bq{\begin{equation}}
\def\eq{\end{equation}}
\def\n{\nonumber}

\newcommand{\X}{\mathbf{X}}

\renewcommand{\O}{\mathbf{O}}

\newcommand{\x}{\mathbf{x}}

\newcommand{\U}{\mathbf{U}}

\newcommand{\Z}{\mathbf{Z}}
\newcommand{\z}{\mathbf{z}}

\newcommand{\bbR}{\mathbb{R}}

\newcommand{\bbN}{\mathbb{N}}
\newcommand{\bbE}{\mathbb{E}}

\newcommand{\calG}{{\mathcal{G}}}

\newcommand{\calW}{\mathcal{W}}
\newcommand{\calF}{\mathcal{F}}

\newcommand{\calD}{\mathcal{D}}
\newcommand{\calS}{\mathcal{S}}
\newcommand{\calL}{\mathcal{L}}
\newcommand{\calV}{\mathcal{V}}
\newcommand{\calN}{\mathcal{N}}
\newcommand{\calB}{\mathcal{B}}

\newcommand{\balpha}{\boldsymbol{\alpha}}

\newcommand{\boldeta}{\boldsymbol{\eta}}

\newcommand{\blambda}{\boldsymbol{\lambda}}

\newcommand{\tp}{^\top}

\newcommand{\var}{\mathrm{Var}}
\newcommand{\tr}{\mathrm{tr}}

\newcommand{\Cov}{\mathrm{Cov}}

\renewcommand{\hat}{\widehat}
\renewcommand{\tilde}{\widetilde}

\renewcommand{\exp}{\mathrm{exp}}

\DeclareMathOperator{\argmax}{argmax}

\newtheorem{theorem}{Theorem} 
\newtheorem{lemma}{Lemma}
  
\newtheorem{assumption}{Assumption} 
 
\newtheorem{corollary}{Corollary}
\newtheorem{proposition}{Proposition}

\allowdisplaybreaks[4]
\usepackage{array}
\usepackage{tikz}
\usetikzlibrary{arrows.meta, positioning}
\title{Semiparametric Causal Inference for Right-Censored Outcomes with Many Weak Invalid Instruments }
\author[1,2]{Qiushi Bu}
\author[1]{Wen Su}
\author[3]{Xingqiu Zhao}
\author[4]{Zhonghua Liu$^*$}

\affil[1]{\small Department of Biostatistics, City University of Hong Kong, Hong Kong}
\affil[2]{\small Academy of Mathematics and Systems Science, Chinese Academy of Sciences, Beijing, China}
\affil[3]{\small Department of Applied Mathematics, The Hong Kong Polytechnic University, Hong Kong}
\affil[4]{\small Department of Biostatistics, Columbia University, New York, NY, USA}

\begin{document}
\maketitle 
\begin{abstract}
We propose a semiparametric framework for causal inference with right-censored survival outcomes and many weak invalid instruments, motivated by Mendelian randomization in biobank studies where classical methods may fail. We adopt an accelerated failure time model and construct a moment condition based on augmented inverse probability of censoring weighting, incorporating both uncensored and censored observations. Under a heteroscedasticity-based condition on the treatment model, we establish point identification of the causal effect despite censoring and invalid instruments. We propose GEL-NOW (Generalized Empirical Likelihood with Non-Neyman Orthogonal and Weak moments) for valid inference under these conditions. A divergent number of Neyman orthogonal nuisance functions is estimated using deep neural networks. A key challenge is that the conditional censoring distribution is a non-Neyman orthogonal nuisance, contributing to the first-order asymptotics of the estimator for the target causal effect parameter. We derive the asymptotic distribution and explicitly incorporate this additional uncertainty into the asymptotic variance formula. We also introduce a censoring-adjusted over-identification test that accounts for this new variance component. Simulation studies and UK Biobank applications demonstrate the method's robustness and practical utility. 

\end{abstract}

%--------------------------------------------
\newpage
\section{Introduction}

\subsection{Motivation}

Understanding the causal effects of modifiable risk factors on disease onset is a central objective in epidemiological research. Large-scale cohorts such as the UK Biobank provide rich longitudinal data to study time-to-event (or survival) outcomes and evaluate how exposures, such as body mass index (BMI), influence disease onset \citep{Sudlow2015}. For example, a recent genetic study analyzed more than 800 time-to-event phenotypes derived from linked electronic health records (EHR) in the UK Biobank \citep{dey2022}.

A major obstacle to establishing causal relationships between modifiable risk factors and survival outcomes in observational studies is unmeasured confounding, which can bias effect estimates and lead to misleading conclusions. Mendelian randomization (MR) \citep{katan1986apoupoprotein,smith2003mendelian} addresses this problem by using genetic variants as instrumental variables (IVs), and has become a widely used approach in epidemiology, particularly in large-scale biobank studies where hundreds or thousands of candidate instruments are available.

However, applying MR to survival outcomes presents unique challenges. First, survival data are often subject to right-censoring, as in the UK Biobank \citep{dey2022}. Second, many genetic instruments may be weak \citep{liu2023,ye2024} or invalid, for example, in the presence of pervasive horizontal pleiotropy \citep{Verbanck2018}, violating core IV assumptions \citep{Didelez2010}. These complications render classical IV estimators biased or inconsistent, and recent methodological advances for uncensored continuous outcomes \citep{kang2024identification,lin2024} are not directly applicable. Existing semiparametric frameworks such as generalized empirical likelihood (GEL) \citep{smith1997alternative,newey2004} and generalized method of moments \citep{hansen1982,wang2025} typically rely on Neyman orthogonality \citep{neyman1959,neyman1979,Chernozhukov2018}, an assumption that fails when key nuisances, most notably the conditional censoring distribution, enter the proposed moment equations non-orthogonally. Together, these issues highlight a sharp methodological gap: there is currently no MR framework capable of delivering valid causal inference for censored survival outcomes in the presence of many weak and potentially invalid instruments.

\subsection{Related work}

While various strands of research have addressed aspects of this problem, a unified solution has not, to the best of our knowledge, been established. Classical IV methods for censored survival outcomes have been proposed \citep{nie2011,Tchetgen2015,Ertefaie2018}, but they assume fully valid instruments. GEL has shown robustness under many weak moments \citep{newey2009}, but existing theory does not extend to non-Neyman orthogonal and weak moment conditions with a diverging number of nuisance functions. Modern machine learning-based semiparametric approaches, such as double/debiased machine learning \citep{Chernozhukov2018}, rely on orthogonality assumption and cross-fitting techniques that are not applicable in censored survival settings where the censoring distribution is shown to be a non-Neyman orthogonal nuisance. 

\subsection{Our contributions}
\label{sec:contributions}

We propose \textbf{\texttt{MAWII-Surv}} (MAny Weak and Invalid Instruments for Survival outcomes), a unified framework for semiparametric causal inference with right-censored survival outcomes and many weak invalid instruments. Identification of causal effect is achieved through a heteroscedasticity-based strategy under a structural accelerated failure time (AFT) model \citep{kalbfleisch2002statistical,robins1992estimation}. Heteroscedasticity is common in genetics \citep{Pare2010,westerman2022variance,xiang2025genome}; for example, BMI exhibits variance differences across genetic variants \citep{yang2012fto,wang2019,miao2022quantile}. 

For estimation and inference, we introduce GEL-NOW (Generalized Empirical Likelihood with Non-Neyman Orthogonal and Weak moments), or GEL 2.0, a new semiparametric framework that incorporates non-Neyman orthogonal nuisances into moment conditions, accommodates many weak instruments, and explicitly derives variance inflation. This represents a fundamental theoretical advance, extending GEL beyond its classical domain.

Compared with GEL under many weak moments \citep{newey2009}, our approach substantially broadens the scope by accommodating (i) a diverging number of nuisance functions estimated via flexible deep neural networks (DNNs) \citep{Goodfellow2016}, and (ii) non-Neyman orthogonal nuisances that influence first-order asymptotics. Relative to \citet{ye2024}, who studied uncensored outcomes with weak and invalid instruments under orthogonality, our framework extends to censored survival outcomes, where the censoring distribution itself is a non-Neyman orthogonal nuisance. In contrast to existing methods, we therefore develop a new semiparametric causal inference framework for censored survival outcomes that simultaneously addresses many weak and invalid instruments, a diverging number of nuisance functions, and non-orthogonality.

\begin{enumerate}
\item \textbf{Identification under right-censoring and invalid instruments.} We establish point identification of the causal effect on log-survival time in a structural AFT model under a heteroscedasticity-based condition for the treatment variable, generalizing prior results \citep{Lewbel2012,Tchetgen2021,ye2024} to right-censored outcomes.

\item \textbf{Semiparametric inference via GEL 2.0.} Unlike previous frameworks that rely on orthogonality or cross-fitting \citep{Chernozhukov2018,farrell2021,xu2022deepmed,ye2024,wang2025}, we develop a new GEL-based estimator that accommodates non-Neyman orthogonal and weak moments with a diverging number of nuisance functions. We estimate orthogonal nuisances using deep neural networks (DNNs) \citep{Goodfellow2016}, avoiding cross-fitting via localization-based empirical process techniques \citep{bartlett2005,farrell2021}, and derive a variance decomposition that explicitly incorporates additional variability from estimating non-Neyman orthogonal nuisances such as the conditional censoring distribution in our setting.

\item \textbf{Diagnostics and software.} We propose a new over-identification test that refines classical Sargan–Hansen tests \citep{sargan1958,hansen1982,ye2024} to censored outcomes. We implement the framework in a scalable Python package, \texttt{`MAWII-Surv'}, and demonstrate its performance through simulations and UK Biobank applications.
\end{enumerate}

\subsection{Organization}

The remainder of the paper is organized as follows. Section~\ref{sec:probelsetup} introduces the problem setup and  a structural AFT model. Section~\ref{sec:identification} presents the identification strategy and moment condition under right-censoring. Section~\ref{sec:semiparametric} describes the GEL estimation procedure and our approach to nuisance estimation via DNNs. Section~\ref{sec:theoretical} establishes the asymptotic properties of the estimator. Section~\ref{sec:diagnostics} develops model diagnostic tools, including a weak identification measure and a censoring-adjusted over-identification test. Section~\ref{sec:simulation} presents the simulation results. Section~\ref{sec:ukb} applies our method to UK Biobank data to estimate the causal effect of BMI on the time-to-onset  of three diseases. Section~\ref{sec:discussion} concludes with a discussion. All technical proofs are in the Supplementary Material. The \texttt{Python} package \texttt{`MAWII-Surv'} for the proposed methods is publicly available at \url{https://github.com/424051725/mawiisurv}.
\section{Problem Setup}
\label{sec:probelsetup}
For causal inference with right-censored outcomes, we adopt a structural accelerated failure time (AFT) model \citep{robins1992estimation}, which generalizes the classical semiparametric AFT framework \citep{kalbfleisch2002statistical} to permit causal interpretation in the presence of confounding. This approach is motivated by two considerations. First, modeling the logarithm of survival time directly yields a straightforward interpretation of the causal effect of exposure on survival, which is particularly valuable for causal inference. Second, in contrast to proportional hazards models \citep{cox1972regression}, accelerated failure time models remain valid when the proportional hazards assumption is violated, a situation that frequently arises in observational health studies \citep{wei1992accelerated,robins1992semiparametric}. 

Let $\tilde T$ and $\tilde C$ be the true survival time and censoring time, and $\tilde Y=\min(\tilde T,\tilde C)$ be the observed survival time and $\delta=I(\tilde T\leq\tilde C)$ be the indicator of whether the outcome is uncensored. In the standard accelerated failure time model, the response is $\log \widetilde T$, where $\widetilde T>0$ denotes the original survival time. For ease of notation, let $T = \log \widetilde T$, $C = \log \widetilde C$, $Y = \log \widetilde Y$, and throughout the rest of the paper, we use $T$, $C$ and $Y$, to denote the log-transformed survival time, censoring time, and observed time, respectively. 
Moreover, $\Z = (Z_1, \dots, Z_m)\tp$ denotes an $m$-dimensional vector of candidate instrumental variables that might violate core IV assumptions \citep{Didelez2010}, where $m$ is allowed to diverge with the sample size $n$. Let $\X$ represent $d_x$-dimensional observed baseline covariates, and $A$ is the exposure. 
Let $\O = (Y,\delta, A, \Z, \X)$, $\O_T = \{T, A, \Z, \X\}$, and $\O_A = \{A, \Z, \X\}$. 

Let $ \mathbb{E}_{P_X}[f(X)] := \int f(x) \, dP_X(x) $ denote the population expectation under the distribution $P$ of a generic random variable $X$, and let $ \mathbb{E}_n[f(X)] := 1/n \sum_{i=1}^n f(X_i) $ denote the empirical average on the observed sample $ \{X_i\}_{i=1}^n $. When the domain of integration is clear from context, we simplify $ \mathbb{E}_{P_X} $ to $ \bbE $. In particular, the generic random variable $X$ may refer to either $\O$, $\O_T$, or $\O_A$, and $P$ denotes the corresponding distribution in each case. Moreover, we use \( \|\cdot\| \) to denote norms. For a vector \( \x \in \mathbb{R}^d \), we write \( \|\x\| = (\sum_{i=1}^d x_i^2)^{1/2} \) for the Euclidean norm. For a function \( f \colon \mathcal{X} \to \mathbb{R} \), denote \( \|f\|_{L^2} = (\int f^2(x)\, d P_X(x))^{1/2} \) as \( L^2 \) norm, and \( \|f\|_{\infty} = \sup_{x \in \mathcal{X}} |f(x)| \) as the supremum norm. For a matrix $\Sigma$, let $\|\Sigma\|$ be the largest singular value of $\Sigma$.

We consider the following semiparametric  structural AFT model:
\begin{equation}\label{ET}
T=\beta_{0}A+\alpha(\Z,\X)+\omega_{T}(\X,\U)+\epsilon,   
\end{equation}
and the structural equation for the treatment variable $A$:
\begin{equation}\label{EA}
\bbE(A | \Z,\X,\U)=\gamma(\Z,\X)+\omega_{A}(\X,\U),    
\end{equation}
where $\beta_0$ is the true causal effect, $\alpha(\cdot)$, $\gamma(\cdot)$, $\omega_T(\cdot)$, $\omega_A(\cdot)$ are unspecified functions, $\U$ represents the unobserved confounders, and $\epsilon$ is the random error satisfying $\bbE(\epsilon|A,\Z,\X,\U)=0$. The specification in \eqref{ET} is a semiparametric structural AFT  model that leaves the error distribution unspecified, thereby providing flexibility in survival data analysis. Several error distributions commonly used in AFT models satisfy this condition, including the normal distribution (yielding lognormal survival times) and the logistic distribution (yielding log-logistic survival times) \citep{lawless2011statistical}. The function $\alpha(\cdot)$ captures the direct effect of $\Z$ on the outcome $T$. When $\alpha(\cdot) \neq 0$, this direct path violates the exclusion restriction and $\Z$ does not satisfy the conditions to be a valid instrumental variable, as shown in Figure \ref{fig1}. If $\alpha(\cdot)=0$, then $\Z$ is a valid set of IVs.
\begin{figure}[h]
\centering
\begin{tikzpicture}[
node distance=2cm and 3cm,
every node/.style={font=\normalsize},
arrow/.style={-{Latex[width=2mm,length=2mm]}, line width=0.8pt},
dashedarrow/.style={arrow, dashed}
]

% Nodes
\node (Z) at (0,0) {$\Z$};
\node (A) at (3,0) {$A$};
\node (T) at (6,0) {$T$};
\node (U) at (4.5,2) {$\U$};

% Arrows
\draw[arrow] (Z) -- (A);                             % Z → A
\draw[arrow] (A) -- node[above] {\small $\beta_0$} (T); % A → Y
\draw[arrow] (Z) to[out=-30,in=-150] node[below] {\small $\alpha$} (T); % curved Z → Y
\draw[arrow] (U) -- (A);                             % U → A
\draw[arrow] (U) -- (T);                             % U → Y

\end{tikzpicture}
\caption{Relationship between candidate instrumental variable $\Z$, exposure $A$, outcome $T$, and unobserved confounder $\U$, conditional on observed $\X$.}
\label{fig1} 
\end{figure}
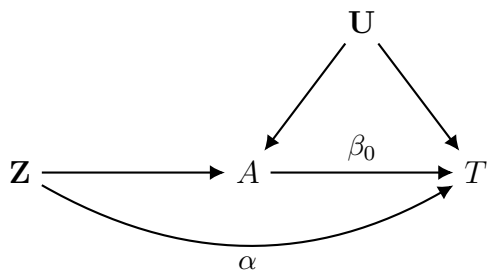

Following common practice in the analysis of right-censored data \citep{kalbfleisch2002statistical}, we adopt the conditional non-informative censoring assumption:
\begin{assumption}[Conditional non-informative censoring]\label{condindependent}
$\tilde T\perp\tilde C |\O_A$.
\end{assumption} 
This assumption states that, conditional on $\O_A$, the censoring mechanism is independent of the failure time. An immediate consequence is that, on the log scale, $T$ and $C$ are also conditionally independent given  $\O_A$.

\section{Identification under the Structural AFT Model}
\label{sec:identification}
% We first outline the identification strategy for the treatment effect under the proposed structural  models, allowing for the possibility that the instrumental variable violates the exclusion restriction:
% \begin{equation}\label{ET-MEAN}
% \bbE(T|A,\Z,\X,\U)=\beta_{0}A+\alpha(\Z,\X)+\omega_{T}(\X,\U).
% \end{equation}
% If $T$ were uncensored, then taking the conditional expectation of \eqref{ET} yields the mean model in \eqref{ET-MEAN}.
% Assumption \ref{condcondind} provides sufficient assumptions for the identifiability of $\beta_0$ in this setting. Let $R_A = A-\bbE(A|\Z,\X)$.
We begin by outlining the identification strategy for the treatment effect under the proposed structural model, which allows the instrumental variable to violate the exclusion restriction:
\begin{equation}\label{ET-MEAN}
\bbE(T \mid A, \Z, \X, \U) = \beta_{0} A + \alpha(\Z,\X) + \omega_{T}(\X,\U).
\end{equation}
If $T$ were fully observed, then taking the conditional expectation of \eqref{ET} would yield the mean model in \eqref{ET-MEAN}. 
Assumption~\ref{condcondind} specifies sufficient conditions for the identifiability of $\beta_{0}$ in this setting. For later use, define $R_A = A - \bbE(A \mid \Z,\X)$.

%When there are no censoring observations and no covariates $\X$, \cite{Tchetgen2021} proved that $\beta_0$ can be uniquely determined under $\Z\perp\U$ and the heteroscedasticity condition $\Cov(\Z,\var(A|\Z))\neq 0$. Building on this, \cite{ye2024} proposed a more general identification strategy for the setting with covariates and uncensored outcomes. 

\begin{assumption}[Identification assumptions] \label{condcondind}
(i).   $\Z\perp\U|\X;$  
(ii). $\bbE\{R_A(T-\beta_0 A)|\Z,\X\}=\bbE\{R_A(T-\beta_0 A)|\X\};$
(iii). $\Cov(\Z,\var(A|\Z,\X)|\X) \neq 0.$    
\end{assumption}
Assumption \ref{condcondind} (i) corresponds to the independence assumption commonly required for instrumental variables, namely that $\Z$ is independent of unmeasured confounders $\U$ given covariates $\X$. According to Lemma 3.2 of \cite{Tchetgen2021}, the conditional independent assumption $\Z \perp \U|\X$ can be relaxed to $\Cov(\omega_T(\X,\U),\omega_A(\X,\U)|\Z,\X)=\Cov(\omega_T(\X,\U),\omega_A(\X,\U)|\X)$. This relaxation accommodates violations of the instrumental variable independence assumption.
Assumption \ref{condcondind} (ii) ensures that conditional on $\X$, $R_A(T - \beta_0 A)$ is independent of $\Z$. Assumption \ref{condcondind} (iii) requires that $\var(A|\Z,\X)$ vary with $\Z$ given $\X$, that is, $\Z$ induces heteroscedasticity in $A$ conditional on $\X$. This assumption enables identification  even when $\Z$ may not satisfy the  exclusion restriction assumption. Similar heteroscedasticity-based identification has been used in the G-Estimation under No Interaction with Unmeasured Selection (GENIUS) approach \citep{Lewbel2012,lewbel2018,Tchetgen2021,ye2024} for robust Mendelian randomization. Assumption~\ref{condcondind}(iii) imposes a heteroscedasticity condition that is empirically testable, and can be assessed using conventional tests for heteroscedasticity \citep{Pare2010,Tchetgen2021,liu2023,ye2024}. Lemma~\ref{lemma0} extends the influence function for $\beta_{0}$, derived in \citet{ye2024}, to the semiparametric structural AFT model in the case where the log-survival time is fully observed.

\begin{lemma} \label{lemma0}
Under the structural  models \eqref{EA} and \eqref{ET-MEAN} and Assumptions \ref{condindependent}-\ref{condcondind}, the influence function of $\beta_0$ is 
$$
g(\beta;\O_{T})=R^h_Z\{R_AR_T-\beta R_A^2-\bbE(R_AR_T-\beta R_A^2|\X)\},
$$
where $R^h_Z=h(\Z)-\bbE(h(\Z)|\X)$ for any scalar-valued function $h$, and $R_T=T-\bbE(T|\Z,\X)$. And the closed-form solution for $\beta_0$ is 
\be\label{closeform1}
\beta_0=\frac{\bbE(R^h_Z\{R_AR_T-\bbE(R_AR_T|\X)\})}{\bbE(R^h_Z\{R_A^2-\bbE(R_A^2|\X)\})}.
\ee
\end{lemma}

In the presence of right censoring, the true log-survival time $T$ is unobserved for censored subjects, and only $Y$ is available. Consequently, the influence function above cannot be applied directly. Simply substituting $T$ with $Y$ without adjusting for censoring would introduce bias. To address this issue, we follow \citet{tang2020} and employ an augmented inverse probability censoring weighting (AIPCW) approach to construct a new moment function $\psi(\beta;\O)$ as follows
\begin{equation}\label{psi}
\psi(\beta;\O)=\frac{\delta}{G_0(Y|\O_A)}g(\beta;\O)+\left(1-\frac{\delta}{G_0(Y|\O_A)}\right)\xi_0(\beta;\O_A).   
\end{equation}
Here $G_{0}(y | \O_{A}) = P(C > y | \O_{A}) > 0$, the standard positivity assumption, refers to the conditional distribution of the log-censoring time $C$, while $\xi_{0}(\beta; \O_{A}) = \bbE[g(\beta; \O_{T}) \mid \O_{A}]$ denotes the conditional expectation of the uncensored-data influence function given the observed  data $\O_{A}$.

The proposed  AIPCW moment function \eqref{psi} contains two components: a weighted term $\delta G_0^{-1}(Y | \O_A) g(\beta; \O)$ and an augmentation term $(1-\delta/G_0(Y|\O_A))\xi_0(\beta;\O_A)$. For uncensored observations ($\delta = 1$), the true event time is observed ($Y = T$), and the original moment function $g(\beta;\O_A)$ is reweighted by the inverse of the conditional probability of being uncensored, $1 / G_0(Y|\O_A)$.  Censored observations ($\delta=0$) contribute to the moment function through $\xi_0(\beta;\O_A)$, which may be viewed as predicting the unobserved failure time $T$ using the observed data $\O_A$, and then evaluating the moment function at this imputed value. An important advantage of the AIPCW moment function $\psi(\beta; \O)$ is its linearity in $\beta$. This follows from the fact that the main component $g(\beta; \O)$ is linear in $\beta$, and the correction term $\xi_0(\beta; \O_A)$ is a conditional expectation that does not alter this linearity. Therefore, the AIPCW moment function $\psi(\beta; \O)$ remains linear in $\beta$.

We now show that the proposed AIPCW moment function \eqref{psi} is uniquely satisfied at the true parameter $\beta_{0}$. This property ensures identification under the structural AFT model.
\begin{theorem}\label{thmid}
Under Assumptions \ref{condindependent} and \ref{condcondind}, the equation $\bbE{\psi(\beta;\O)}=0$ admits a unique solution $\beta=\beta_0$, with the closed-form expression for the true causal effect $\beta_0$ given by 
\be\label{closeform2}
\beta_0=\frac{\bbE(R^h_Z\{\frac{\delta}{G_0}R_AR_Y+(1-\frac{\delta}{G_0})R_A\bbE(R_Y|\O_A)-\bbE(R_AR_Y|\X)\})}{\bbE(R^h_Z\{R_A^2-\bbE(R_A^2|\X)\})},
\ee
where $R_Y = Y-\bbE(Y|\Z,\X)$.
\end{theorem}
Theorem~\ref{thmid} establishes point identification of the true causal effect $\beta_{0}$ from the observed data. From the closed-form solutions in both the uncensored and censored cases, it is evident that the denominator is identical, implying that identification in both settings relies on the heteroscedasticity of the treatment variable $A$. When the censoring rate is zero, so that $G_{0}(Y \mid \O_{A}) \equiv 1$, the expression reduces exactly to \eqref{closeform1}.

As noted in Lemma \ref{lemma0}, the function $h(\Z)$ can be any scalar-valued function.  In practice, we choose $h(\Z) = Z_j$ for each $j = 1, \dots, m$, and stack them to construct an $m$-dimensional moment function. This allows us to estimate $\beta$ using all components of $\Z$. That is, the resulting $g(\cdot)$ in AIPCW moment function \eqref{psi} takes the form
$$
g(\beta;\O)=R_Z\{R_AR_Y-\beta R_A^2-\bbE(R_AR_Y-\beta R_A^2|\X)\},
$$
where $R_Z=\Z-\bbE(\Z|\X)$. When multiple candidate instruments are available for a single scalar causal parameter, the parameter is said to be \emph{over-identified}, meaning that the number of moment conditions exceeds the number of parameters. This has two main advantages: (i) efficiency can be improved by aggregating information across multiple moment conditions; and (ii) over-identification permits specification checks through tests of over-identifying restrictions, such as the Sargan–Hansen test \citep{sargan1958,hansen1982}. By contrast, with a single candidate instrument the model is exactly identified, and such tests are not available.

\section{Estimation via GEL 2.0}

\label{sec:semiparametric}

\subsection{GEL estimation procedure}
\label{subsec:GEL}

% Given an $m$-dimensional generic moment function $\Psi(\beta; \O)$, a classical estimator is the  two-step generalized method of moments (GMM, \citet{hansen1982}). However, when identification is weak or moments are nearly collinear, two-step GMM can exhibit non-negligible finite-sample bias and size distortions \citep{newey2004,newey2009}. To improve stability in these regimes, we adopt the generalized empirical likelihood (GEL) framework \citep{newey2004}.

Given an $m$-dimensional generic moment function $\Psi(\beta; \mathcal{O})$, the generalized method of moments (GMM; \citealp{hansen1982}) is widely used. We adopt the generalized empirical likelihood (GEL) framework \citep{newey2004}, which provides a likelihood-based alternative for estimation from moment conditions and can deliver improved stability in certain settings. GEL constructs likelihood-type tilting weights that enforce the sample moment conditions and is known to possess superior higher-order properties, to allow refinements such as Bartlett correction \citep{Imbens2002,kitamura2004}, and to avoid the explicit inversion of ill-conditioned weighting matrices. We therefore estimate the target causal parameter using a GEL criterion, defined as
$$
\hat\beta=\arg\min_{\beta\in \calB} \hat Q(\beta)=\arg\min_{\beta\in \calB}\sup_{\blambda\in L(\beta)}\frac{1}{n}\sum_{i=1}^n\rho(\blambda\tp \Psi(\beta;\O_i)),
$$
where $\rho(\cdot)$ is a concave function defined on an open interval $\calV$ containing 0, $\lambda$ is an $m$-dimensional vector and $L(\beta) = \{\blambda: \blambda\tp \Psi(\beta;\O_i)\in \calV\}$. Within the GEL framework, the choice of concave function $\rho(\cdot)$ determines the specific estimator. Common choices include $\rho(v) = -v - v^{2}/2$ for the continuous updating estimator (CUE), $\rho(v) = -e^{v} + 1$ for exponential tilting (ET), and $\rho(v) = \log(1-v)$ for empirical likelihood (EL). While \citet{ye2024} employed CUE to address weak identification bias, CUE is generally regarded as less attractive within the GEL family because of its relatively larger higher-order bias and the absence of refinements such as Bartlett correction \citep{Imbens2002,kitamura2004}. In addition, CUE requires inversion of the moment covariance matrix, which may be numerically unstable when the moment conditions are nearly collinear.

Figure \ref{fig:rho_plots} plots the three commonly used functions $\rho$ together with their first and second derivatives, allowing a direct comparison of their local curvature near zero and their tail behavior, which in turn highlights the differences in higher-order properties across the GEL family. As shown in panel (a), all three functions behave similarly near zero. However, notable differences emerge in their higher-order properties: CUE has a constant second-order derivative, reflecting its limited curvature and inability to accommodate higher-order bias corrections. In contrast, EL and ET exhibit sharper curvature because of their logarithmic and exponential forms. Panel (d) zooms out on a wider range of $v$, further highlighting that CUE imposes substantially weaker penalties on large deviations compared to EL and ET, resulting in a flatter curve away from zero. This weaker penalization may lead to reduced numerical stability in finite samples. In contrast, the steeper slopes of EL and ET provide stronger control over tail behavior, making them more favorable under weak identification. 
\begin{figure}[htbp]
\centering
\resizebox{0.8\textwidth}{!}{%
\begin{minipage}{1\textwidth}
\subfloat[$\rho(v)$]{%
\includegraphics[width=0.4\textwidth]{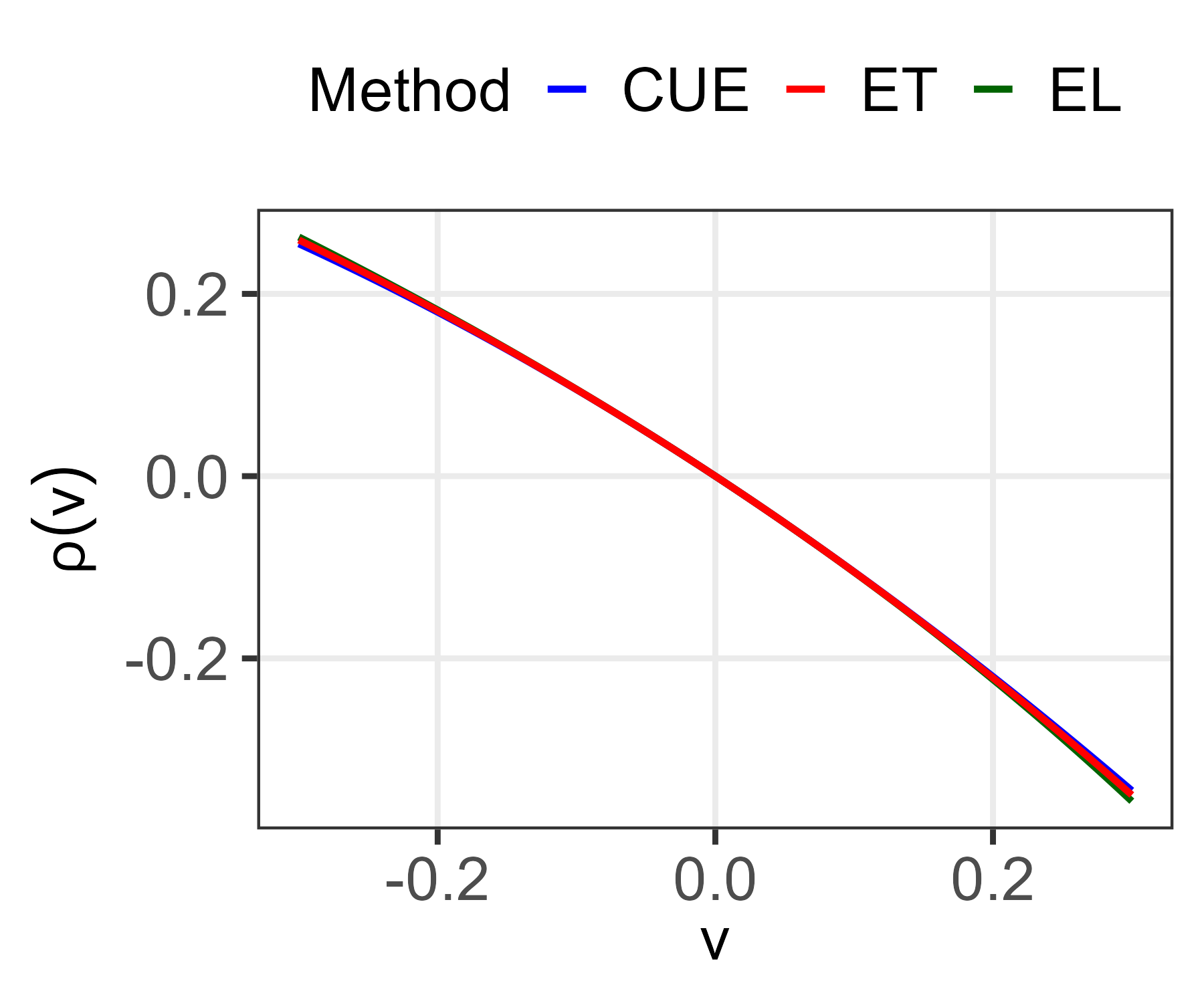}}
\hfill
\subfloat[$\rho'(v)$]{%
\includegraphics[width=0.4\textwidth]{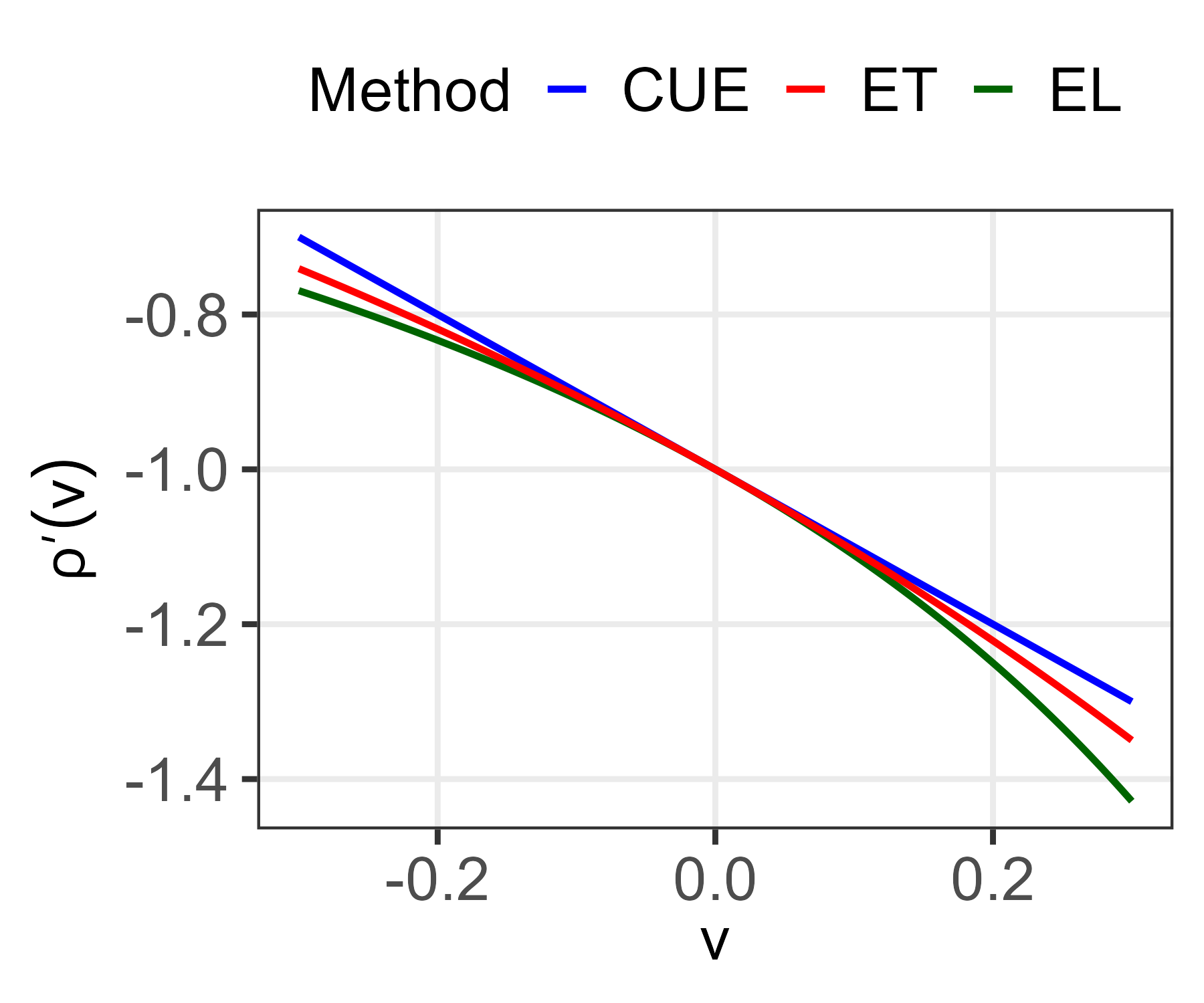}}
\vspace{1em}
\subfloat[$\rho''(v)$]{%
\includegraphics[width=0.4\textwidth]{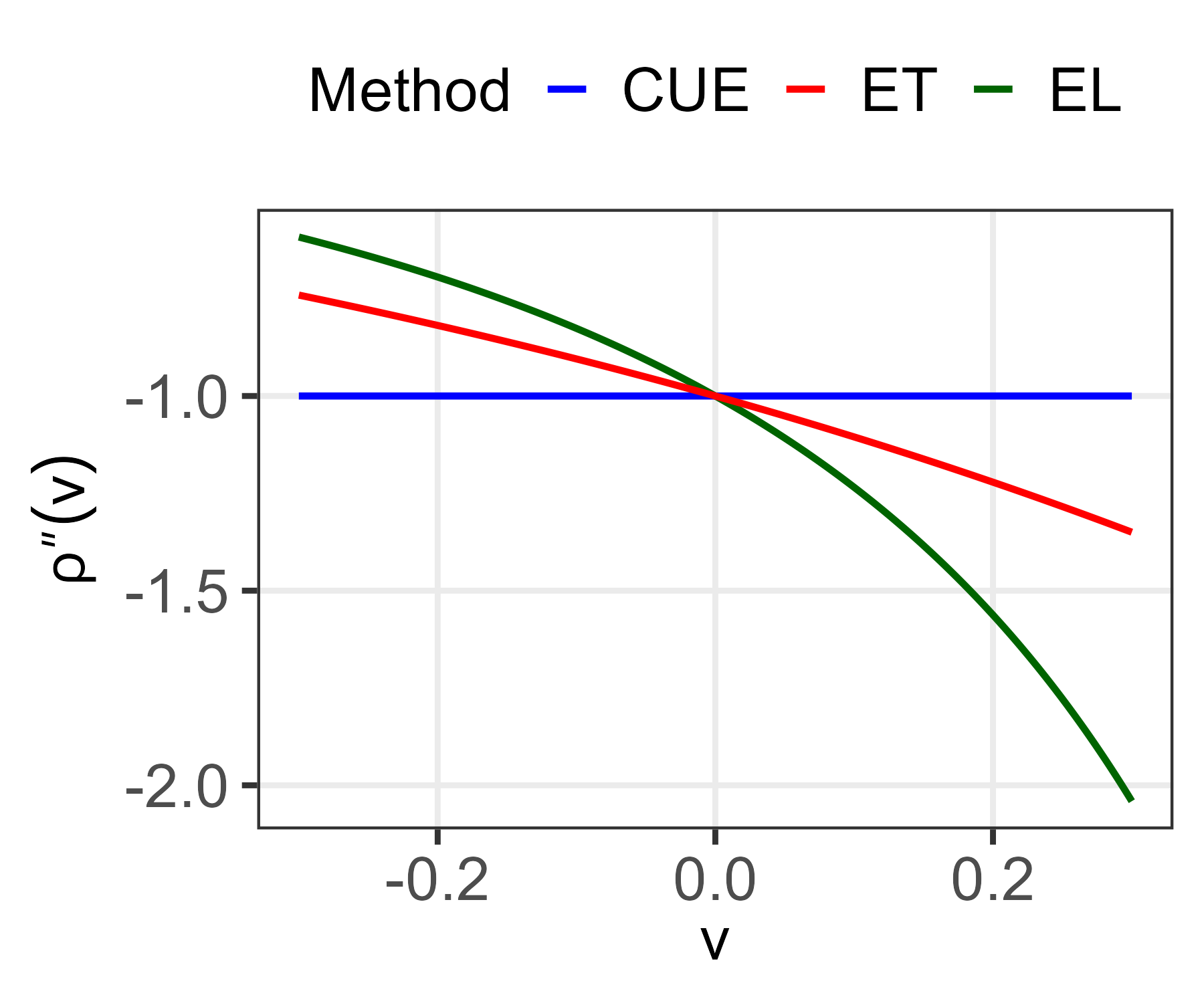}}
\hfill
\subfloat[$\rho(v)$ in a wider range]{%
\includegraphics[width=0.4\textwidth]{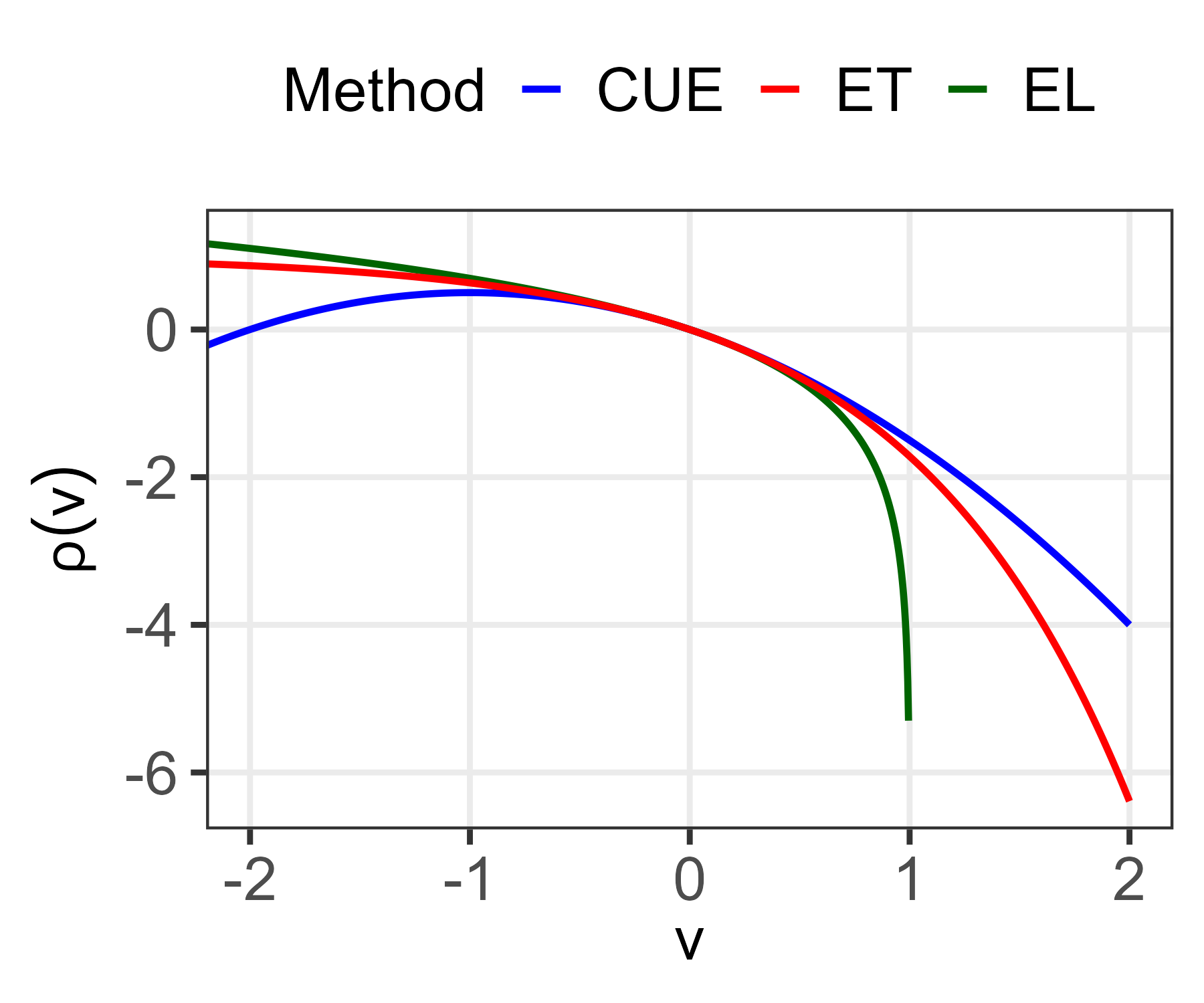}}
\end{minipage}
}
\caption{Plots of three different choices of $\rho(v)$ and their derivatives.\label{fig:rho_plots}}

\end{figure}

\subsection{Neyman orthogonal and non-Neyman orthogonal nuisances}
\label{subsec:nuisance}

A nuisance function is an auxiliary component that is not of direct interest but may affect estimation of the causal parameter. Certain nuisance functions satisfy the \textit{Neyman orthogonality condition} \citep{neyman1959,neyman1979,Chernozhukov2018}, which ensures that small estimation errors in these components do not influence the first-order behavior of the estimator of the parameter of interest. Formally, the Neyman orthogonality condition requires that the G\^ateaux derivative of the moment function $\Psi(\beta,\boldeta;\O)$ with respect to the nuisance function $\boldeta$, evaluated at the true nuisance component $\boldeta_0$, vanishes:
\[
\left. \frac{\partial}{\partial t}\bbE[\Psi(\beta_0,\boldeta_0+t(\boldeta-\boldeta_0);\O)] \right|_{t=0}=0.
\]
Intuitively, this condition ensures that the moment function is locally insensitive to errors in estimating $\boldeta$.

In contrast, when the derivative does not vanish, the nuisance function is \textit{non-Neyman orthogonal}, and its estimation error propagates directly into the first-order asymptotics of the estimator of the target parameter. This distinction is crucial in our setting: conditional expectations such as $\bbE(\Z|\X)$ or $\bbE(A|\Z,\X)$ are orthogonal nuisances, whereas the conditional censoring distribution $G(\cdot|\O_A)$ is non-Neyman orthogonal. As we will demonstrate, the latter introduces an additional variance component into the asymptotic distribution of the estimator of the target parameter, which must be explicitly accounted for to ensure valid inference. Rigorous definitions using G\^ateaux and Fr\'echet derivatives, together with a detailed proof establishing the non-orthogonality of $G(\cdot|\O_A)$, are provided in Section~\ref{suppproofproposition1} of the Supplementary Material.

\subsection{Nuisance function estimation}
\label{subsec:nuisance}
To evaluate the GEL criterion, all nuisance functions in the proposed AIPCW moment function $\psi$ must be estimated. We classify these nuisance functions into two categories. The first consists of \emph{Neyman orthogonal} nuisances, including $\bbE(\Z|\X)$, $\bbE(A|\Z,\X)$, $\bbE(Y|\Z,\X)$, $\bbE(R_A^2|\X)$, $\bbE(R_A R_Y|\X)$, and $\xi_0(\beta_0;\O_A)$. For these components, orthogonality ensures that small estimation errors do not affect the first-order behavior of the target estimator. This property allows us to employ flexible machine learning methods without altering the asymptotic distribution of $\hat\beta$ for the target causal parameter. In particular, we use DNNs, which, by the universal approximation theorem \citep{hornik1989}, can approximate complex nonlinear functions and thereby reduce bias from model misspecification.

The second category consists of the \emph{non-Neyman orthogonal} nuisance, namely the conditional censoring distribution $G(\cdot|\O_A)$. Because its first-order derivative does not vanish, estimation error in this component enters directly into the asymptotic distribution of $\hat\beta$. The following proposition formalizes this distinction, with proof provided in Section~\ref{pflemmaneyman} of the Supplementary Material.

\begin{proposition}\label{lemmaneyman}
For the AIPCW moment function \eqref{psi}, the set of nuisances ${\bbE(\Z|\X), \bbE(A|\Z,\X)}$, ${\bbE(Y|\Z,\X), \bbE(R_A^2|\X), \bbE(R_A R_Y|\X), \xi_0(\beta_0;\O_A)}$ are Neyman orthogonal, whereas $G(\cdot|\O_A)$ is non-Neyman orthogonal.
\end{proposition}

We provide some intuition for why $G(\cdot| \O_A)$ is non-Neyman orthogonal. Its first-order G\^ateaux derivative is 
\begin{equation}\label{gateauxG}
\bbE\Big[
\underbrace{
\frac{ G(T | \O_A) - G_0(T | \O_A)}{G_0(T | \O_A)}
}_{\text{Depends on } T \text{ and } \O_A}
\underbrace{
\left\{ g(\beta_0;\O_T) - \xi_0(\beta_0;\O_A) \right\}
}_{\text{Also depends on } T \text{ and } \O_A}
\Big] \neq 0.
\end{equation} 
Hypothetically, if $G(T|\O_A)$ were to depend only on $\O_A$ (and not on $T$), then by iterated expectation the entire expression would vanish, since $\bbE\{g(\beta_0;\O_T)-\xi_0(\beta_0;\O_A)\mid \O_A\}=0$. However, because $G(T|\O_A)$ is a function of both $T$ and $\O_A$, this cancellation does not occur, and the AIPCW moment function $\psi$ is not orthogonal with respect to $G(T|\O_A)$. Intuitively, the censoring weights depend jointly on the event time $T$ and $\O_A$, so small estimation errors in the nuisance function $G$ are transmitted directly into the AIPCW moment equation.

The computational implementation proceeds in three steps:
(i) estimate the non-Neyman orthogonal nuisance using a local Kaplan-Meier estimator \citep{dabrowska1989}; 
(ii) estimate the Neyman orthogonal nuisances using flexible DNNs \citep{Goodfellow2016}; and
(iii) plug the estimated nuisances into the GEL 2.0 criterion to obtain the causal effect estimate $\widehat\beta$.\\
\textit{(i) Estimating non-Neyman orthogonal nuisance via local Kaplan-Meier estimator.}
%Because the conditional censoring distribution $G(T| \O_A)$ is non-orthogonal, its estimation error enters the leading term of the asymptotic expansion. 
Following prior literature \citep{dabrowska1989,he2013,tang2020}, we estimate $G(\cdot| \O_A)$ using a local Kaplan-Meier estimator:
$$
\hat G (y|\O_A)=\prod_{i=1}^n\left\{1-\frac{B_{n i}(\O_A)}{\sum_{j=1}^n I\left(Y_j \geq Y_i\right) B_{n j}(\O_A)}\right\}^{I\left(Y_i \leq y, \delta_i=0\right)},
$$
where $B_{n j}(\O_A)=K\left(\frac{\O_A-\O_{A,j}}{h}\right) /\left\{\frac{1}{n} \sum_{i=1}^n K\left(\frac{\O_A-\O_{A,i}}{h}\right)\right\}, j=1, \ldots, n $, is the Nadaraya-Watson
weights, with $h$ denoting the bandwidth and $K$ a higher-order kernel function \citep{fan1992}.  One may use the Cox model \citep{cox1972regression}, parametric survival models \citep{cox2007parametric}, or other suitable approaches to estimate $G(\cdot|\O_A)$. The choice is flexible, provided the resulting estimator satisfies the conditions stated in Assumption~\ref{condkernel} in Section~\ref{additionalcond} of the Supplementary Material.

\textit{(ii) Estimating Neyman orthogonal nuisances via DNNs.}
To formalize the approximation properties of DNNs, we first introduce the H\"older class of functions \citep{schmidt2020}. Let $\zeta=s+r>0$ be the smoothness parameter, $r\in(0,1]$ and $s=\lfloor\zeta\rfloor\in\mathbb{N}_{0}$ be the fractional and integer parts of $\zeta$, respectively, where $\lfloor\zeta\rfloor$ denotes the largest integer not larger than $\zeta$, $d$ is the input dimension, and $\mathbb{N}_0$ denotes the set of non-negative integers. For a positive constant $B_0>0$, the H\"older class of functions $\mathcal{H}^\zeta([0,1]^d,B_0)$ is defined as
$$\begin{aligned}&\mathcal{H}^\zeta([0,1]^d,B_0)\\&=\Big\{f:[0,1]^d\to\mathbb{R},\max_{\|\balpha\|_1\leq s}\|\partial^{\balpha} f\|_\infty\leq B_0,\max_{\|\alpha\|_1=s}\sup_{x\neq y}\frac{|\partial^{\balpha} f(x)-\partial^{\balpha} f(y)|}{\|x-y\|^r}\leq B_0\Big\},\end{aligned}$$
where $\partial^{\balpha}=\partial^{\alpha_1}\cdots\partial^{\alpha_d}$ with $\balpha=(\alpha_1,\ldots,\alpha_d)^\top\in\mathbb{N}_0^d$ and $\|\balpha\|_1=\sum_{i=1}^d\alpha_i.$
Then we impose the following assumption on the nuisance functions. 

\begin{assumption}[Neyman orthogonal nuisance functions]\label{condnuispara}
Suppose that $\Z, \X$ are bounded, $\X \in \mathbb{R}^{d_x}$ with $d_x<\infty$, $\Z \in \mathbb{R}^{m}$,
\bse
\bbE\left(Z_j | \X=\x\right)&=&f_j(\x), j=1, \ldots, m, \\
\bbE(A | \Z=\z, \X=\x)&=&h_1(\z,\x), \\
\bbE(T | \Z=\z, \X=\x)&=&h_2(\z,\x), \\
\bbE\left(R_A R_T | \X =\x\right)&=&h_3(\x), \\
\bbE\left(R_A^2 | \X =\x\right)&=&h_4(\x).
\ese
Let $\boldeta_0=\{f_1,\dots,f_m,h_1, h_2, h_3, h_4\}$ collects $m+4$ nuisance functions, where $f_1,\dots,f_m, h_3,h_4$ all lie in a H\"older class $\mathcal{H}^\zeta([0,1]^{d_x},B_0)$, and $h_1, h_2$ lie in a H\"older class $\mathcal{H}^\zeta([0,1]^{m+d_x},B_0)$.
\end{assumption}
In Assumption \ref{condnuispara}, each conditional expectation, such as $\bbE(Z_j|\mathbf{X})$ and $\bbE(A|\mathbf{Z},\mathbf{X})$, is modeled as an unknown but sufficiently smooth function. Specifically, we assume that these functions lie in a H\"older class with bounded norm $B_0$, which imposes smoothness constraints and guarantees that DNNs can approximate them well \citep{jiao2023}. Importantly, the number of nuisance functions increases with sample size, accommodating a diverging number of candidate instruments. We estimate each nuisance function by minimizing the empirical risk over the rectified linear unit (ReLU) activated deep neural networks \citep{schmidt2020}
\be\label{reluclass}
\calF_{B_0,\calW,\calD,\calS,d}=\{f:\bbR^{d}\to\bbR \text{ with width } \calW, \text{ depth }\calD, \text{ size } \calS, \text{ and } \|f\|_{\infty}\leq B_0\}.
\ee
That is, $\hat f_j = \arg\min_{f\in\calF_{B_0,\calW,\calD,\calS,d_x}} 1/n\sum_{i=1}^n(Z_{ij}-f(\X_i))^2$ for $j=1,\dots,m$, and $h_1$ to $h_4$ can be estimated in the same way. The resulting estimators are denoted by $\hat\boldeta=\{\hat f_1, \dots,\hat f_m, \hat h_1, \hat h_2, \hat h_3, \hat h_4\}$. In addition, $\xi_0(\beta;\O_A)$ can be estimated by 
$$
\hat \xi(\beta;\O_A)=\frac{1}{n}\sum_{i=1}^n \frac{B_{n i}(\O_{A}) \delta_i}{\hat G \left(Y_i|\O_{A,i}\right)} g(\beta,\hat\boldeta;\O_i),
$$ 
where $g(\beta,\hat\boldeta;\O)$ replaces all the nuisance functions in the original $g(\beta;\O)$ with estimated values.

\textit{(iii) Plug-in of estimated nuisance functions into the GEL 2.0 criterion.}
After estimating all nuisance functions, we substitute the estimates into the GEL objective.
The corresponding plug-in moment function for \eqref{psi} is given by
$$
\psi(\beta,\hat G ,\hat\boldeta;\O)=\frac{\delta}{\hat G (Y|\O_A)}g(\beta,\hat\boldeta;\O)+(1-\frac{\delta}{\hat G (Y|\O_A)})\hat \xi(\beta;\O_A).
$$
Then the causal effect  is estimated by 
\be \label{betahat}
\hat\beta_{\textrm{GEL2}}^{\textrm{C}}= \arg\min_{\beta\in\calB}\hat Q(\beta,\hat G ,\hat\boldeta),
\ee
where 
$$
\hat Q(\beta,\hat G ,\hat\boldeta)=\sup_{\blambda\in L(\beta,\hat G ,\hat\boldeta)}\frac{1}{n}\sum_{i=1}^n \rho(\lambda(\beta,\hat G ,\hat\boldeta)\tp\psi(\beta,\hat G ,\hat\boldeta;\O_i)),
$$
and $L(\beta,\hat G ,\hat\boldeta) = \{\blambda:\blambda\tp\psi(\beta,\hat G ,\hat\boldeta;\O_i) \in \calV\}$. The GEL 2.0 formulation above generalizes the classical GEL criterion to settings where some nuisance functions may violate Neyman orthogonality and the moment condition is weak.

Note that when there is no censoring, that is, $\delta \equiv 1$, the censoring distribution satisfies $G_0(Y | \O_A) \equiv 1$, and the AIPCW moment function $\psi$ reduces to  the influence function $g(\beta;\O_T)$ in Lemma \ref{lemma0}. Even in the absence of censoring, our approach is more robust than previous methods \citep{ye2024,wang2025}, as it employs a more general GEL formulation and estimates nuisance functions nonparametrically with flexible  DNNs, without relying on cross-fitting. This added flexibility and computational efficiency enhance both the robustness and the practical utility of the resulting estimator for the target causal effect.

\section{Asymptotic Properties}
\label{sec:theoretical}

\subsection{Consistency and asymptotic normality}

In this subsection, we establish the consistency and asymptotic normality of the proposed estimator $\widehat{\beta}^{\textrm{C}}_{\textrm{GEL2}}$ in the presence of many weak and non-Neyman orthogonal moment conditions. We first formalize the many-weak-moment asymptotic regime, then set out the regularity assumptions on the nuisance functions and the moment structure, and finally present the main convergence result.

For clarity of exposition, we begin by introducing the notation that will be used in the theoretical
statements that follow. Let $\O_T^{'}$ be an independent copy of $\O_T$. Also let
\bse
&&\quad \bar \psi(\beta,G,\boldeta)=\frac{1}{n} \sum_{i=1}^n \psi(\beta,G,\boldeta;\O_i),\quad \bar\Sigma(\beta,G,\boldeta)=\frac{1}{n} \sum_{i=1}^n \psi(\beta,G,\boldeta;\O_i)\psi(\beta,G,\boldeta;\O_i)\tp,\\
&& \Sigma(\beta,G,\boldeta) = \bbE(\psi(\beta,G,\boldeta;\O)\psi(\beta,G,\boldeta;\O)\tp),\quad \Sigma_0=\Sigma(\beta_0,G_0,\boldeta_0),\\
&& \psi^{'}(G,\boldeta;\O_i)=\frac{\partial \psi(\beta,G,\boldeta;\O_i)}{\partial \beta},\quad  \bar \psi^{'}(G,\boldeta)=\frac{1}{n} \sum_{i=1}^n \psi^{'}(G,\boldeta;\O_i), \quad \psi^{*}=\bbE(\psi^{'}(G_0,\boldeta_0;\O)),\\
&&\phi(\O_T^{'},\O_T) =B_{nj}(\O_A)G_0(T|\O_A)\left\{\frac{I(T^{'}\leq T,\delta^{'}=0)}{P(T^{'}<T|\O_A)}-\int_0^{\min\{T^{'},T\}}\frac{d P(T^{'}\leq s, \delta^{'}=0)}{P(T^{'}<s|\O_A)^2}\right\}.
\ese

We now state the regularity conditions needed to establish the asymptotic properties of $\widehat{\beta}^{\textrm{C}}_{\textrm{GEL2}}$.

\begin{assumption}[Many weak moment asymptotics]\label{condweak}
(i). There exist scalars $\mu_n, c, c^{\prime}>0$, such that
$$
\mu_n^2 c \leq n \psi^{*\top}\Sigma_0^{-1} \psi^{*} \leq \mu_n^2 c^{\prime} ,
$$
and $\mu_n$ satisfies $\mu_n \rightarrow \infty$ as $n \rightarrow \infty$ and $m / \mu_n^2$ is bounded. 

(ii). Given the $L^2$ norm of the local Kaplan-Meier estimator is $$\kappa\triangleq\bbE\|\hat G(Y|\O_A)-G_0(Y|\O_A)\|_{L^2},$$ and we require $\sqrt{nm}\kappa/\mu_n=o(1)$.  
\end{assumption}

Assumption \ref{condweak} (i) imposes bounds on the \textit{concentration parameter} $n \psi^{*\top}\Sigma_0^{-1} \psi^{*}$ \citep{Staiger1997},
a measure of the overall moment strength. When $\mu_n = \sqrt{n}$, the assumption reduces to the classical setting of strong moments. Our framework allows $\mu_n = o(\sqrt{n})$, so that the concentration parameter may diverge more slowly, thus covering the many weak moment regimes. Moreover, $\mu_n$ directly determines the convergence rate of the estimator $\widehat\beta_{\textrm{GEL2}}^{\textrm{C}}$, and all asymptotic results are established without requiring $\mu_n$ to be of the classical $\sqrt{n}$ order.
Assumption \ref{condweak} (ii) imposes a requirement on the moment strength $\mu_n$ related to the convergence rate of the estimated censoring distribution. In particular, we require that $\mu_n$ is not too weak, so that $\sqrt{nm}\,\kappa/\mu_n = o(1)$. This condition arises because the AIPCW moment function $\psi$ is not Neyman orthogonal with respect to the conditional censoring distribution $G$. As a result, the estimation error in $G$ contributes to the leading term of the estimator, and therefore a minimal level of moment strength is needed to control its impact. This requirement reflects an additional cost of handling censored outcomes, which does not arise in the uncensored setting. Moreover, $\kappa$ is related to the bandwidth of the higher-order kernel estimator in the local Kaplan-Meier estimator. More details about the bandwidth $h$ can be found in Assumption \ref{condkernel} in the Supplementary Material.

Suppose that the true function $f_0 \in \mathcal{H}^{\zeta}([0,1]^d,B_0)$ and $\hat f$ is the empirical risk minimizer over the ReLU activated deep neural networks $\calF_{B_0,\calW,\calD,\calS,d}$ defined in \eqref{reluclass}. 
Let $\mathcal{E}(\zeta,B_0,\mathcal{W},\mathcal{D},\mathcal{S},d)$ denote the non-asymptotic 
error bound of $\hat f$ as defined in \cite{jiao2023}. The explicit form of 
$\mathcal{E}(\zeta,B_0,\mathcal{W},\mathcal{D},\mathcal{S},d)$ is given in 
Lemma~\ref{lemmadnn} of the Supplementary Material.

\begin{assumption}[Convergence rate of DNNs] \label{condDNNs}
(i). The empirical risk minimizers $\{\hat f_j(\X), \\\hat h_3(\X), \hat h_4(\X)\}$ over the ReLU activated deep neural networks $\mathcal \calF_{B_0, \calW, \calD, \calS, d_x}$ satisfy $$\mathcal{E}(\zeta,B_0, \calW, \calD, \calS, d_x)=o(n^{-1/4}m^{-1/4}).$$

(ii). The empirical risk minimizers $\{\hat h_1(\Z,\X),\hat h_2(\Z,\X)\}$ over the ReLU activated deep neural networks $\mathcal \calF_{B_0, \calW, \calD, \calS, d_x+m}$ satisfy $$\mathcal{E}(\zeta, B_0, \calW, \calD, \calS, d_x+m)=o(n^{-1/4}m^{-1/4}).$$
\end{assumption}

Assumption~\ref{condDNNs} imposes constraints on the width and depth of the DNNs so that the
non-asymptotic error bound converges at the rate $o(n^{-1/4}m^{-1/4})$. This rate is 
precisely what is required to establish the asymptotic normality of the estimator for 
the causal effect parameter $\beta$.
The non-asymptotic bound also plays a key role in the proof of the localization-based 
empirical process argument \citep{bartlett2005,farrell2021}, as demonstrated in 
Lemma~\ref{lemmag} of the Supplementary Material.
In Assumption \ref{condDNNs} (ii), we allow the input dimension $d_x+m$ to diverge, and the DNNs can still achieve the desirable convergence rate when the depth, width and the number of candidate IVs meet certain requirements. In Lemma \ref{lemmadnn} of the Supplementary Material, we provide an example in which Assumption \ref{condDNNs} holds. If other machine learning estimator $\hat f$ (e.g., random forests \citep{breiman2001} or gradient boosting \citep{chen2016}) is used to estimate the nuisance function $f_0$, and the estimator satisfies $(\bbE_n-\bbE)(l(\hat f)-l(f_0))=o_p(n^{-1/2}m^{-1/2})$ for any  Lipschitz function $l$, a convergence rate of $\bbE\|\hat f-f_0\|_{L^2}=o(n^{-1/4}m^{-1/4})$ is also sufficient to ensure the asymptotic properties of $\widehat{\beta}$. In addition, the rate $o(n^{-1/4})$ is well established in the nonparametric estimation literature; see, for instance, Assumption 5.1 in \citet{Chernozhukov2018} and Assumption (b) in Theorem 3 of \citet{farrell2021}. We include an extra factor $m^{-1/4}$ because our framework allows the number of moment functions $m$ to diverge. When the number of moment conditions $m$ is fixed, the required rate in Assumption \ref{condDNNs} reduces to $o(n^{-1/4})$. 

We now formally establish the consistency and asymptotic normality of the proposed estimator $\hat{\beta}_{\textrm{GEL2}}^{\textrm{C}}$.

\begin{theorem}\label{thmcensor}
Suppose Assumptions \ref{condindependent}-\ref{condDNNs} and the regularity conditions stated in Assumptions \ref{condmoment}-\ref{condkernel} in Section \ref{additionalcond} of the Supplementary Material hold, then $\hat{\beta}_{\textrm{GEL2}}^{\textrm{C}}$ is consistent and asymptotical normal, that is, as $n \rightarrow \infty$, $m^3/n\to 0$, $\hat{\beta}_{\textrm{GEL2}}^{\textrm{C}} \xrightarrow{p} \beta_0$ and
$$
\frac{\mu_n\left(\hat{\beta}_{\textrm{GEL2}}^{\textrm{C}}-\beta_0\right)}{\sqrt{H^{-1}(H+V_1+V_2)H^{-1}}} \xrightarrow{d} N(0,1),
$$
where  
$V_1 = \mu_n^{-2} \bbE\left[U_i^{C\top} \Sigma_0^{-1} U^C_i\right]$,
$V_2 = \bbE_{P_{\O_T^{'}}}(\bbE_{P_{\O_T}}D_G\{n/\mu_n\partial^2 \rho(\blambda\tp\psi(\beta_0,G_0,\boldeta_0;\O_{T}))/\partial \beta\allowbreak\phi(\O_{T}^{'},\O_{T})\})^2 +2\bbE_{P_{\O_T^{'}}}(\bbE_{P_{\O_T}}\{D_Gn^2/\mu_n^2\partial\rho(\blambda\tp\psi(\beta_0,G_0,\boldeta_0;\O_{T}))/\partial\beta\phi(\O_{T}^{'},\O_{T})\}\{\psi^{*}\Sigma_0^{-1}$\\$\psi(\beta_0,G_0,{\boldeta_0};\O_{T})\})$, 
$H = n \psi^{*\top} \Sigma_0^{-1} \psi^{*}/\mu_n^2$,
and $U^C_i=\psi^{*}(G_0,\boldeta_0;\O_i)-\psi^{*}-\{\Sigma_0^{-1} \bbE(\psi(\beta_0,G_0,\boldeta_0;\O)\allowbreak \psi^{'}(G_0,\boldeta_0;\O)\tp)\}\tp \psi(\beta_0,G_0,\boldeta_0;\O_i)$ is the population residual of the least squares regression of $\psi^{'}(G_0,\boldeta_0;\O_i)-\psi^{*}$ on $\psi(\beta_0,G_0,\boldeta_0;\O_i)$. 

\end{theorem}
Theorem~\ref{thmcensor} shows that under regularity assumptions, the proposed estimator $\hat\beta_{\textrm{GEL2}}^{\textrm{C}}$ is consistent and asymptotically normal with convergence rate $\mu_n$. A direct implication is that, when the censoring rate is zero, $\delta \equiv 1$ and 
$G_0(Y|\O_A) \equiv 1$, the AIPCW moment function $\psi$ reduces exactly to the influence function $g(\beta;\O_T)$ in the uncensored setting in Lemma \ref{lemma0}. As a result,  the asymptotic variance term $V_2$ vanishes, as shown in Corollary \ref{coro1} in the Supplementary Material.

To facilitate understanding of Theorem~\ref{thmcensor}, we outline the main steps of the
proof. The argument begins with a Taylor expansion around $\beta_0$, which yields 
the following decomposition:

\begin{equation*}
0 = \frac{n}{\mu_n}\frac{\partial \hat Q(\hat\beta,\hat G ,\hat\boldeta)}{\partial \beta}
= 
\underbrace{
\frac{n}{\mu_n}\frac{\partial\hat Q(\beta_0, \hat G ,\hat{\boldeta})}{\partial\beta}
}_{\substack{\text{first-order term,}\\ \text{converge to $N(0,H+V_1+V_2)$}}}
+ 
\underbrace{
\frac{n}{\mu_n^2}\frac{\partial^2\hat Q(\bar\beta,\hat G ,\hat\boldeta)}{\partial \beta^2}
}_{\substack{\text{second-order term,}\\ \text{converge to $H$}}}\mu_n(\hat{\beta}-\beta_0),
\end{equation*}
where $\bar\beta$ lies between $\hat\beta$ and $\beta_0$. 
In addition, the first-order term admits the following  decomposition:
\be \label{firstorderterm}
&&\frac{n}{\mu_n}\frac{\partial\hat Q(\beta_0, \hat G ,\hat{\boldeta})}{\partial\beta}\n\\
&=& (\bbE_n - \bbE)\left[\frac{n}{\mu_n}\frac{\partial \rho(\blambda\tp\psi(\beta_0,\hat G,\hat\boldeta;\O))}{\partial \beta}\right]+\bbE\left[\frac{n}{\mu_n}\frac{\partial \rho(\blambda\tp\psi(\beta_0,\hat G,\hat\boldeta;\O))}{\partial \beta}\right]\n\\
&=& \underbrace{
(\bbE_n - \bbE) \left[\frac{n}{\mu_n}
\frac{\partial \rho(\blambda\tp\psi(\beta_0, \hat G, \hat\boldeta;\O))}{\partial \beta}
- 
\frac{n}{\mu_n}\frac{\partial \rho(\blambda\tp\psi(\beta_0, G_0, \boldeta_0;\O))}{\partial \beta}
\right]
}_{=o_p(1)\text{, controlled by localization-based empirical process argument}} \n\\
&& + \underbrace{
(\bbE_n - \bbE)\left[\frac{n}{\mu_n}
\frac{\partial \rho(\blambda\tp\psi(\beta_0, G_0, \boldeta_0;\O))}{\partial \beta}\right]
}_{\text{converge to a normal distribution}}\n\\
&&+\underbrace{\bbE\left[\frac{n}{\mu_n}\frac{\partial \rho(\blambda\tp\psi(\beta_0,G_0,\boldeta_0;\O))}{\partial \beta}\right]}_{\text{=0, true moment function}} + \underbrace{
\bbE
\left[\frac{n}{\mu_n}D_{\boldeta}\frac{\partial \rho(\blambda\tp\psi(\beta_0, G_0, \boldeta_0;\O))}{\partial \beta}
(\hat\boldeta - \boldeta_0)\right]
}_{\text{=0, due to Neyman orthogonal}}\n\\
&&+ \underbrace{
\bbE
\left[\frac{n}{\mu_n}D_{G}\frac{\partial \rho(\blambda\tp\psi(\beta_0, G_0, \boldeta_0;\O))}{\partial \beta}
(\hat G - G_0)\right]
}_{\text{non-Neyman orthogonal, converge to a normal distribution}} + \underbrace{
o_p(1)
}_{\text{remainder}}.
\ee

The first term in equation \eqref{firstorderterm} involves estimated nuisance functions and is shown to be asymptotically negligible. To control it, we localize the nuisance estimators in a shrinking neighborhood of the true nuisance function, where the associated function class has controlled the growth of entropy \citep{bartlett2005}. Under Lipschitz continuity, applying a maximal inequality \citep{chernozhukov2014} yields stochastic equicontinuity of the empirical process over the function class. As a result, the impact of nuisance estimation error on this term has no impact on first-order asymptotics. The second term is the empirical process at the true nuisance functions, and converges to a normal distribution and includes two variance components: (i) a classical GMM variance \citep{hansen1982}; and (ii) an additional component due to weak moments \citep{newey2009}. The third term vanishes because the population moment condition holds at the truth. The fourth term also vanishes due to Neyman orthogonality, which eliminates the first-order impact of estimation error in orthogonal nuisances. The fifth term arises from  estimating non-Neyman orthogonal nuisance and can be written in the form of a Kaplan--Meier integral \citep{stute1995,gerds2017}:
$$
\frac{1}{n}\sum_{i=1}^n \bbE_{P_{\O_T}} [D_G\frac{n}{\mu_n}\frac{\partial\rho(\blambda\tp\psi(\beta_0,G_0,\boldeta_0;\O_{T}))}{\partial\beta}\phi(\O^{'}_{T,i},\O_{T})]+o_p(1).
$$ 
This term is asymptotically normal. The detailed proof of the theorem can be found in Section \ref{pfthm2} of the Supplementary Material. It contributes an additional variance component $V_2$ to the GEL 2.0 estimator of $\beta$.

\subsection{Comparison with classical GEL theory }

Theorem~\ref{thmcensor} illustrates how GEL~2.0 extends the classical GEL framework. Unlike \citet{newey2009}, who considered GEL under many weak moment conditions without nuisance functions, our framework admits a diverging number of nuisance functions (both Neyman orthogonal and non-Neyman orthogonal) estimated via nonparametric methods. In contrast to \citet{ye2024}, who employed parametric estimation of nuisance functions under a CUE criterion with Neyman orthogonal moments, our results allow for nonparametric estimation and explicitly accommodate non-Neyman orthogonal nuisances, thereby substantially broadening the theoretical scope of GEL.

Theorem \ref{thmcensor} also distinguishes between Neyman orthogonal and non-Neyman orthogonal nuisance functions. For orthogonal nuisances, the estimation error has no first-order impact on the estimator of $\beta$, so a valid inference for $\beta$ is achievable once the nuisance estimators converge at the rate $o(n^{-1/4} m^{-1/4})$ when the number of moments diverges. We also provide an example showing that DNNs can achieve this rate under suitable architectural choices in Lemma \ref{lemmadnn} of the Supplementary Material, illustrating the applicability of the theory to modern DNNs estimators. Compared with the common use of sample splitting to handle orthogonal nuisances \citep{Chernozhukov2018,wang2025}, our approach achieves the same convergence rate without sample splitting, thereby simplifying computation. \citet{Chernozhukov2018} study the case with a fixed number of moment conditions, where sample splitting is computationally feasible. In contrast, our setting involves a diverging number of moment conditions involving nuisance functions, for which sample splitting becomes computationally prohibitive, as it would substantially increase the computational burden and reduce the effective sample size available for estimating each nuisance function. To address this, we instead employ a localization-based empirical process technique that controls overfitting bias without requiring sample splitting.

For non-Neyman orthogonal nuisances, the plug-in error contributes directly to the first-order expansion of the estimator for the causal parameter because the pathwise derivative with respect to the nuisance is nonzero.
Consequently, the estimation of these nuisances contributes an additional first-order 
term to the asymptotic distribution of $\hat{\beta}_{\textrm{GEL2}}^{\textrm{C}}$. To retain the asymptotic normality of $\hat{\beta}_{\textrm{GEL2}}^{\textrm{C}}$, it is unnecessary
that the nuisance estimator itself is $\sqrt{n}$-normal; rather, it suffices
that the estimator admits an
asymptotically linear expansion with bias $o_p(n^{-1/2})$. In our framework, the 
local Kaplan-Meier estimator, with higher-order kernels and carefully chosen 
bandwidths, satisfies these requirements: the plug-in bias is negligible, but 
the associated variance component persists and appears as the $V_2$ term in the 
asymptotic variance of $\hat{\beta}_{\textrm{GEL2}}^{\textrm{C}}$. To our knowledge, this is the first extension of semiparametric GEL theory that
explicitly incorporates an additional variance component arising
from estimating non-Neyman orthogonal nuisances.
A summary comparing Neyman orthogonal and non-Neyman orthogonal nuisances is provided in Table~\ref{tab:orthovsnonorth}.

\begin{table}
\caption{\label{tab:orthovsnonorth}Comparison of Neyman orthogonal and non-Neyman orthogonal nuisances}
\begin{tabular}{lll}
\hline
\quad & Neyman orthogonal  & Non-Neyman orthogonal \\
\hline
Definition 
& pathwise derivative $=0$ 
& pathwise derivative $\neq 0$ \\
Rate  on nuisance estimator
& $o(n^{-1/4})$ sufficient 
& $\sqrt{n}$-consistency of functionals  \\
First-order bias 
& eliminated by orthogonality 
& may appear but can be corrected \\
Asymptotic variance of $\hat\beta$
& unaffected 
& additional variance component \\
\hline
\end{tabular}
\end{table}

\section{Model Diagnostics}
\label{sec:diagnostics}
\subsection{Measure of weak identification}
In this section, we apply the weak identification test \citep{ye2024} for the proposed \texttt{MAWII-Surv} framework. The heteroscedasticity-robust $F$-statistic can be applied here because we ultimately need to know whether the conditional variance of the exposure $A$ depends on the instruments $\Z$ after controlling for baseline covariates $\X$.  Concretely, let
\[
\widetilde R_A^{2}=R_A^{2}-{\bbE}(R_A^{2}|\X),
\qquad
\widetilde\Z=\Z-{\bbE}(\Z|\X).
\]
We run the no-intercept regression:
\begin{equation}\label{weakid}
\widetilde R_A^{2}=\widetilde\Z\gamma+\epsilon,
\end{equation}
and form the heteroscedasticity-robust $F$-statistic  \citep{koenker1981}, denoted by $F_{\text{MAWII}}$. As suggested by \citet{ye2024}, we recommend checking that  $F_{\text{MAWII}}>2$, which indicates detectable heteroscedasticity of $A$ in $\Z$ and provides the variation required for more reliable estimation.

\subsection{Over-identification test with censoring}
Testing the validity of over-identifying restrictions is an important step in MR with 
multiple candidate IVs, as it provides a diagnostic check of model specification and 
moment condition validity. In the GMM literature, this is typically carried out by testing the 
null hypothesis $H_0 : \bbE[\psi(\beta_0, G_0, \boldeta_0)] = 0$.
With uncensored data, a widely used test statistic is $2n\widehat{Q}(\widehat{\beta}_{\textrm{GEL2}}^{\textrm{C}}, \widehat{\boldeta})$, which asymptotically follows a normal distribution under $H_0$ \citep{newey2009,ye2024}. This test is closely related to Sargan's $J$-test \citep{sargan1958,hansen1982}, which was introduced as a way to assess the validity of over-identifying restrictions.

However, this classical approach is not directly applicable in our setting due to
the presence of non-Neyman orthogonal nuisance $G(Y|\O_A)$. Unlike the uncensored setting, 
the presence of this nuisance alters the asymptotic distribution of 
$2n\widehat{Q}(\widehat{\beta}_{\textrm{GEL2}}^{\textrm{C}}, \hat G , \widehat{\boldeta})$. 
To address this issue, we develop a new testing procedure that explicitly accounts for 
the impact of estimating the non-Neyman orthogonal nuisance $G(Y | \O_A)$.
% To obtain the asymptotic distribution of the test statistic, we need another condition: 
% \begin{assumption}\label{condtest}
% Given that the convergence rate of $\widehat G - G$ is $\kappa$, and that of the DNNs estimators is $\kappa_{\mathrm{DNNs}}$, we require
% $(\kappa^2+\kappa_{\textrm{DNNs}}^2)\mu_n^2/\sqrt{m}=o_p(1).$
% \end{assumption}
% This condition further constrains the relationship between $\mu_n$ and the convergence rate of the nuisance estimator. It is introduced to control the remainder term in the Taylor expansion of the objective function $Q$. Even if this condition does not hold, the over-identification test is still feasible. However, the asymptotic distribution of the test statistic may no longer be valid. In such case, we can apply to bootstrap methods to approximate the sampling distribution.
\begin{theorem}\label{thmtest}
Under Assumptions \ref{condindependent}-\ref{condDNNs} and Assumptions \ref{condmoment}-\ref{condkernel} in Section \ref{additionalcond} of the Supplementary Material, and under the null hypothesis $H_0: \bbE \psi(\beta_0,G_0,\boldeta_0)=0$, we have 
$$
\frac{2n\hat Q(\hat\beta_{\textrm{GEL2}}^{\textrm{C}},\hat G ,\hat\boldeta)-(m-1)}{\sqrt{2(m-1)(1+V_3)}}\to N(0,1),
$$
where $V_3 = \bbE_{P_{\O_T^{'}}}(\bbE_{P_{\O_T}} D_G 2n/\sqrt{2(m-1)}\rho(\blambda\tp\psi(\beta_0,G_0,\boldeta_0;\O_T))\phi(\O_{T}^{'},\O_{T}))^2$. 
\end{theorem}

Theorem~\ref{thmtest} establishes the asymptotic null distribution of $2n\hat Q(\hat\beta_{\textrm{GEL2}}^{\textrm{C}}, \hat G , \hat\boldeta)$, which includes a variance correction term $V_3$ to account for the additional uncertainty arising from the estimation of the non-Neyman orthogonal nuisance function $G(Y|\O_A)$. The inclusion of $V_3$ is a key contribution of our work, as it overcomes a fundamental limitation of standard over-identification tests, such as \citet{newey2009} and \citet{ye2024}, which are designed for fully observed data. Specifically, $V_3$ depends on the derivative of the GEL loss function with respect to $G(Y|\O_A)$ and captures the way in which the estimation of censoring weights interacts with the moment function. The magnitude of $V_3$ depends on the censoring rate: it increases as censoring becomes more severe and vanishes when censoring is absent, in which case the test statistic reduces to its classical form \citep{newey2009}. We then reject $H_0$ at level $\alpha$ if
\[
t_{\textrm{MAWII}}=\left| \frac{2n \hat Q(\hat\beta^{\textrm{C}}_{\textrm{GEL2}}, \hat G , \hat \boldeta) - (m-1)}{\sqrt{2(m-1)(1 + \hat V_3)}} \right| > z_{1 - \alpha/2},
\]
where $z_{1 - \alpha/2}$ is the $(1 - \alpha/2)$ quantile of the standard normal distribution, and $V_3$ is estimated by:
\[
\hat V_3 = \frac{1}{n}\sum_{j=1}^{n} \left(\frac{1}{n} \sum_{i=1}^{n}D_G\rho(\blambda(\hat\beta_{\textrm{GEL2}}^{\textrm{C}},\hat G,\hat \boldeta)^\top \psi(\hat\beta_{\textrm{GEL2}}^{\textrm{C}},\hat G,\hat \boldeta; \O_i)) \phi(\O_j, \O_i) \right)^2.
\]
%Simulation results in Section \ref{sec:simulation} confirm that our adjusted test maintains the desired significance level. This highlights the practical necessity of properly accounting for $V_3$ in censored settings.
%------------------------------------------------

% %------------------------------------------------

\section{Simulation Study}
\label{sec:simulation}
%------------------------------------------------
\subsection {Finite-sample performance }\label{sec:censorsimu}

The purpose of this section is to evaluate the finite-sample performance of the proposed  $\widehat{\beta}^{\textrm{C}}_{\textrm{GEL2}}$ estimator and to compare it with related methods in a wide range of  settings with right-censored outcomes. We assess bias, standard errors, and empirical 95\% confidence interval (CI) coverage across simulation settings that incorporate both weak and invalid instruments. 
%The study proceeds as follows: we first describe the data-generating process and explain the rationale for our design; we then outline the competing methods 
%and performance measures; and finally, we present results demonstrating that our proposed  approach is more robust and accurate than existing methods.

We generate data from an AFT model with right-censoring. Let $\X=(X_1,\dots, X_5)\tp$, where $X_j \sim \textrm{Unif}[-2,2]$ for $j=1,\dots,5$, $\epsilon_A \sim N(0,0.4(1-h^2))$, $\epsilon_T \sim N(0,0.4(1-h^2))$, $U \sim N(0,0.6(1-h^2))$, and $\upsilon_j \sim N(0,h^2/(1.5m))$ for $j=1,\dots,m$, where $h^2=0.2$ is the proportion of variance in $A$ that is attributed to $\bbE(A|\Z,\X)$. $\Z$ are generated as: 
$Z_{1} = \cos\left(\pi X_{1}\right) + N(0, 0.4)$, 
$Z_{2} = \left(X_{1} + 1\right)\left(X_{2} - 1\right) + N(0, 0.4)$,
$Z_{3} = X_{1} + X_{2} + N(0, 0.4)$, 
$Z_{4} = \left(X_{2} - 0.5\right)^2 + N(0, 0.4)$, and
$Z_{5} = \sin\left(X_{1} + X_{2}\right) + N(0, 0.4)$.
For $j=6,\dots,m$, $Z_j\sim \textrm{Unif}[-2,2]$. We specify exposure $A$ and outcome $T$ as:
$$
A = \alpha_1(\Z) + \alpha_2(\X) + U + (1 + \sum_{j=1}^{m} \upsilon_j Z_j )\epsilon_A,
$$
\begin{equation*}
T = \beta_0 A + \gamma_1(\Z) + \gamma_2(\X) - U + \epsilon_T.   
\end{equation*}
Here $\alpha_1,\alpha_2,\gamma_1$, and $\gamma_2$ are defined as
$
\alpha_1(\Z) = \sum_{j=1}^5 \xi_{A,j}\cos(\pi Z_j), \gamma_1(\Z) = \sum_{j=1}^5 \xi_{Y,j}\cos(\pi Z_j),
$
$
\alpha_2(\X)=2\sin(\sin(\sin(\sum_{j=1}^2 X_j + 1) + \sin(\sum_{j=3}^5 X_j + 1)) + \sin(\sum_{j=1}^5 X_j + 1)),
$
and
$
\gamma_2(\X) = \{\cos(X_1) + X_1X_2 + \sin( X_3+X_4+X_5 - 1)\}/2,
$ where $ \xi_{A,j} \sim N(0, 0.4(1 - h^2)) $ and $\xi_{Y,j}$ determines whether IVs are valid, with the specific form given in the next paragraph. Moreover, let $C\sim \textrm{Unif}[-3+\tau, 3+\tau]$, where $\tau$ is used to keep the censoring rate around 0.4 \citep{cai2007}. Finally, let $Y = \min(T, C)$ be the observed outcome and $\delta = I(T\leq C)$ be the censoring indicator. 

Figure \ref{figdgp} illustrates three types of candidate IVs, $\mathcal{S}_1,\mathcal{S}_2,\mathcal{S}_3$, representing: (i) valid IVs, $\xi_{Y,j}=0$; (ii) invalid IVs under the instrument strength independent of direct effect (InSIDE) assumption \citep{Bowden2015}, $\xi_{Y,j} \sim N(0,0.4(1-h^2))$; (iii)  invalid IVs with correlated pleiotropy, $\xi_{Y,j} = \xi_{A,j}/2$, respectively. We consider the following four settings: 
% We consider three different choices for $\xi_{Y,i}$, corresponding to different cases of weak or invalid instrumental variables, as illustrated in Figure \ref{figdgp}. We classify these candidate IVs into three types: $\mathcal{S}_1$, $\mathcal{S}_2$, and $\mathcal{S}_3$, with proportions $p_1$, $p_2$, $p_3$ such that $p_1 + p_2 + p_3 = 1$. Set $\mathcal{S}_1$ consists of valid IVs with no direct effect on $Y$, thus $\xi_{Y,i}=0$. Set $\mathcal{S}_2$ comprises invalid IVs whose pleiotropic effects are uncorrelated with their effects on $A$, ensuring the instrument strength independent of direct effect (InSIDE, \cite{Bowden2015}) assumption holds. For these IVs, we have $\xi_{Y,i} \sim_{\text{i.i.d.}} N(0,0.4(1-h^2))$. Set $\mathcal{S}_3$ contains invalid IVs that affect both $A$ and $Y$ through a common factor $C$, resulting in correlated pleiotropic effects, and set $\xi_{Y,i} = \xi_{A,i}/2$.  Here are four scenarios of weak and invalid IVs:

\begin{itemize}
\item Case 1. (No invalid IVs): $p_1 = 1$, $p_2 = p_3 = 0$;
\item Case 2. (40\% invalid IVs): $p_1 = 0.6$, $p_2 = 0.2$, $p_3 = 0.2$;
\item Case 3. (90\% invalid IVs with InSIDE): $p_1 = 0.1$, $p_2 = 0.9$, $p_3 = 0$;
\item Case 4. (90\% invalid IVs without InSIDE): $p_1 = 0.1$, $p_2 = 0$, $p_3 = 0.9$.
\end{itemize}

\begin{figure}[pbht]
\center
\includegraphics[width=0.35\linewidth]{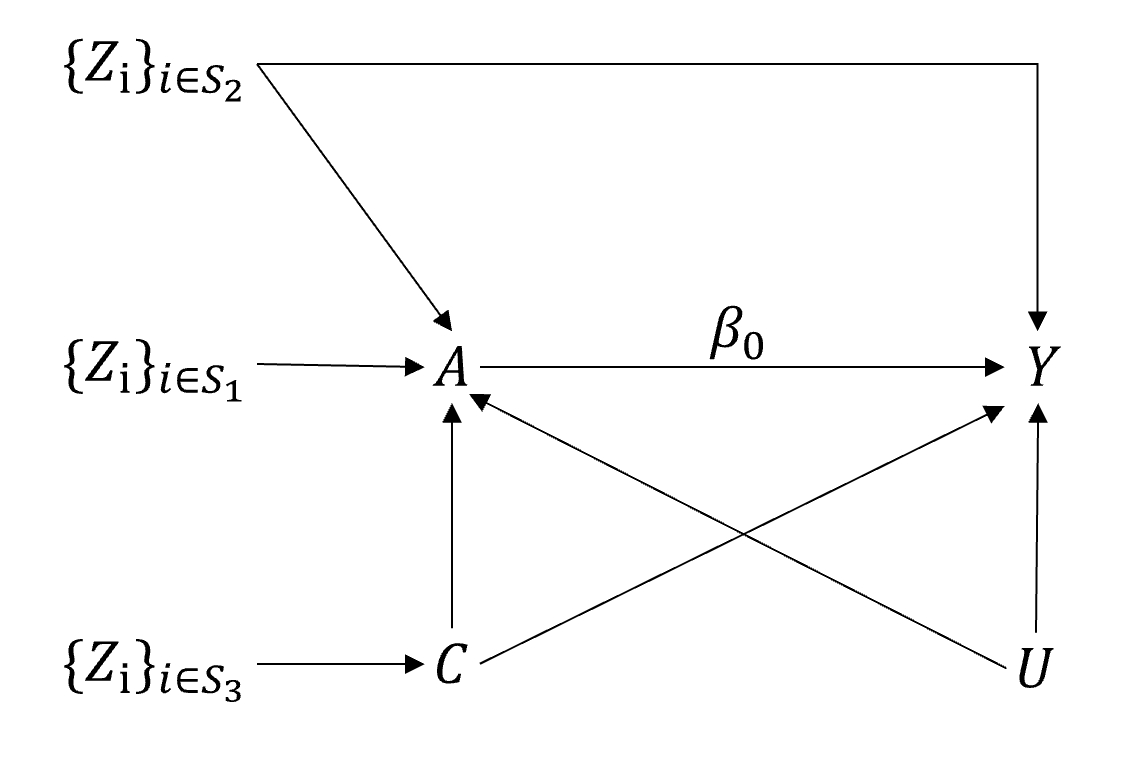}
\caption{Illustration of three types of candidate IVs: $\mathcal{S}_1$ for valid IV, $\mathcal{S}_2$ for invalid IV with InSIDE, and $\mathcal{S}_3$ for invalid IV without InSIDE.}
\label{figdgp}
\end{figure}

We compare three GEL estimators: exponential tilting (ET), empirical likelihood (EL), and continuous updating (CUE); we also used multiple machine learning algorithms to estimate Neyman-orthogonal nuisances: deep neural networks (DNNs), linear regression (LR), random forest (RF), and extreme gradient boosting (XGB). Details on architectures and hyperparameter selection for DNNs, RFs, and XGB are provided in Section \ref{suppsec:censorsimu} of the Supplementary Material. As a benchmark, we also report results from a classical AFT model that cannot control unmeasured confounding. Performance is assessed in terms of the average percentage bias (BIAS), the Monte Carlo standard deviation (SD), the average estimated standard error (SE), and the empirical coverage probability (CP) in 200 replications.

Table~\ref{censor-simu-nonlinear} reports the simulation results with non-linear nuisance functions. Across all four settings, our DNN-based GEL estimators (particularly with ET or EL) consistently exhibit lower bias and more accurate CI coverage than the other competing methods. In contrast, LR, RF, and XGB suffer from non-negligible bias, because they cannot capture complex nonlinear relationships well. Across the three GEL criteria with DNN-based orthogonal nuisance estimation, CUE shows slightly greater bias than EL and ET, in line with theory \citep{Imbens2002,newey2004}. The classical AFT model  performs poorly across all cases due to unmeasured confounding as expected. In addition, the $F_{\textrm{MAWII}}$ statistic for heteroscedasticity is 5.74, exceeding the recommended threshold of 2. 

Overall, our simulation results demonstrate the good performance of the proposed \texttt{MAWII-Surv} framework, particularly with EL and ET in the GEL family. Additional simulation settings and results in uncensored settings show that our proposed estimator outperforms previous methods as in \citet{ye2024}, especially with complex non-linear nuisance functions. More results are provided in Section~\ref{suppsec:uncensorsimu} of the Supplementary Material.
\begin{table} 
\caption{\label{censor-simu-nonlinear}Simulation results for censored outcomes with nonlinear nuisance functions, based on 200 replications with $n=10{,}000$ and $m=20$. BIAS denotes the percentage bias between $\hat\beta$ and $\beta_0$, SD the Monte Carlo standard deviation of $\hat\beta$, SE the average estimated standard error, and CP the $95\%$ coverage probability. AFT stands for classical AFT model.}
\begin{tabular}{lrrrrlrrrr}
\hline
\multicolumn{5}{c}{Case 1 (No invalid IVs)}                 & \multicolumn{5}{c}{Case 2 (40\% invalid   IVs)}              \\ \hline
Model       & BIAS         & SD      & SE      & CP       & Model        & BIAS          & SD       & SE     & CP      \\ \hline
DNN\_ET     & 0.875\%      & 0.253    & 0.253    & 0.945    & DNN\_ET      & -0.988\%      & 0.278     & 0.275   & 0.940   \\
DNN\_EL     & -1.049\%     & 0.257    & 0.303    & 0.945    & DNN\_EL      & -1.731\%      & 0.281     & 0.282   & 0.940   \\
DNN\_CUE    & 3.968\%      & 0.255    & 0.268    & 0.945    & DNN\_CUE     & 1.960\%       & 0.283     & 0.291   & 0.940   \\
LR\_ET      & 36.894\%     & 0.149    & 0.209    & 0.995    & LR\_ET       & 21.936\%      & 0.163     & 0.218   & 0.995   \\
LR\_EL      & 36.449\%     & 0.151    & 0.293    & 1.000    & LR\_EL       & 22.827\%      & 0.164     & 0.235   & 0.995   \\
LR\_CUE     & 38.225\%     & 0.149    & 0.221    & 0.995    & LR\_CUE      & 22.759\%      & 0.170     & 0.229   & 0.995   \\
RF\_ET      & 0.084\%      & 0.183    & 0.263    & 0.995    & RF\_ET       & -18.796\%     & 0.201     & 0.272   & 0.980   \\
RF\_EL      & -1.523\%     & 0.186    & 0.322    & 0.995    & RF\_EL       & -19.457\%     & 0.200     & 0.290   & 0.980   \\
RF\_CUE     & 2.764\%      & 0.186    & 0.272    & 0.990    & RF\_CUE      & -16.603\%     & 0.207     & 0.281   & 0.985   \\
XGB\_ET     & 13.859\%     & 0.173    & 0.246    & 1.000    & XGB\_ET      & -0.613\%      & 0.187     & 0.255   & 0.990   \\
XGB\_EL     & 13.271\%     & 0.176    & 0.317    & 1.000    & XGB\_EL      & -0.755\%      & 0.188     & 0.277   & 0.995   \\
XGB\_CUE    & 16.266\%     & 0.176    & 0.256    & 1.000    & XGB\_CUE     & 0.647\%       & 0.190     & 0.266   & 0.995   \\
AFT         & -23.575\%    & 0.008    & 0.007    & 0.000    & AFT          & -20.428\%     & 0.008     & 0.007   & 0.000   \\ \hline
\multicolumn{5}{c}{Case 3 (90\%   invalid IVs with InSIDE)} & \multicolumn{5}{c}{Case 4 (90\% invalid IVs without InSIDE)} \\ \hline
Model       & BIAS         & SD      & SE      & CP       & Model        & BIAS          & SD       & SE     & CP      \\ \hline
DNN\_ET     & -0.037\%     & 0.259    & 0.254    & 0.940    & DNN\_ET      & 0.992\%       & 0.260     & 0.257   & 0.940   \\
DNN\_EL     & -1.490\%     & 0.262    & 0.260    & 0.935    & DNN\_EL      & -0.658\%      & 0.269     & 0.316   & 0.940   \\
DNN\_CUE    & 2.692\%      & 0.262    & 0.268    & 0.945    & DNN\_CUE     & 3.480\%       & 0.262     & 0.272   & 0.950   \\
LR\_ET      & 26.322\%     & 0.160    & 0.209    & 0.980    & LR\_ET       & 35.205\%      & 0.165     & 0.211   & 0.985   \\
LR\_EL      & 26.382\%     & 0.161    & 0.228    & 0.985    & LR\_EL       & 35.154\%      & 0.168     & 0.269   & 0.985   \\
LR\_CUE     & 27.552\%     & 0.161    & 0.219    & 0.980    & LR\_CUE      & 36.221\%      & 0.165     & 0.223   & 0.985   \\
RF\_ET      & -12.135\%    & 0.194    & 0.271    & 0.995    & RF\_ET       & -4.082\%      & 0.199     & 0.261   & 0.990   \\
RF\_EL      & -12.999\%    & 0.195    & 0.282    & 0.990    & RF\_EL       & -3.881\%      & 0.202     & 0.323   & 0.990   \\
RF\_CUE     & -9.882\%     & 0.195    & 0.279    & 0.995    & RF\_CUE      & -1.917\%      & 0.205     & 0.269   & 0.995   \\
XGB\_ET     & 5.270\%      & 0.190    & 0.252    & 0.995    & XGB\_ET      & 13.549\%      & 0.197     & 0.243   & 0.985   \\
XGB\_EL     & 5.123\%      & 0.192    & 0.263    & 0.995    & XGB\_EL      & 13.181\%      & 0.199     & 0.318   & 0.990   \\
XGB\_CUE    & 6.830\%      & 0.191    & 0.263    & 1.000    & XGB\_CUE     & 15.756\%      & 0.199     & 0.253   & 0.990   \\
AFT         & -22.861\%    & 0.010    & 0.007    & 0.000    & AFT          & -23.093\%     & 0.009     & 0.007   & 0.000   \\ \hline
\end{tabular}

\end{table}
\subsection{Type I error and power of over-identification test}

We next evaluate the finite-sample performance of the proposed over-identification test. 
The simulation design follows Section~\ref{suppsec:censorsimu} in the Supplementary Material with linear nuisance functions.  Considering a setting with violations of causal identification assumptions in Assumption \ref{condcondind} (i),  we allow the unmeasured confounder $U$ to depend on the candidate IVs  $\Z$ in both the exposure and  outcome models. Specifically,
\[
A = \alpha_1(\Z) + \alpha_2(\X) + \bigl(1+\vartheta\sum_{j=1}^m\mu_{A,j} Z_j\bigr)U
+ \bigl(1+\sum_{j=1}^{m} \upsilon_j Z_j \bigr)\epsilon_A,
\]
\[
T = \beta_0 A + \gamma_1(\Z) + \gamma_2(\X)
- \bigl(1+\vartheta\sum_{j=1}^m \mu_{T,j} Z_j\bigr)U + \epsilon_T,
\]
where $\vartheta \in \{0,0.1,0.2,0.3,0.4\}$ controls the degree of violation of Assumption \ref{condcondind} (i), and $\mu_{A,j}, \mu_{T,j}$ are drawn once from $N\!\bigl(0,\,h^2/(1.5m)\bigr)$, where $m=20$ and $h^2=0.2$. When $\vartheta=0$,  Assumption \ref{condcondind} (i) holds, allowing us to evaluate type I error under the null. For $\vartheta > 0$,  Assumption \ref{condcondind} (i) is violated, and we will assess the power of our test.

% Under model misspecification, additional simulations are carried out to evaluate the over-identification test statistic. The model setup is the same as that in Section \ref{suppsec:censorsimu} in the Supplementary Material under linear nuisance, except that we allow the unmeasured confounder $U$ to depend on the instruments $\Z$ in both the exposure and outcome equations, which induces violations of the instrument independence assumption. That is, $A$ and $T$ are generated by
% $$
% A = \alpha_1(\Z) + \alpha_2(\X) + (1+\vartheta\sum_{i=1}^m\delta_{A,i} Z_i)U + (1 + \sum_{i=1}^{m} \delta_i Z_i )\epsilon_A,
% $$
% $$
% T = \beta_0 A + \gamma_1(\Z) + \gamma_2(\X) - (1+\vartheta\sum_{i=1}^m \delta_{T,i}Z_i)U + \epsilon_T,    
% $$
% where $\vartheta$ is used to control the level of model misspecification, and $\delta_{A,i}, \delta_{T,i}$ are constants generated once from a normal distribution $N(0,h^2/(1.5m))$. In particular, we vary $\vartheta$ in $\{0, 0.1, 0.2, 0.3, 0.4\}$. When $\vartheta = 0$, there is no model misspecification. As $\vartheta$ increases, the degree of model misspecification becomes more severe. Table \ref{table-test} presents the power of the test under 200 repeated experiments (when $\vartheta = 0$, the result means the size of the test).

The simulation results in Table~\ref{table-test} show that our proposed test has correct type I error. 
Among the GEL criteria, CUE exhibits slightly lower type I error than EL and ET. As $\vartheta$ increases, the power increases sharply. The results show that our proposed over-identification test maintains correct Type~I error  
and has good power to detect violations of model assumptions and validity of moment conditions.

\begin{table}
\caption{Type I error and power of the proposed over-identification test, based on 200 replications with $n=10{,}000$ and $m=20$. Results with $\vartheta=0$ correspond to Type I error, while results with $\vartheta>0$ correspond to power. \label{table-test}}
\begin{tabular}{ccccccccccc}
\hline
& Censoring rate & \multicolumn{3}{c}{0} & \multicolumn{3}{c}{0.2} & \multicolumn{3}{c}{0.4} \\ \cline{2-11} 
$\vartheta$          & Model       & ET    & EL    & CUE   & ET     & EL     & CUE   & ET     & EL     & CUE   \\ \hline
\multirow{2}{*}{0}   & DNN         & 0.045 & 0.045 & 0.045 & 0.060  & 0.050  & 0.040 & 0.050  & 0.050  & 0.040 \\
& LR          & 0.035 & 0.035 & 0.030 & 0.045  & 0.040  & 0.035 & 0.045  & 0.040  & 0.030 \\ \hline
\multirow{2}{*}{0.1} & DNN         & 0.145 & 0.150 & 0.120 & 0.115  & 0.100  & 0.075 & 0.110  & 0.095  & 0.090 \\
& LR          & 0.160 & 0.155 & 0.115 & 0.125  & 0.110  & 0.080 & 0.075  & 0.075  & 0.065 \\ \hline
\multirow{2}{*}{0.2} & DNN        & 0.700 & 0.665 & 0.660 & 0.550  & 0.540  & 0.500 & 0.405  & 0.380  & 0.340 \\
& LR          & 0.715 & 0.700 & 0.675 & 0.545  & 0.530  & 0.510 & 0.365  & 0.335  & 0.295 \\ \hline
\multirow{2}{*}{0.3} & DNN         & 0.995 & 0.995 & 0.990 & 0.920  & 0.905  & 0.890 & 0.765  & 0.765  & 0.750 \\
& LR          & 0.990 & 0.990 & 0.990 & 0.920  & 0.920  & 0.910 & 0.750  & 0.725  & 0.700 \\ \hline
\multirow{2}{*}{0.4} & DNN         & 1.000 & 1.000 & 1.000 & 0.995  & 0.995  & 0.995 & 0.975  & 0.975  & 0.965 \\
& LR          & 1.000 & 1.000 & 1.000 & 1.000  & 1.000  & 1.000 & 0.975  & 0.975  & 0.965 \\ \hline
\end{tabular}

\end{table}

\section{Application to Time-to-Event Outcomes in the UK Biobank}
\label{sec:ukb}

\subsection{Data description}
\label{subsec:ukb-data}

UK Biobank is a population-based cohort of approximately 500,000 adults recruited in the 
United Kingdom between 2006 and 2010 \citep{Sudlow2015}. At baseline, participants underwent 
physical examinations, provided biological samples, and completed lifestyle and health 
questionnaires. Follow-up is achieved through linkage to national electronic health records (EHR), which capture hospital admissions, primary care visits, and death records. These 
records include dates of diagnoses (coded using ICD-9/10), prescriptions, and medical 
procedures, enabling the study of disease onset and progression using real-world clinical 
data. We focus on three right-censored outcomes derived from EHR: ischaemic heart disease, 
delirium dementia, and diabetes mellitus. Each time-to-event outcome is defined by the first recorded 
date of a relevant ICD-9/10 code mapped to grouped disease phenotypes \citep{dey2022}. Time 
is measured in months from birth year to diagnosis, and for censored participants, follow-up 
is defined as the latest of the last UK Biobank visit, most recent hospital or primary care 
record, or date of death.  

The exposure of interest is BMI, 
with age, sex, and systolic blood pressure as covariates. For genetic data pre-processing, we follow previous quality control procedures 
\citep{liu2023,ye2024}. We restrict the sample to unrelated individuals of white British 
ancestry and aged older than 65 years.  Following prior literature \citep{locke2015,ye2024}, we use $m=93$ single-nucleotide polymorphisms (SNPs) as candidate instruments for the exposure BMI.

The final sample for analysis comprises 32,810 individuals with ischaemic heart disease (3,944 events), 
35,722 with delirium dementia (1,039 events), and 34,213 with diabetes mellitus (2,664 events). The incidence of these diseases increases dramatically after age 65, and clinical guidelines emphasize body weight management
to reduce cardiovascular, and metabolic risks \citep{lane2022}. More
details on disease subcategories, event counts, and censoring proportions are provided in
Section~\ref{suppsec:ukb} of the Supplementary Material.

\subsection{Results and interpretation}

We apply our \texttt{MAWII-Surv} framework to estimate the causal effects of BMI on the three right-censored time-to-event outcomes described in Section~\ref{subsec:ukb-data}. Specifically, we use DNNs to estimate the Neyman orthogonal nuisance functions and employ ET as the GEL criterion. The DNN is a two-layer fully connected 
network with 50 neurons per layer, trained with the Adam optimizer (learning rate 0.0005, 
batch size 256) for up to 1,000 epochs with early stopping (see Section~\ref{suppsec:censorsimu} for details). We also report 
results with random forests, extreme gradient boosting, and linear regressions as nuisance estimators, and with EL and CUE as alternative 
GEL functions, which show similar results; see Section~\ref{suppsec:ukb} of the Supplementary Material. 
Table~\ref{tab:ihd} summarizes the estimated coefficients across specifications.  

\begin{table}
\center
\caption{Estimated causal effect of BMI on three survival outcomes. The $\exp(\hat\beta)$ column shows the point estimate with standard error (SE) in parentheses, and the corresponding $p$-value is reported in the third column. 
$F_{\text{MAWII}}>2$ suggests the presence of heteroscedasticity; $p_{\text{over-ID}}$: $p$-value for the over-identification test. \texttt{MAWII-Surv} uses DNNs to estimate nuisances and ET as the GEL function. AFT (without covariates) regresses log-survival time on BMI only, and AFT (with covariates) regresses log-survival time on BMI together with age, sex, and systolic blood pressure.\label{tab:ihd}} 
\begin{tabular}{lrrrr}
\hline
Model & \multicolumn{1}{c}{$\exp(\hat{\beta})$ (SE)} & $p$-value & \multicolumn{1}{c}{$F_\textrm{MAWII}$} & \multicolumn{1}{c}{$p_{\textrm{over-ID}}$} \\
\hline
\multicolumn{5}{c}{{Ischaemic Heart Disease}} \\\hline
AFT (without covariates) & $0.9766\ (0.0019)$ & $<0.001$ & -- & -- \\
AFT (with covariates)    & $0.9760\ (0.0020)$ & $<0.001$ & -- & -- \\
\texttt{MAWII-Surv}     & $0.8776\ (0.0422)$ & $0.007$  & 2.10 & 0.889 \\\hline
\multicolumn{5}{c}{{Delirium Dementia}} \\\hline
AFT (without covariates) & $0.9988\ (0.0023)$ & $0.602$ & -- & -- \\
AFT (with covariates)    & $0.9991\ (0.0024)$ & $0.708$ & -- & -- \\
\texttt{MAWII-Surv}     & $1.0867\ (0.0231)$ & $<0.001$ & 2.01 & 0.064 \\\hline
\multicolumn{5}{c}{{Diabetes Mellitus}} \\\hline
AFT (without covariates) & $0.9315\ (0.0022)$ & $<0.001$ & -- & -- \\
AFT (with covariates)    & $0.9290\ (0.0022)$ & $<0.001$ & -- & -- \\
\texttt{MAWII-Surv}     & $0.7291\ (0.0399)$ & $<0.001$ & 2.76 & 0.435 \\
\hline
\end{tabular}
\end{table}

Across the three diseases, the estimated $F$-statistics are 2.10, 2.01, and 2.76 respectively, providing evidence of heteroscedasticity in BMI. The over-identification 
tests yield $p$-values of 0.889, 0.064, and 0.435 respectively providing no evidence against the 
validity of the moment conditions. Traditional AFT models serve as benchmarks because they cannot adequately control unmeasured confounding. 
%This limitation is evident in our analysis. For delirium dementia, the traditional AFT model produces statistically insignificant associations regardless of covariate adjustment. By contrast, \texttt{MAWII-Surv}, which is developed to remove unmeasured confounding,  detects a significant protective effect. For ischaemic heart disease and diabetes mellitus, the AFT model may underestimate the estimated effects relative to \texttt{MAWII-Surv}, again reflecting the bias from unmeasured confounding. 

Under \texttt{MAWII-Surv}, the estimated causal effect on time to ischaemic heart disease is
$\exp(\hat\beta_{\text{IHD}}) = 0.8776 \ (95\%$ CI: $ [0.7949,\ 0.9603])$,
indicating that a one-unit increase in BMI multiplies the time to disease onset by $0.8776$, corresponding to an average $12.24\%$ earlier onset.  For diabetes mellitus, the estimate is
$\exp(\hat\beta_{\text{DM}}) = 0.7291\ (95\%$ CI: $ [0.6509,\ 0.8073])$,
implying that a one-unit increase in BMI multiplies the time to disease onset by $0.7291$, corresponding to an average $27.09\%$ earlier onset. For those two diseases, the traditional AFT model appears to underestimate the effects relative to \texttt{MAWII-Surv}, likely reflecting the influence of unmeasured confounding.

For delirium dementia, the estimated causal effect under \texttt{MAWII-Surv} is
$\exp(\hat\beta_{\text{DD}}) = 1.0867 \ ( 
95\%$ CI: $ [1.0416,\ 1.1318])$,
indicating that each one-unit increase in BMI multiplies the time to disease onset by $1.0867$, corresponding to an average $8.67\%$ longer time to onset.
Although
counterintuitive, this finding is in line with the so-called obesity paradox
\citep{elagizi2018}. The association is robust across alternative nuisance estimators
(Table~\ref{supptab:ihd}, Supplementary Material). Whereas midlife obesity has been linked
to elevated dementia risk \citep{xu2011}, later-life studies report the opposite pattern,
with higher BMI associated with slower cognitive decline or reduced incidence and
mortality of dementia \citep{fitzpatrick2009,natale2023obesity}. Importantly, \texttt{MAWII-Surv} uncovers this seemingly paradoxical protective association, which is entirely missed by standard AFT models, thereby underscoring the ability of robust MR methods to reveal clinically relevant associations masked by unmeasured confounding. Nevertheless, the results should be interpreted with caution given the potential for selection bias in biobank data.
\citep{swanson2012uk}.

\section{Discussion}
\label{sec:discussion}

We developed \texttt{MAWII-Surv}, a semiparametric framework for causal inference with censored outcomes and many weak  invalid instruments. Our approach combines a heteroscedasticity-based identification strategy with generalized empirical likelihood framework for inference, extending existing IV tools to settings with right-censored survival outcomes and non-Neyman orthogonal nuisance functions. The key contributions are fourfold: (i) adapt heteroscedasticity-based identification to right-censored survival outcomes; (ii) provide a systematic treatment of Neyman orthogonal versus non-Neyman orthogonal nuisance functions; (iii) develop GEL 2.0, which enables valid inference without Neyman orthogonality or cross-fitting; and (iv) introduce new inferential tools, including a variance decomposition that quantifies the effect of non-Neyman orthogonal nuisance estimation and an over-identification test for censored outcomes. In the UK Biobank data analysis, \texttt{MAWII-Surv} can help correct the attenuation bias of standard AFT models. As a practical recommendation, we suggest combining deep neural networks with empirical likelihood or exponential tilting for robust performance in complex data applications.

Several open directions remain. One is to extend the framework from the structural AFT model to more general transformation models \citep{zeng2007maximum}. Another is to establish semiparametric efficiency bound under many weak and non-Neyman orthogonal moment conditions, as existing theory generally does not cover this setting. More broadly, the techniques underlying GEL 2.0 provide a foundation for  inference in other challenging problems such as missing not at random, causal inference with time-varying confounding, and deep learning applications.

\section*{Acknowledgments}
This research has been conducted using the UK Biobank resource under application number 52008.

\newpage
\appendix
\begin{center}
\LARGE \textbf{Supplementary Material}
\end{center}
\setcounter{lemma}{1}
\setcounter{assumption}{5}
\renewcommand{\thetable}{S\arabic{table}}
\setcounter{table}{0}
\renewcommand{\thesection}{S\arabic{section}}
\renewcommand{\thesubsection}{S\arabic{section}.\arabic{subsection}}
This Supplementary Material includes technical details for the proofs of Lemmas and Theorems. All symbols in this
section are consistent with those in the main body unless otherwise specified.
%--------------------------------------------
\section{Proof of Theorem \ref{thmid}}
First, we prove $\bbE(g(\beta_0;\O_T))=0$ under \eqref{ET} and \eqref{EA}. 
\bse
&&\bbE(g(\beta_0;\O_T))\\
&=&\bbE\{R_Z(R_AR_T-\beta_0R_A^2-\bbE(R_AR_T-\beta_0R_A^2|\X))\}\\
&=&\bbE\{R_Z(R_AR_T-\beta_0R_A^2)\}-\bbE\{R_Z(\bbE(R_AR_T-\beta_0R_A^2|\X))\}\\
&=&\bbE[\bbE\{R_ZR_A(R_T-\beta_0R_A)|\Z,\X\}]-\bbE\left(\bbE[R_Z\{\bbE(R_AR_T-\beta_0R_A^2|\X)\}|\X]\right)\\
&=&\bbE\left[R_Z\bbE\big\{R_A(T-\beta_0A-\{\bbE(T|\Z,\X)-\beta_0\bbE(A|\Z,\X)\})|\Z,\X\big\}\right]\\
&&+\bbE[\underbrace{\bbE\{\Z-\bbE(\Z|\X)|\X\}}_{=0}\bbE(R_AR_T-\beta_0R_A^2|\X)]\\
&=&\bbE[R_Z\bbE\{R_A(T-\beta_0A)|\Z,\X\}]- \bbE\left(R_Z\bbE\big[R_A\{\bbE(T|\Z,\X)-\beta_0\bbE(A|\Z,\X)\}|\Z,\X\big]\right)\\
&=&\bbE[R_Z\underbrace{\bbE\{R_A(T-\beta_0A)|\X\}}_{\text{by Assumption \ref{condcondind} (ii)}}]- \bbE\left(R_Z\bbE\big[R_A|\Z,\X\big]\{\bbE(T|\Z,\X)-\beta_0\bbE(A|\Z,\X)\}\right)\\
&=&\bbE[\underbrace{\bbE\{\Z-\bbE(\Z|\X)|\X\}}_{=0}\bbE\{R_A(T-\beta_0A)|\X\}]\\
&&- \bbE(R_Z\underbrace{\bbE\{A-\bbE(A|\Z,\X)|\Z,\X\}}_{=0}\{\bbE(T|\Z,\X)-\beta_0\bbE(A|\Z,\X)\})\\
&=&0.
\ese

Noticing that $T\perp C|\X$ and $\delta=0$ when $Y\neq T$, we have
\bse
&&\bbE\{\psi(\beta_0;\O)\}\\
&=&\bbE\{\frac{\delta}{G_0(Y|\O_A)}\{g(\beta_0;\O)-\xi_0(\beta_0;\O_A)\}+\xi_0(\beta_0;\O_A)\}\\
&=&\bbE\{\frac{\delta}{G_0(T|\O_A)}\{g(\beta_0;\O_T)-\xi_0(\beta_0;\O_A)\}+\xi_0(\beta_0;\O_A)\}\\
&=&\bbE\{\frac{I(T\leq C)}{G_0(T|\O_A)}\{g(\beta;\O_T)-\xi_0(\beta_0;\O_A)\}+\xi_0(\beta_0;\O_A)\}\\
&=&\bbE\left(\bbE[\frac{I(T\leq C)}{G_0(T|\O_A)}\{g(\beta_0;\O_T)-\xi_0(\beta_0;\O_A)\}|\O_T]\right)+\bbE\{\xi_0(\beta_0;\O_A)\}\\
&=&\bbE\left(\frac{\bbE\{I(T\leq C)|\O_T\}}{G_0(T|\O_A)}\{g(\beta_0;\O_T)-\xi_0(\beta_0;\O_A)\}\right)+\bbE\{\xi_0(\beta_0;\O_A)\}\\
&=&\bbE\{g(\beta_0;\O_T)-\xi_0(\beta_0;\O_A)\}+\bbE\{\xi_0(\beta_0;\O_A)\}\\
&=&\bbE\{g(\beta_0;\O_T)\}\\
&=&0.
\ese
Next we prove that $\beta_0$ is the unique solution to $\bbE\{\psi(\beta;\O)\}=0$. Suppose there is another $\beta$ also satisfying $\bbE\{\psi(\beta;\O)\}=0$, then we have
\bse
&&\bbE\{\psi(\beta;\O)\}-\bbE\{\psi(\beta_0;\O)\}\n\\
&=&\bbE\{g(\beta,\boldeta_0;\O_T)\}-\bbE\{g(\beta_0,\boldeta_0;\O_T)\}\n\\
&=&\bbE[(\Z-\bbE(\Z|\X))\{R_AR_T-\bbE(R_AR_T|\X)-\beta(R_A^2-\bbE(R_A^2|\X))\}]\n\\
&&-\bbE[(\Z-\bbE(\Z|\X))\{R_AR_T-\bbE(R_AR_T|\X)-\beta_0(R_A^2-\bbE(R_A^2|\X))\}]\n\\
&=&\bbE[(\Z-\bbE(\Z|\X))\{(\beta_0-\beta)(R_A^2-\bbE(R_A^2|\X))\}]\n\\
&=&(\beta_0-\beta)\bbE\{(\Z-\bbE(\Z|\X))(R_A^2-\bbE(R_A^2|\X))\}\n\\
&=&(\beta_0-\beta)\Cov(\Z,R_A^2|\X)\n\\
&=&(\beta_0-\beta)\Cov(\Z,\var(A|\Z,\X)|\X)\n\\
&=&0,
\ese 
which means $\beta=\beta_0$, given the heteroskedasticity assumption $\Cov(\Z,\var(A|\Z,\X)|\X)\neq 0$.

\section{Neyman orthogonal nuisance, non-Neyman orthogonal nuisance, and Proof of Proposition \ref{lemmaneyman} \label{pflemmaneyman}}

% Formal definitions of Neyman orthogonality

In this section, we provide the formal mathematical definitions of G\^ateaux and Fr\'echet derivatives, and show how they characterize Neyman orthogonality of nuisance functions. These details are omitted from the main text for readability.

\subsection{ G\^ateaux and Fr\'echet derivatives}

Let $\Psi(\beta,\boldeta;\O)$ be a generic moment function that depends on the parameter of interest $\beta$, a nuisance function $\boldeta$, and data $\O$. Suppose the true values are $(\beta_0,\boldeta_0)$.

\paragraph{G\^ateaux derivative.}  
The G\^ateaux derivative \citep{gateaux1919} measures the sensitivity of $\Psi$ along a fixed direction. For any fixed direction $\boldeta-\boldeta_0$, the G\^ateaux derivative of $\Psi(\beta_0,\cdot;\O)$ at $\boldeta_0$ in the direction $\boldeta-\boldeta_0$ is
\[
\left. \frac{\partial}{\partial t}\bbE_{P_\O}\left[\Psi(\beta_0,\boldeta_0+t(\boldeta-\boldeta_0);\O)\right]\right|_{t=0}.
\]

\paragraph{Fr\'echet derivative.}  
The Fr\'echet derivative \citep{frechet1906} is a stronger notion, requiring a uniform linear approximation in all directions. Formally, $D_{\boldeta}\Psi(\beta_0,\boldeta_0;\O)(\cdot)$ is the unique bounded linear operator such that
\[
\lim_{\|\boldeta-\boldeta_0\|_{L^2}\to 0}
\frac{\big\|\Psi(\beta_0,\boldeta;\O)-\Psi(\beta_0,\boldeta_0;\O)-D_{\boldeta}\Psi(\beta_0,\boldeta_0;\O)(\boldeta-\boldeta_0)\big\|}
{\|\boldeta-\boldeta_0\|_{L^2}} = 0.
\]

\subsection{ Neyman orthogonality}

With these definitions in place, a nuisance function $\boldeta$ is said to be \textit{Neyman orthogonal} if the moment function is locally insensitive to perturbations in $\boldeta$. That is,
\[
\bbE_{P_\O}\left[ D_{\boldeta}\Psi(\beta_0,G_0,\boldeta_0;\O)(\boldeta-\boldeta_0)\right] = 0 \quad \text{for all } \boldeta-\boldeta_0.
\]
Equivalently, the G\^ateaux derivative vanishes in every fixed direction.  

In practice, when $\boldeta$ is estimated by some $\hat{\boldeta}$, Neyman orthogonality ensures that the first-order asymptotics of the estimator of $\beta$ are unaffected by small estimation errors in $\hat{\boldeta}$.

\subsection{Neyman Orthogonal vs. non-Neyman orthogonal nuisances in our setting}

For the AIPCW moment function considered in the main text,
\[
\psi(\beta;\O) = \frac{\delta}{G_0(Y|\O_A)}g(\beta;\O) + \Big(1-\frac{\delta}{G_0(Y|\O_A)}\Big)\xi_0(\beta;\O_A),
\]
we show that:
\begin{itemize}
\item Conditional expectations such as $\bbE(\Z|\X)$, $\bbE(A|\Z,\X)$, $\bbE(Y|\Z,\X)$, $\bbE(R^2_A|\X)$, $\bbE(R_AR_Y|\X)$, and $\xi_0(\beta_0;\O_A)$ are \textit{Neyman orthogonal nuisances}. Their estimation error is asymptotically negligible to first order.
\item The conditional censoring distribution $G(\cdot|\O_A)$ is a \textit{non-Neyman orthogonal nuisance}. Its G\^ateaux derivative does not vanish, because $G(T|\O_A)$ depends jointly on both $T$ and $\O_A$. As a result, estimation error in $G$ enters directly into the first-order expansion of the estimator and contributes an additional variance component $V_2$.
\end{itemize}

This distinction underpins the variance decomposition in Section \ref{secvardecom} of the main text, where $V_2$ explicitly accounts for the impact of estimating the non-Neyman orthogonal nuisance $G(\cdot|\O_A)$.

\section{Proof of Proposition \ref{lemmaneyman} }\label{suppproofproposition1}

For the G\^ateaux derivative of $G$, we have
\be
&&\bbE \frac{\partial \psi(\beta_0,G_0+t(G-G_0),\boldeta_0;\O)}{\partial t}\Big |_{t=0}
\n\\
&=& \bbE \frac{\delta \{ G(Y \mid \O_A) - G_0(Y \mid \O_A)\}}{G_0^2(Y \mid \O_A)}\left\{ g(\beta_0,\boldeta_0;\O) - \xi_0(\beta_0;\O_A) \right\}\n\\
&=&\bbE\left( \bbE\left[\frac{\delta \{ G(T \mid \O_A) - G_0(T \mid \O_A)\}}{G_0^2(T \mid \O_A)}\left\{ \xi_0(\beta_0;\O_A) \right\}|\O_T\right]\right) \n\\
&=& \bbE\left[\frac{ \{  G(T \mid \O_A) - G_0(T \mid \O_A)\}}{G_0(T \mid \O_A)}\left\{ g(\beta_0,\boldeta_0;\O_T) - \xi_0(\beta_0;\O_A) \right\}\right].\n
\ee
Note that though written as $G(T | \O_A)$, $G$ is actually a function of $T$ and $\O_A$, thus we can not cancel this term through taking expectation about $\O_A$. This means it cannot be simplified any further, and the first order G\^ateaux derivative is not 0. Therefore, $\psi$ is not Neyman Orthogonal to $G$. Let $\boldeta$ collect $\{\bbE(\Z|\X),\bbE(A|\Z,\X),\bbE(Y|\Z,\X),\bbE(R_A|\X),\bbE(R_T|\X),\xi\}$. For any direction $\boldeta^*=\{f^*,h_1^*,\dots,h_4^*,\xi^*\}$, define $\Phi(t;\O)=\psi(\beta_0,G_0,\boldeta_0+t(\boldeta^*-\boldeta_0);\O)$. Then we have
\be
&&\bbE \frac{\partial \psi(\beta_0,G_0,\boldeta_0+t(\boldeta^*-\boldeta_0);\O)}{\partial t}\Big |_{t=0}\n\\
&=&\bbE\left( \frac{\delta}{G_0}\frac{\partial g(\beta_0,\boldeta_0+t(\boldeta^*-\boldeta_0);\O)}{\partial t}\Big |_{t=0}\right)+\bbE\left( (1-\frac{\delta}{G_0})(\xi^*(\beta_0;\O_A)-\xi_0(\beta_0;\O_A))\right)\n\\
&=&\bbE\left\{\partial[\Z-\{f+t(f^*-f)\}]\Big([A-\{h_1+t(h_1^*-h_1)\}][T-\{h_2+t(h_2^*-h_2)\}\right.\n\\
&&\left.\quad -\beta_0(A-\{h_1+t(h_1^*-h_1)\})]-[h_3+t(h_3^*-h_3)-\beta_0\{h_4+t(h_4^*-h_4)\}]\Big)/\partial t|_{t=0}\right\}\n\\
&&+\bbE\left\{\bbE\left( (1-\frac{\delta}{G_0})(\xi(\beta_0;\O_A)-\xi_0(\beta_0;\O_A))|\O_T\right)\right\}\n\\
&=&\bbE\left\{(f-f^* )(R_A(R_T-\beta_0R_A)-h_3+\beta_0h_4)\right\}\n\\
&&+\bbE\left\{(\Z-f)(h_1-h_1^*)(R_T-\beta_0R_A)\right\}\n\\
&&+\bbE\left[(\Z-f)R_A\{(h_2-h_2^*)-\beta_0(h_1-h_1^*)\}\right]  \n\\
&&+\bbE\left((\Z-f)[R_A\{R_T-\beta_0R_A\}-(h_3^*-h_3)+\beta_0(h_4- h_4^*)]\right)  \n\\
&=&\bbE\left[\bbE\left\{(f-f^* )(R_A(R_T-\beta_0R_A)-h_3+\beta_0h_4)|\X\right\}\right]\n\\
&&+\bbE\left[\bbE\left\{(\Z-f)(h_1-h_1^*)(R_T-\beta_0R_A)|\Z,\X\right\}\right]\n\\
&&+\bbE\left(\bbE\left[(\Z-f)R_A\{(h_2-h_2^*)-\beta_0(h_1-h_1^*)\}|\Z,\X\right]\right)  \n\\
&&+\bbE\left[\bbE\left((\Z-f)\{(h_3-h_3^*)-\beta_0( h_4-h_4^*)\}|\X \right) \right] \n\\
&=&\bbE\left[(f-f^* )\bbE\left\{(R_AR_T-h_3)-\beta_0(R_A^2-h_4)|\X\right\}\right]\n\\
&&+\bbE\left[(\Z-f)(h_1-h_1^*)\bbE\left\{R_T-\beta_0R_A|\Z,\X\right\}\right]\n\\
&&+\bbE\left[(\Z-f)\{(h_2-h_2^*)-\beta_0(h_1-h_1^*)\}\bbE(R_A|\Z,\X)\right] \n\\
&&+\bbE\left[\{(h_3-h_3^*)-\beta_0( h_4-h_4^*)\}\bbE\left(\Z-f|\X \right) \right] \n\\
&=&0\n.
\ee  
This implies that $\psi$ is Neyman Orthogonal to $\boldeta$. We can also write into the form of Fr\'echet derivative 
\bse
\left.\frac{\partial}{\partial t} \bbE\left[\psi(\beta_0, \boldeta_0 + t (\boldeta - \boldeta_0))\right]\right|_{t = 0} 
= \bbE\left[ D_{\boldeta} \psi(\beta_0, \boldeta_0)(\boldeta - \boldeta_0) \right] = 0,
\ese
where 
\bse
D_{\boldeta} \psi(\beta_0, \boldeta_0)
= \left(
\begin{array}{c}
\frac{\delta}{G_0}(R_A(R_T-\beta_0R_A)-h_3+\beta_0h_4) \\
\frac{\delta}{G_0}(\Z - f)(R_T - 2\beta_0 R_A) \\
\frac{\delta}{G_0}(\Z - f) R_A  \\
\frac{\delta}{G_0}(f - \Z)  \\
\frac{\delta}{G_0}\beta_0 (\Z - f)\\
1-\frac{\delta}{G_0}
\end{array}
\right)\tp
\ese
and 
\bse
\boldeta^* - \boldeta &=& 
\left(
\begin{array}{l}
f^*-f_0\\
h_1^* - h_1 \\
h_2^* - h_2 \\
h_3^* - h_3 \\
h_4^* - h_4 \\
\xi^*-\xi_0
\end{array}
\right).
\ese

\section{Proof of Theorem \ref{thmcensor}} \label{pfthm2}
Since the uncensored case is easier than the censored case, we first present Corollary \ref{coro1} and its proof, and then generalize it to the right-censored case.

For notational simplicity, we omit the superscript and subscript in \(\hat\beta_{\mathrm{GEL2}}^{\mathrm{C}}\) and write it simply as \(\hat\beta\) throughout the proof. Since uncensored data can be viewed as a special case of censored data, and the proof is generally simpler in the absence of censoring, we first establish key results under the uncensored setting. Here are some notations for the uncensored data, where $\psi$ reduces to $g$. Let
\bse
&&\quad \bar g(\beta,\boldeta)=\frac{1}{n} \sum_{i=1}^n g(\beta,\boldeta;\O_i),\quad\tilde g(\beta,\boldeta)=\bbE(g(\beta,\boldeta;\O)),\\
&& \bar\Omega(\beta,\boldeta)=\frac{1}{n} \sum_{i=1}^n g(\beta,\boldeta;\O_i)g(\beta,\boldeta;\O_i)\tp,\quad \Omega(\beta,\boldeta) = \bbE(g(\beta,\boldeta;\O)g(\beta,\boldeta;\O)\tp),\quad \Omega_0=\Omega(\beta_0,\boldeta_0),\\
&& g^{'}(\boldeta;\O_i)=\frac{\partial g(\beta,\boldeta;\O_i)}{\partial \beta},\quad  \bar {g^{'}}(\boldeta)=\frac{1}{n} \sum_{i=1}^n g^{'}(\boldeta;\O_i), \quad g^{*}=\bbE(g^{'}(\boldeta_0;\O)).
\ese

\subsection{Regularity conditions}\label{additionalcond}

\begin{assumption}[Bounded higher-order moments]\label{condmoment}
(i). $\bbE(T^8), \bbE(A^8)$, and $\Z$ are bounded;\\
(ii). $\{\bbE(\|\psi(\beta_0,G_0,\boldeta_0;\O)\|^4)+\bbE(\|\psi^{'}(G_0,\boldeta_0;\O)\|^4)\}m/n$=o(1);\\
(iii). There exists $\gamma>4$ such that $n^{1/\gamma}\{\bbE(\sup_{\beta\in \calB}\|\psi(\beta,G_0,\boldeta_0;\O)\|^\gamma)\}^{1/\gamma}(m+\mu_n)/\sqrt{n}\rightarrow 0$. 
\end{assumption}
\begin{assumption}[Bounded eigenvalue] \label{condeigen}
There exists a positive constant $c$ such that $ 1/c\leq \lambda_{\min}(\Sigma(\beta,\boldeta_0)) < \lambda_{\max}(\Sigma(\beta,\boldeta_0))\leq c $ for all $\beta\in\calB$, and $\lambda_{\max}\bbE(\psi^{'}(G_0,\boldeta_0;\O)\psi^{'}(G_0,\boldeta_0;\O)\tp)\leq c $.
\end{assumption}

Assumptions \ref{condmoment} (i) and \ref{condeigen} require that $\bbE(T^8), \bbE(A^8)$, $\Z$ and the eigenvalue of matrices $\Sigma$ are bounded. Assumptions \ref{condmoment} (ii) and (iii) impose further constraints on the AIPCW moment function $\psi$, which are particularly important when $m$ diverges with $n$ or when the moment conditions are weak. These assumptions are also commonly adopted in the weak IV literature and ensure the validity of asymptotic expansions for GEL-type estimators; see Assumption 6 in \citet{newey2009} and Assumptions 5 and 6 in \citet{ye2024}.

\begin{assumption}[GEL function] \label{condgel}
(i). $\rho(\cdot)$ is concave and three times continuously differentiable; (ii). $\rho(0)=0$, and $\rho^{'}(0)=\rho^{''}(0)=-1$.
\end{assumption}
Assumption \ref{condgel} (i) requires sufficient smoothness to allow for a third-order Taylor expansion. Assumption \ref{condgel} (ii) normalizes the function at zero, ensuring the GEL objective is centered and correctly scaled. These assumptions are standard in the GEL literature and align with the requirements of GEL function in \citet{newey2004}.

\begin{assumption}[Local Kaplan-Meier estimator]  \label{condkernel} 
(i). The conditional survival function $G_0$ is $\lfloor 3(m+d_x+1)/2\rfloor +1$ order differentiable in $t$ and in each component of $\O_A$, with all partial derivatives uniformly bounded. In addition, there exists a constant $c > 0$ such that $\inf_{\O_A} P(t \leq T \leq C |\O_A) \geq c$. Moreover, there exist constants $b_1, b_2 > 0$ such that 
$$     
b_1 \leq \sup\{ t : G_0(t| \O_A) > 0 \} \leq b_2
$$
holds uniformly over $\O_A$, and $\sup\{ t : P(T > t| \O_A) > 0 \} \geq \sup\{ t : G_0(t|\O_A) > 0 \}$ almost surely in $\O_A$;

(ii). The bandwidth $h$ satisfies $nh^{3(d_x+m+1)+2} \to 0$, $[\log(n)]^{-3}n h^{3(d_x+m+1)} \to \infty$ as $n \to \infty$;

(iii). The higher-order kernel function $K$ is a $\lfloor 3(m+d_x+1)/2\rfloor +1$-order kernel, e.g., $\int K(u)du=1$, $\int u^i K(u)du=0$ for $i=1,\dots, \lfloor 3(m+d_x+1)/2\rfloor$, and $\int u^{\lfloor 3(m+d_x+1)/2\rfloor+1} K(u)du <\infty$.
\end{assumption}
Assumption~\ref{condkernel} imposes regularity assumptions on the conditional censoring distribution $G_0(t | \O_A)$, the higher-order kernel function $K$, and bandwidth $h$ used in the nonparametric estimation. Specifically, Assumption~\ref{condkernel} (i) ensures that $G_0(t | \O_A)$ is sufficiently smooth and uniformly bounded away from zero over its support. This prevents the inverse weight $\delta / G_0(Y | \O_A)$ in the estimator from becoming excessively large, which could otherwise lead to numerical instability. Assumptions \ref{condkernel} (ii) and (iii) impose bandwidth assumptions, ensuring the integrated local Kaplan-Meier estimator achieves $\sqrt{n}$ asymptotic normality. In the case of local Kaplan-Meier estimation with finite-dimensional covariates, asymptotic normality has already been established, for example, by \cite{stute1995,gerds2017}. Since in our setting the dimension $m$ diverges, we employ higher-order kernel to control the bias and impose stronger requirements on the bandwidth in Assumptions \ref{condkernel} (ii) and (iii). The higher order kernel has previously been used in the literature; see, for example, \citet{parzen1962,fan1992}, and a concrete construction of a higher order kernel is given in Section \ref{suppsec:kernel}. 

\textbf{Remark}: Assumptions \ref{condweak} (ii) and \ref{condkernel} jointly ensure that the influence of the non-Neyman orthogonal nuisance function $G$ on the estimation of $\beta$ is minimized. Specifically, the first-order bias introduced by the non-orthogonality of $G$ is eliminated by employing a high-order kernel and constraining the moment strength in Assumption {condweak} (ii), while the second-order bias is controlled via the convergence rate requirement imposed in Assumption \ref{condkernel} (ii). However, the additional variance induced by the lack of Neyman orthogonality is generally difficult to eliminate, as we show in the following Theorem \ref{thmcensor}. In other words, while the bias resulting from a non-Neyman orthogonal nuisance can be removed through careful estimation and regularity conditions, variance inflation is often an inherent consequence in such settings.

\subsection{Some useful Lemmas and their proofs}
\begin{lemma} \label{lemma1}
(i) Under Assumptions \ref{condweak} and \ref{condeigen}, there is a positive constant $c>0$ with $\left|\beta-\beta_0\right| \leq c \sqrt{n}\left\|\tilde{g}\left(\beta, \boldeta_0\right)\right\| / \mu_n$ for all $\beta \in \calB$.

(ii)  Under Assumptions \ref{condweak} and \ref{condeigen}, there is a positive constant $c$ and $\hat{M}=O_p(1)$ such that for all $\beta^{\prime}, \beta \in \calB$,
$$\sqrt{n}\left\|\tilde{g}\left(\beta^{\prime}, \boldeta_0\right)-\tilde{g}\left(\beta, \boldeta_0\right)\right\| / \mu_n \leq c\left|\beta^{\prime}-\beta\right|$$ and $$\sqrt{n}\left\|\bar{g}\left(\beta^{\prime}, \boldeta_0\right)-\bar{g}\left(\beta, \boldeta_0\right)\right\| / \mu_n \leq \hat{M}\left|\beta^{\prime}-\beta\right|.$$

(iii) Under Assumption \ref{condmoment}, $\left|a\tp\left\{\Omega\left(\beta^{\prime}, \boldeta_0\right)-\Omega\left(\beta, \boldeta_0\right)\right\} b\right| \leq c\|a\|\|b\|\left|\beta^{\prime}-\beta\right|$ for all $a, b \in \mathbb{R}^m, \beta^{\prime}, \beta \in \calB.$

(iv) Under Assumption \ref{condeigen}, there is a positive constant $c>0$ such that
$$
\sup _{\beta \in \calB} \bbE\left[\left\{g(\beta, \boldeta_0;\O_i)\tp g(\beta, \boldeta_0;\O_i)\right\}^2\right] \leq c m^2 .
$$
\end{lemma}

Lemma \ref{lemma1} is the same as Lemmas S1-S4 in \cite{ye2024}. So, the proof is omitted here.
\begin{lemma}[Theorem 4.2 in \cite{jiao2023}]\label{lemmadnn}
For any $M,N\in \bbN^{+}$, the function class of ReLU MLP with width $\mathcal{W}=38(\lfloor\zeta\rfloor+$ $1)^2 d^{\lfloor\zeta\rfloor+1} N\left\lceil\log _2(8 N)\right\rceil$ and depth $\mathcal{D}=21(\lfloor\zeta\rfloor+1)^2 M\left\lceil\log _2(8 M)\right\rceil$, for $n \geq \operatorname{Pdim}\left(\mathcal{F}_n\right) / 2$, the prediction error of the $\hat{f}_n$ satisfies
$$
\mathbb{E}\left\|\hat{f}_n-f\right\|_{L^2}^2 \leq \mathcal{E}(\zeta,B_0,\mathcal{W},\mathcal{D},\mathcal{S},d)= c B_0^5(\log (n))^5 \frac{1}{n} \mathcal{S D} \log (\mathcal{S})+324 B_0^2(\lfloor\zeta\rfloor+1)^4 d^{2\lfloor\zeta\rfloor+\zeta \vee 1}(N M)^{-4 \zeta / d}.
$$
\end{lemma}
Together with Assumption \ref{condDNNs}, we can obtain the convergence rate of DNN is $o(n^{-1/4}m^{-1/4})$, and $\|\hat\boldeta-\boldeta_0\|=o_p(n^{-1/4})$. Here we provide a set of configurations under which Assumption 7 is satisfied. Set H\"older smooth parameter $\zeta=m+d_x$, $(m+d_x)^{3(m+d_x)+4.5}=o(n^{1/6})$, $\calW=\calD=O(n^{1/12}), \calS=O(n^{1/4})$, then we have $\mathbb{E}\left\|\hat{f}_n-f\right\|_{L^2}^2 = O((\log n)^6 n^{1/4+1/12-1})+O((m+d_x)^{3(m+d_x)+4}(n^{-(1/12)\times 2\times 4}))=o(n^{-1/2}m^{-1/2})$. If the candidate IVs is concentrated on some neighborhood of a low-dimensional manifold, the error bound of neural networks can be correspondingly loosened, and the restriction $(m+d_x)^{3(m+d_x)+4.5}=o(n^{1/6})$ can be suitably relaxed.  For detailed discussions on low-dimensional manifold assumptions, see, for example, \citet{schmidt2020,nakada2020,jiao2023}.

\begin{lemma}\label{lemmag}
Under Assumptions \ref{condnuispara}-\ref{condeigen}, 
(i). $\left\|\bar{g}\left(\beta_0, \hat{\boldeta}\right)-\bar{g}\left(\beta_0, \boldeta_0\right)\right\|=o_p(n^{-1/2})$;

(ii). $\left\|\bar{g^{'}}(\hat{\boldeta})-\bar{g^{'}}\left(\boldeta_0\right)\right\|=o_p(n^{-1/2})$;

(iii). $\sup _{\beta \in \calB}\left\|\bar{g}(\beta, \hat{\boldeta})-\bar{g}\left(\beta, \boldeta_0\right)\right\|=o_p(n^{-1/2})$;   

(iv). $\sup_{\beta\in\calB}\|\bar g(\beta,\hat\boldeta)\|=O_p(\mu_n/\sqrt{n})$ and $\|\bar g(\beta_0,\hat\boldeta)\|=O_p(\sqrt{m/n})$.
\end{lemma}
\begin{proof}
(i). Let $g_j(\beta_0,\boldeta;\O_i)$ be the $j$th element of $g(\beta_0,\boldeta;\O_i)$ for $j=1,\dots,m$. Let $\hat R_{A,i} = A_i-\hat h_1(\Z_i,\X_i)$, $\hat R_{Y,i} = Y_i-\hat h_2(\Z_i,\X_i)$, $\Delta_i = R_{A,i}R_{Y,i}-\beta_0R_{A,i}^2$, and $\hat\Delta_i = \hat R_{A,i}\hat R_{Y,i}-\beta_0\hat R_{A,i}^2$. First, we will prove $\bar g_j\left(\beta_0, \hat{\boldeta}\right)-\bar g_j\left(\beta_0, \boldeta_0\right)$ has order $o_p(n^{-1/2}m^{-1/2})$, i.e., $\bbE_n[\{g_j(\beta_0, \hat{\boldeta};\O)-g_j(\beta_0, \boldeta_0;\O)\}]=o_p(n^{-1/2}m^{-1/2})$. It can be decomposed as
\be 
&&\bbE_n[g_j(\beta_0, \hat{\boldeta};\O)-g_j(\beta_0, \boldeta_0;\O)]\n\\
&=&\frac{1}{n}\sum_{i=1}^n(Z_{ij}-\hat f_j(\X_i))(\hat\Delta_{i}-(\hat h_3(\X_i)-\beta_0\hat h_4(\X_i)))\n\\
&&\quad -\frac{1}{n}\sum_{i=1}^n(Z_{ij}- f_j(\X_i))(\Delta_{i}-( h_3(\X_i)-\beta_0 h_4(\X_i)))  \n\\
&=&     \frac{1}{n}\sum_{i=1}^n(Z_{ij}-\hat  f_j(\X_i))(\hat\Delta_{i}-(\hat h_3(\X_i)-\beta_0\hat h_4(\X_i)))\n\\
&&\quad -\frac{1}{n}\sum_{i=1}^n(Z_{ij}- \hat  f_j(\X_i))(\Delta_{i}-(h_3(\X_i)-\beta_0 h_4(\X_i))) \n\\
&&+ \frac{1}{n}\sum_{i=1}^n(Z_{ij}- \hat  f_j(\X_i))(\Delta_{i}-(h_3(\X_i)-\beta_0 h_4(\X_i)))\n\\
&&\quad -\frac{1}{n}\sum_{i=1}^n(Z_{ij}-  f_j(\X_i))(\Delta_{i}-( h_3(\X_i)-\beta_0 h_4(\X_i)))  \n\\
&=&\frac{1}{n}\sum_{i=1}^n(\hat  f_j(\X_i)- f_j(\X_i))(\Delta_{i}-( h_3(\X_i)-\beta_0 h_4(\X_i)))   \n\\
&&+ \frac{1}{n}\sum_{i=1}^n(Z_{ij}- \hat  f_j(\X_i))(\hat R_{A,i}(\hat R_{Y,i}- \beta_0\hat R_{A,i})-R_{A,i}( R_{Y,i}- \beta_0 R_A))     \n\\
&&+  \frac{1}{n}\sum_{i=1}^n(Z_{ij}- \hat  f_j(\X_i))(h_3(\X_i)-\beta_0 h_4(\X_i)-(\hat h_3(\X_i)-\beta_0\hat h_4(\X_i)))      \n\\
&\equiv& \bbE_n(A_1+A_2+A_3)\n.
\ee 
For $A_1$, we will prove $\bbE(A_1)=0$ and $(\bbE_n-\bbE)(A_1)=o_p(n^{-1/2}m^{-1/2})$. First, we have
\be\label{eqpa1}
&&\bbE(A_1)\n\\
&=&\bbE(\bbE((\hat  f_j(\X)- f_j(\X))(\Delta-( h_3(\X)-\beta_0 h_4(\X)))|\X))\n\\
&=&\bbE((\hat  f_j(\X)- f_j(\X))\bbE((\Delta-( h_3(\X)-\beta_0 h_4(\X)))|\X))\n\\
&=& 0. 
\ee 
Then we will prove $(\bbE_n-\bbE)(A_1)=o_p(n^{-1/2}m^{-1/2})$. Here we use the idea of localization \citep{bartlett2005,farrell2021}. Define 
$$
\calL_{j,\delta}=\{(\hat  f_j- f_j)(\Delta-( h_3-\beta_0 h_4)):\hat f_j \in\calG_{1},\|\hat f_j-f_{j}\|_{L^2}\leq\delta\}.
$$
Let $F(\O_i)=2B_0[(|A_i|+B_0)\{|Y_i|+B_0+|\beta_0|(|A_i|+B_0)\}+B_0+|\beta_0|B_0]$. It is easy to prove that $\sup_{l\in \calL_{j,\delta}}\bbE l^2\leq \|F\|^2$ since each term in $F(\O_i)$ is larger than the corresponding term in $A_1$. This also implies that $l\in \calL_{j,\delta}$ is Lipschitz on $\hat f_j-f_j$ since $\Delta-( h_3-\beta_0 h_4)$ is uniformly bounded. Let $\sigma^2=c B_0^5(\log (n))^5 \mathcal{S D} \log (\mathcal{S})/N+324 B_0^2(\lfloor\zeta\rfloor+1)^4 d_x^{2\lfloor\zeta\rfloor+\zeta \vee 1}(N M)^{-4 \zeta / d_x},$ and Lemma \ref{lemmadnn} guarantees that $\sup_{l\in \calL_{j,\delta}}\bbE l^2\leq \sigma^2\leq \|F\|$. Let $\delta=\sigma/\|F\|$, and $F_{\max}=\max_{1\leq i\leq n}F(\O_i)$. According to Theorem 5.2 in \cite{chernozhukov2014}, we have
$$
\bbE\left\{\sup_{l\in\calL_{j,\delta}}|\sqrt{n}(\bbE_n-\bbE)(l)|\right\} \lesssim J(\delta, \calL_{j,\delta}, F)\|F\|+\frac{\|F_{\max}\| J^2(\delta, \calL_{j,\delta}, F)}{\delta^2 \sqrt{n}},
$$
where $a\lesssim b$ means there exists a positive constant $c$ such that $a\leq cb$, $J(\delta, \calL_{j,\delta}, F)=\int_{0}^{\delta}\sup_{Q_n}\\\sqrt{1+\log\calN(\epsilon\|F\|,\calL_1,\|\cdot\|_{2,Q_n})}d\epsilon$ and $Q_n$ is  all finitely discrete probability measures on $\bbR^{m+d_x+2}$. For any $\epsilon>0$, based on Lemma 6 in \cite{jiao2023}, we have
\be    \label{maxineq}
&&J(\delta, \calL_{j,\delta}, F)\n\\
&=&\int_{0}^{\delta}\sup_{Q_n}\sqrt{1+\log\calN(\epsilon\|F\|,\calL_1,\|\cdot\|_{2,Q_n})}d\epsilon\n\\
&\leq&\int_{0}^{\delta}\sup_{Q_n}\sqrt{1+\log\calN_{[]}(2\epsilon\|F\|,\calL_1,\|\cdot\|_{2,Q_n})}d\epsilon\n\\
&\lesssim& \int_{0}^{\delta}\sqrt{1+\calS\calD\log(\calS)\log(B_0n/\epsilon)}d\epsilon\n\\
&\lesssim& \delta \sqrt{\mathcal{SD}\log(\calS)\log(B_0n/\delta)}\n,
\ee
where the last step comes from $\int_{0}^\delta \sqrt{\log(1/\epsilon)}d\epsilon\leq\delta\sqrt{\log(1/\delta)}$. Since $\delta\geq 1/n$, according to Assumption \ref{condDNNs}, we have 
\be    \label{pn-p}
&&\bbE\left\{\sup_{l\in\calL_{j,\delta}}|\sqrt{n}(\bbE_n-\bbE)(l)|\right\}\n\\
&\lesssim& \delta \sqrt{\calS\calD\log(\calS)\log(B_0n/\delta)}\|F\|+\frac{\|F_{\max}\|\delta^2\calS\calD\log(\calS)\log(B_0n/\delta)}{\delta^2\sqrt{n}} \n\\
&\lesssim& \sigma\sqrt{\calS\calD\log(\calS)\log(B_0n/\delta)}+\frac{c\calS\calD\log(\calS)\log(B_0)(\log (n))^2}{\sqrt{n}}\n\\
&=&o(m^{-1/2}).
\ee 
This bound ensures stochastic equicontinuity over shrinking $L^2$-neighbourhoods. Since $A_1\in\calL_{j,\delta}$, combining \eqref{eqpa1} and \eqref{pn-p}, we have $\bbE_n (A_1)=o_p(n^{-1/2}m^{-1/2})$.  For $A_2$, we have
\bse 
&& \frac{1}{n}\sum_{i=1}^n(Z_{ij}- \hat  f_j(\X_i))(\hat R_{A,i}(\hat R_{Y,i}- \beta_0\hat R_{A,i})-R_{A,i}( R_{Y,i}- \beta_0 R_{A,i}))      \\
&=& \frac{1}{n}\sum_{i=1}^n(Z_{ij}-  f_j(\X_i))(\hat R_{A,i}(\hat R_{Y,i}- \beta_0\hat R_{A,i})-R_{A,i}( R_{Y,i}- \beta_0 R_{A,i}))      \\
&&+  \frac{1}{n}\sum_{i=1}^n( f_j(\X_i)- \hat  f_j(\X_i))(\hat R_{A,i}(\hat R_{Y,i}- \beta_0\hat R_{A,i})-R_{A,i}( R_{Y,i}- \beta_0 R_{A,i}))      \\ 
&=&  \frac{1}{n}\sum_{i=1}^n(Z_{ij}-  f_j(\X_i))(\hat R_{A,i}(\hat R_{Y,i}- \beta_0\hat R_{A,i})- R_{A,i}(\hat R_{Y,i}- \beta_0\hat R_{A,i}))     \\
&&+  \frac{1}{n}\sum_{i=1}^n(Z_{ij}-  f_j(\X_i))(R_{A,i}(\hat R_{Y,i}- \beta_0\hat  R_{A,i})- R_{A,i}( R_{Y,i}- \beta_0 R_{A,i})) \\
&&+  \frac{1}{n}\sum_{i=1}^n( f_j(\X_i)- \hat  f_j(\X_i))(\hat R_{A,i}(\hat R_{Y,i}- \beta_0\hat R_{A,i})-R_{A,i}( R_{Y,i}- \beta_0 R_{A,i}))      \\ 
&=&  \frac{1}{n}\sum_{i=1}^n(Z_{ij}-  f_j(\X_i))(\hat R_{A,i}- R_{A,i})(\hat R_{Y,i}-R_{Y,i}-\beta_0(\hat R_{A,i}-R_{A,i})) \\
&&+  \frac{1}{n}\sum_{i=1}^n(Z_{ij}-  f_j(\X_i))(\hat R_{A,i}- R_{A,i})(R_{Y,i}-\beta_0 R_{A,i}) \\
&&+ \frac{1}{n}\sum_{i=1}^n(Z_{ij}-  f_j(\X_i))R_{A,i}(\hat R_{Y,i}-R_{Y,i}-\beta_0(\hat R_{A,i}-R_{A,i})) \\
&&+  \frac{1}{n}\sum_{i=1}^n( f_j(\X_i)- \hat  f_j(\X_i))(\hat R_{A,i}(\hat R_{Y,i}- \beta_0\hat R_{A,i})-R_{A,i}( R_{Y,i}- \beta_0 R_{A,i}))      \\ 
&&\equiv \bbE_n(A_{21}+A_{22}+A_{23}+A_{24}).
\ese 
For $A_{21}$, it contains a random error $Z_{ij}-  f_j(\X_i)$ and two residual term of DNNs $R_A-\hat R_A$ and $R_Y-\hat R_Y$. Similar to the proof of $A_1$, we have
$$
(\bbE_n-\bbE)\{(Z_j-  f_j(\X))(\hat R_{A}- R_{A})(\hat R_{Y}-R_{Y}-\beta_0(\hat R_{A}-R_{A}))\}=o_p(n^{-1/2}m^{-1/2}).
$$
According to Lemma \ref{lemmadnn} and Assumption \ref{condDNNs}, we have
\be \label{err}
&& \bbE(Z_j- f_j(\X))(\hat R_{A}- R_{A})(\hat R_{Y}-R_{Y}-\beta_0(\hat R_{A}-R_{A}))  \n\\
&=& \bbE(Z_j- f_j(\X))(\hat h_1(\Z,\X)-h_1(\Z,\X))(\hat h_2(\Z,\X)-h_2(\Z,\X))  \n\\
&&-\bbE(Z_j- f_j(\X))\beta_0(\hat h_1(\Z,\X)-h_1(\Z,\X))^2 \n\\
&\leq& c\bbE\left|(\hat h_1(\Z,\X)-h_1(\Z,\X))(\hat h(\Z,\X)-h_2(\Z,\X))\right|\n\\
&&+\beta_0c\bbE(\hat h_1(\Z,\X)-h_1(\Z,\X))^2\n\\
&\leq& c\sqrt{\bbE(\hat h_1(\Z,\X)-h_1(\Z,\X))^2}\sqrt{\bbE(\hat h_2(\Z,\X)-h_2(\Z,\X))^2}\n\\
&&+\beta_0\bbE(\hat h_1(\Z,\X)-h_1(\Z,\X))^2\n\\
&\lesssim& c B_0^5(\log (n))^5 \frac{1}{n} \mathcal{S D} \log (\mathcal{S})+324 B_0^2(\lfloor\zeta\rfloor+1)^4 d_x^{2\lfloor\zeta\rfloor+\zeta \vee 1}(N M)^{-4 \zeta / d_x}\n\\
&=& o_p(n^{-1/2}m^{-1/2}).
\ee
Thus $\bbE_n A_{21}=(\bbE_n-\bbE)A_{21}+\bbE A_{21}=o_p(n^{-1/2}m^{-1/2})$.
For $A_{22}$ and $A_{23}$, similarly to the proof of 
\eqref{pn-p} and \eqref{err}, we can prove that
$$
(\bbE_n-\bbE)\{(Z_j-  f_j(\X))(\hat R_A- R_A)(R_Y-\beta_0 R_A)\} =o_p(n^{-1/2}m^{-1/2}),
$$
and 
$$
(\bbE_n-\bbE)\{(Z_j-  f_j(\X))R_A(\hat R_Y-R_Y-\beta_0(\hat R_A-R_A))\}=o_p(n^{-1/2}m^{-1/2}).
$$
Furthermore, we also have
\bse
&&\bbE\left[(Z_j-  f_j(\X))(\hat R_A- R_A)(R_Y-\beta_0 R_A)\right]\\
&=&\bbE\left[\bbE\left\{(Z_j-  f_j(\X))(h_1(\Z,\X)-\hat h_1(\Z,\X))(R_Y-\beta_0 R_A)|\Z,\X\right\}\right]\\
&=&\bbE\left[(Z_j-  f_j(\X))(h_1(\Z,\X)-\hat h_1(\Z,\X))\bbE\left\{(Y-h_2(\Z,\X))-\beta_0 (A-h_1(\Z,\X))|\Z,\X\right\}\right]\\
&=&0,
\ese
and similarly
$$
\bbE(Z_j-  f_j(\X))R_A(\hat R_Y-R_Y-\beta_0(\hat R_A-R_A))=0.
$$
Thus $\bbE_nA_{22},\bbE_nA_{23}=o_p(n^{-1/2}m^{-1/2})$.  For $A_{24}$, the proof is similar to \eqref{err} since they all contain two residuals of DNNs and a noise term. 

Finally, similar to the proof of $A_1$, we can constrain the order of $A_3$:
$$
(\bbE_n - \bbE)\{(Z_j- f_j(\X))(h_3(\X)-\beta_0 h_4(\X)-(\hat h_3(\X)-\beta_0\hat h_4(\X)))\} =o_p(n^{-1/2}m^{-1/2}) ,
$$
and
\bse
&&\bbE\left\{(Z_j- f_j(\X))(h_3(\X)-\beta_0 h_4(\X)-(\hat h_3(\X)-\beta_0\hat h_4(\X)))\right\}\\
&=&\bbE\left[\bbE\left\{(Z_j- f_j(\X))\{h_3(\X)-\beta_0 h_4(\X)-(\hat h_3(\X)-\beta_0\hat h_4(\X))\}|\X\right\}\right]\\
&=&\bbE\left[\{h_3(\X)-\beta_0 h_4(\X)-(\hat h_3(\X)-\beta_0\hat h_4(\X))\}\bbE\left\{Z_j- f_j(\X)|\X\right\}\right]\\
&=& 0.
\ese 
So we can obtain $\bbE_nA_3=o_p(n^{-1/2}m^{-1/2})$. Until now, we have proved that all three term $\bbE_n A_1,\bbE_nA_2,\bbE_nA_3$ have order $o_p(n^{-1/2}m^{-1/2})$. Thus we obtain $g_j(\beta_0,\hat\boldeta;\O_i)-g_j(\beta_0,\boldeta_{0};\O_i)=o_p(n^{-1/2}m^{-1/2})$ for $j=1,\dots,m$, and it follows $\|\bar{g}(\beta_0, \hat{\boldeta})-\bar{g}(\beta_0, \boldeta_0)\|=o_p(m^{1/2}n^{-1/2}m^{-1/2})=o_p(n^{-1/2})$.

(ii). The origin term can be decomposed as
\bse 
&&\left\|\bar{g^{'}}(\hat{\boldeta})-\bar{g^{'}}\left(\boldeta_0\right)\right\|\\
&=& \left\|\frac{1}{n}\sum_{i=1}^n(\Z_i-\hat  f(\X_i))(\hat R_{A,i}^2-\hat h_4(\X_i))-\frac{1}{n}\sum_{i=1}^n(\Z_i-  f(\X_i))( R_{A,i}^2- h_4(\X_i))\right\|\\
&\leq& \left\|\frac{1}{n}\sum_{i=1}^n(\Z_i-\hat  f(\X_i))(\hat R_{A,i}^2-\hat h_4(\X_i)-R_{A,i}^2+h_4(\X_i))  \right\|\\
&&+\left\|\frac{1}{n}\sum_{i=1}^n( f(\X_i)-\hat  f(\X_i))(R_{A,i}^2-h_4(\X_i))  \right\|.
\ese
We can use similar techniques like those in (i), including the localization-based empirical process and the expectation term being zero, and the details are omitted here.

(iii). Notice that $g(\beta,\boldeta;\O_i)$ is a linear function of $\beta$, which means $\bar g(\beta,\boldeta)-\bar g(\beta_0,\boldeta)=\bar {g^{'}}(\boldeta)(\beta-\beta_0)$. So we have
\bse 
&&\sup_{\beta\in \calB}\left\|\bar{g}\left(\beta, \hat{\boldeta}\right)-\bar{g}\left(\beta, \boldeta_0\right)\right\|\\
&=& \sup_{\beta\in \calB}\left\|\bar{g}\left(\beta, \hat{\boldeta}\right)-\bar{g}\left(\beta_0, \hat{\boldeta}\right)+\bar{g}\left(\beta_0, \hat{\boldeta}\right)-\bar{g}\left(\beta_0, {\boldeta}_0\right)+\bar{g}\left(\beta_0, {\boldeta}_0\right)-\bar{g}\left(\beta, {\boldeta}_0\right) \right\|\\
&\leq&\sup_{\beta\in \calB}\left\|(\bar{g^{'}}(\hat\boldeta)-\bar{g^{'}}(\boldeta_0))(\beta-\beta_0)\right\|  + \left\|\bar{g}\left(\beta_0, \hat{\boldeta}\right)-\bar{g}\left(\beta_0, {\boldeta}_0\right) \right\|\\
&=&o_p\left(n^{-1 / 2}\right).
\ese

(iv). By Assumption \ref{condeigen}, we know that $n\bbE(\|\bar g(\beta_0,\boldeta_0)\|^2)/m\leq \tr(\Omega_0)/m\leq c$. Combining the Markov inequality, we can obtain that $\|\bar g(\beta_0,\boldeta_0)\|=O_p(\sqrt{m/n})$. Also by Lemma \ref{lemma1} (ii) and \ref{lemmag} (iii), we have 
\bse 
&&\sup_{\beta\in\calB}\|\bar g(\beta,\hat\boldeta)\|\\
&=&\sup_{\beta\in\calB}\|\bar g(\beta,\hat\boldeta)-\bar g(\beta,\boldeta_0)+\bar g(\beta,\boldeta_0)-\bar g(\beta_0,\boldeta_0)+\bar g(\beta_0,\boldeta_0)\|\\
&\leq &\sup_{\beta\in\calB}\|\bar g(\beta,\hat\boldeta)-\bar g(\beta,\boldeta_0)\|+\sup_{\beta\in\calB}\|\bar g(\beta,\boldeta_0)-\bar g(\beta_0,\boldeta_0)\|+\|\bar g(\beta_0,\boldeta_0)\|\\
&\leq& o_p(n^{-1/2})+O_p(\mu_n/\sqrt{n})+O_p(\sqrt{m/n})\\
&=&O_p(\mu_n/\sqrt{n}).
\ese
For $\sup_{\beta\in\calB}\|\bar g(\beta_0,\hat\boldeta)\|$, the second term in the above decomposition is zero, so we have $\sup_{\beta\in\calB}\|\bar g(\beta_0,\hat\boldeta)\|=O_p(\sqrt{m/n})$. This completes the proof of Lemma \ref{lemmag}.
\end{proof}

\begin{lemma} \label{lemmaomega}
Under Assumptions \ref{condnuispara}-\ref{condeigen}, we have

(i). $\left\|\bar{\Omega}\left(\beta_0, \hat{\boldeta}\right)-\Omega_0\right\|=O_p(\sqrt{m\log (m)/n})$;

(ii). $\left\|n^{-1} \sum_{i=1}^n g(\beta_0, \hat{\boldeta};\O_i) g^{'}(\hat{\boldeta};\O_i)\tp-\bbE\left\{g(\beta_0, \boldeta_0;\O) g^{'}(\boldeta_0;\O)\tp\right\}\right\|=O_p(\sqrt{m\log(m)/n})$;

(iii). $\left\|n^{-1} \sum_{i=1}^n g^{'}(\hat{\boldeta};\O_i) g^{'}(\hat{\boldeta};\O_i)\tp-\bbE\left\{g^{'}(\boldeta_0;\O) g^{'}(\boldeta_0;\O)\tp\right\}\right\|=O_p(\sqrt{m\log(m)/n})$;

(iv). $\sup _{\beta \in \calB}\left\|\bar{\Omega}(\beta, \hat{\boldeta})-\Omega_0\left(\beta, \boldeta_0\right)\right\|=O_p(\sqrt{m\log(m)/n})$;

(v). $\sup _{\beta \in \calB}\left\|n^{-1} \sum_{i=1}^n g(\beta, \hat{\boldeta};\O_i) g^{'}(\hat{\boldeta};\O_i)\tp-\bbE\left\{g(\beta, \boldeta_0;\O) g^{'}(\boldeta_0;\O)\tp\right\}\right\|=O_p(\sqrt{m\log(m)/n})$.
\end{lemma}
\begin{proof}
(i). The matrix spectral norm is smaller than the matrix trace:
\be
\left\|\bar{\Omega}\left(\beta_0, \hat{\boldeta}\right)-\Omega_0\right\|&\leq& \left\|\bar{\Omega}\left(\beta_0, \hat{\boldeta}\right)-\bar{\Omega}\left(\beta_0, {\boldeta}_0\right)\right\|+\left\|\bar{\Omega}\left(\beta_0, {\boldeta}_0\right)-\Omega_0\right\|\n\\
&\leq&\tr(\bar{\Omega}\left(\beta_0, \hat{\boldeta}\right)-\bar{\Omega}\left(\beta_0, {\boldeta}_0\right))+\left\|\bar{\Omega}\left(\beta_0, {\boldeta}_0\right)-\Omega_0\right\|\n\\
&=&\sum_{j=1}^{m}\frac{1}{n}\sum_{i=1}^n\{g_j(\beta_0,\hat\boldeta;\O_i)^2-g_j(\beta_0,\boldeta_0;\O_i)^2\}+\left\|\bar{\Omega}\left(\beta_0, {\boldeta}_0\right)-\Omega_0\right\|\n
\ee 
First, we will prove $1/n\sum_{i=1}^n\{g_j(\beta_0,\hat\boldeta;\O_i)^2-g_j(\beta_0,\boldeta_0;\O_i)^2\}=o_p(n^{-1/2})$.
\be
&&\frac{1}{n}\sum_{i=1}^n\{g_j(\beta_0,\hat\boldeta;\O_i)^2-g_j(\beta_0,\boldeta_0;\O_i)^2\}\n\\
&=&\frac{1}{n}\sum_{i=1}^n(Z_{ij}-\hat f_j(\X_i))^2(\hat\Delta_{i}-(\hat h_3(\X_i)-\beta_0\hat h_4(\X_i)))^2\n\\
&&\quad -\frac{1}{n}\sum_{i=1}^n(Z_{ij}- f_j(\X_i))^2(\Delta_{i}-( h_3(\X_i)-\beta_0 h_4(\X_i)))^2  \n\\
&=&\frac{1}{n}\sum_{i=1}^n(Z_{ij}-\hat f_j(\X_i))^2(\hat\Delta_{i}-(\hat h_3(\X_i)-\beta_0\hat h_4(\X_i)))^2\n\\
&&\quad -\frac{1}{n}\sum_{i=1}^n(Z_{ij}- \hat f_j(\X_i))^2(\Delta_{i}-( h_3(\X_i)-\beta_0 h_4(\X_i)))^2  \n\\
&& +\frac{1}{n}\sum_{i=1}^n(Z_{ij}- \hat f_j(\X_i))^2(\Delta_{i}-( h_3(\X_i)-\beta_0 h_4(\X_i)))^2  \n\\
&&\quad -\frac{1}{n}\sum_{i=1}^n(Z_{ij}- f_j(\X_i))^2(\Delta_{i}-( h_3(\X_i)-\beta_0 h_4(\X_i)))^2  \n\\
&=&\frac{1}{n}\sum_{i=1}^n(f_j(\X_i)-\hat f_j(\X_i))(2Z_{ij}-f_j(\X_i)-\hat f_j(\X_i))(\Delta_{i}-(h_3(\X_i)-\beta_0 h_4(\X_i)))^2\n\\
&&+\frac{1}{n}\sum_{i=1}^n(Z_{ij}- \hat f_j(\X_i))^2(\hat \Delta_{i}-( \hat h_3(\X_i)-\beta_0 \hat h_4(\X_i))+\Delta_{i}-( h_3(\X_i)-\beta_0 h_4(\X_i)))\n\\
&& \quad \times (\hat\Delta_{i}-\Delta_{i}+(h_3(\X_i)-\hat h_3(\X_i))+\beta_0 (\hat h_4(\X_i)- h_4(\X_i)))\n\\
&\equiv& B_1+B_2.\n
\ee
For $B_1$, we have 
\be 
&&B_1\n\\
&=&\frac{1}{n}\sum_{i=1}^n(f_j(\X_i)-\hat f_j(\X_i))(2Z_{ij}-2f_j(\X_i)+f_j(\X_i)-\hat f_j(\X_i))(\Delta_{i}-(h_3(\X_i)-\beta_0 h_4(\X_i)))^2\n\\
&=&\frac{1}{n}\sum_{i=1}^n2(f_j(\X_i)-\hat f_j(\X_i))(Z_{ij}-f_j(\X_i))(\Delta_{i}-(h_3(\X_i)-\beta_0 h_4(\X_i)))^2\n\\
&&+\frac{1}{n}\sum_{i=1}^n(f_j(\X_i)-\hat f_j(\X_i))^2(\Delta_{i}-(h_3(\X_i)-\beta_0 h_4(\X_i)))^2\n\\
&\equiv&  B_{1,1}+B_{1,2}.\n
\ee
We can prove $B_{1,1}=o_p(n^{-1/2}m^{-1/2})$ similar to the proof of \eqref{pn-p}, and 
$$
B_{1,2}\lesssim c^2\frac{1}{n}\sum_{i=1}^n(f_j(\X_i)-\hat f_j(\X_i))^2=o_p(n^{-1/2}m^{-1/2}).
$$
Thus we obtain $B_1=o_p(n^{-1/2}m^{-1/2})$. Similar proofs can be applied to $B_2$, but it requires breaking down $B_2$ into several terms, just like the decomposition of $A_2$. Since the process is repeating, the proof details are omitted here. Therefore, we obtain 
\bse
\left\|\bar{\Omega}\left(\beta_0, \hat{\boldeta}\right)-\bar{\Omega}\left(\beta_0, {\boldeta}_0\right)\right\|&\leq&\sum_{j=1}^{m}\frac{1}{n}\sum_{i=1}^n\{g_j(\beta_0,\hat\boldeta;\O_i)^2-g_j(\beta_0,\boldeta_0;\O_i)^2\}\\
&=&mo_p(n^{-1/2}m^{-1/2})=o_p(\sqrt{m/n}).
\ese
For $\left\|\bar{\Omega}\left(\beta_0, {\boldeta}_0\right)-\Omega_0\right\|$, applying the matrix concentration inequality (Theorem 1 in \cite{tropp2016}, Lemma 6 in \cite{ye2024}), we obtain $\left\|\bar{\Omega}\left(\beta_0, {\boldeta}_0\right)-\Omega_0\right\|=O_p(\sqrt{m\log(m)/n})$. Since $O_p(\sqrt{m\log(m)/n})$ dominates $o_p(\sqrt{m/n})$, we obtain $\left\|\bar{\Omega}\left(\beta_0, \hat{\boldeta}\right)-\Omega_0\right\|=o_p(\sqrt{m/n})+O_p(\sqrt{m\log (m)/n})=O_p(\sqrt{m\log (m)/n})$, and thus (i) is proved.

(ii) and (iii). The proof is similar to the proof of (i) and is omitted here.

(iv). According to Lemma \ref{lemmag}, Lemma \ref{lemmaomega} and Taylor expansion, we have
\be 
&&\sup_{\beta \in \calB}\left\|\bar{\Omega}(\beta, \hat{\boldeta})-\Omega_0\left(\beta, \boldeta_0\right)\right\|\n\\
&=&\sup_{\beta \in \calB}\left\|\bar{\Omega}(\beta, \hat{\boldeta})-\bar\Omega(\beta_0,\hat\boldeta)+\bar\Omega(\beta_0,\hat\boldeta)-\Omega(\beta_0,\boldeta_0)+\Omega(\beta_0,\boldeta_0)-\Omega\left(\beta, \boldeta_0\right)\right\|\n\\
&\leq& \sup_{\beta \in \calB}\Big\|\frac{2}{n}\sum_{i=1}^ng(\beta_0,\hat\boldeta;\O_i)g^{'}(\hat\boldeta;\O_i)\tp(\beta-\beta_0)+ \frac{1}{n}\sum_{i=1}^ng^{'}(\hat\boldeta;\O_i)g^{'}(\hat\boldeta;\O_i)\tp(\beta-\beta_0)^2\n\\
&&-2\bbE \{g(\beta_0,\boldeta_0;\O_i)g^{'}(\boldeta_0;\O_i)\tp\}(\beta-\beta_0)- \bbE\{ g^{'}(\boldeta_0;\O_i)g^{'}(\boldeta_0;\O_i)\tp\}(\beta-\beta_0)^2 \Big\|\n\\
&&+\|\bar\Omega(\beta_0,\hat\boldeta)-\Omega(\beta_0,\boldeta_0)\|\n\\
&\leq& 2\sup_{\beta \in \calB}|\beta-\beta_0|\left\|\frac{1}{n}\sum_{i=1}^ng(\beta_0,\hat\boldeta;\O_i)g^{'}(\hat\boldeta;\O_i)\tp-\bbE\{ g(\beta_0,\boldeta_0;\O_i)g^{'}(\boldeta_0;\O_i)\tp\}\right\|\n\\
&&+\sup_{\beta \in \calB}|\beta-\beta_0|^2\left\|\frac{1}{n}\sum_{i=1}^ng^{'}(\boldeta_0;\O_i)g^{'}(\hat\boldeta;\O_i)\tp-\bbE \{g^{'}(\boldeta_0;\O_i)g^{'}(\boldeta_0;\O_i)\tp\}\right\| \n\\
&&+\|\bar\Omega(\beta_0,\hat\boldeta)-\Omega(\beta_0,\boldeta_0)\|\n\\
&=&O_p(\sqrt{m\log(m)/n}).\n
\ee 
A useful implication of this conclusion here is $\sup_{\beta\in\calB}\|\bar \Omega(\beta,\hat\boldeta)\|\leq\sup_{\beta\in\calB}\|\bar \Omega(\beta,\hat\boldeta)-\Omega(\beta,\boldeta_0)\|+\sup_{\beta\in\calB}\|\Omega(\beta,\boldeta_0)\|=O_p(1) $.

(v). The proof is similar to the proof of (iv) and is omitted here.
\end{proof}
Define $P(\blambda,\beta,\boldeta)=1/n\sum_{i=1}^n \rho(\blambda\tp g(\beta,\boldeta;\O_i))$, $\blambda(\beta,\hat\boldeta)=\argmax_{\blambda \in {L}(\beta,\hat\boldeta)} P(\blambda,\beta,\hat\boldeta)$ and $\blambda(\beta,\boldeta_0)=\argmax_{\blambda \in {L}(\beta,\boldeta_0)}P(\blambda,\beta,\boldeta_0)$.
\begin{lemma} \label{lemmalambda}
Under Assumptions \ref{condnuispara}-\ref{condgel}, we have 

(i) $\sup_{\beta\in \calB}\|\blambda(\beta,\hat\boldeta)\|=O_p(\mu_n/\sqrt{n})$, $\sup_{\beta\in \calB}\|\blambda(\beta,\boldeta_0)\|=O_p(\mu_n/\sqrt{n})$, $\|\blambda(\beta_0,\boldeta_0)\|=O_p(\sqrt{m/n})$.

(ii) $\|\blambda(\beta_0,\hat\boldeta)- \blambda(\beta_0,\boldeta_0)\|  =o_p(1/\sqrt{n})$, and $\sup_{\beta\in\calB}\|\blambda(\beta,\hat\boldeta)- \blambda(\beta,\boldeta_0)\|  =o_p(\mu_n/\sqrt{n})$.  
\end{lemma}
\begin{proof}
(i) Let $\tau=\max_{i\leq n}\sup_{\beta\in\calB}\|g(\beta,\boldeta_0;\O_i)\|$. A standard result gives that $\tau=O_p(n^{1/\gamma}\\\bbE\{\sup_{\beta\in\calB}\|g(\beta,\boldeta_0;\O)\|^{\gamma}\}^{1/\gamma})$ for $\gamma>2$. Also we can prove that 
\bse
&&\bbE\{\sup_{\beta\in\calB}\|g(\beta,\boldeta_0;\O)\|^{\gamma}\}^{1/\gamma}-\bbE\{\sup_{\beta\in\calB}\|g(\beta,\hat\boldeta;\O)\|^{\gamma}\}^{1/\gamma}\\
&\leq& \bbE\{\sup_{\beta\in\calB}\|g(\beta,\boldeta_0;\O)-g(\beta,\hat\boldeta;\O)\|^{\gamma}\}^{1/\gamma}   \\
&=& \bbE\{\|D_{\boldeta} g(\beta,\boldeta_0;\O) \|^{\gamma}\}^{1/\gamma}\|\hat\boldeta-\boldeta_0\|\\
&=& O_p(\|\hat\boldeta-\boldeta_0\|).
\ese
Given $\gamma>4$, by Assumption \ref{condmoment} and $\|\hat\boldeta-\boldeta_0\|=o_p(n^{-1/4})$ from Lemma \ref{lemmadnn}, we can obtain 
\bse 
&&n^{1/\gamma}\bbE\{\sup_{\beta\in\calB}\|g(\beta,\hat\boldeta;\O)\|^{\gamma}\}^{1/\gamma}\frac{m+\mu_n}{\sqrt{n}}\\
&=&n^{1/\gamma}[\bbE\{\sup_{\beta\in\calB}\|g(\beta,\hat\boldeta;\O)\|^{\gamma}\}^{1/\gamma}-\bbE\{\sup_{\beta\in\calB}\|g(\beta,\boldeta_0;\O)\|^{\gamma}\}^{1/\gamma}]\frac{m+\mu_n}{\sqrt{n}}  \\
&&+ n^{1/\gamma}\bbE\{\sup_{\beta\in\calB}\|g(\beta,\boldeta_0;\O)\|^{\gamma}\}^{1/\gamma}\frac{m+\mu_n}{\sqrt{n}}  \\
&=& n^{1/\gamma}o_p(n^{-1/4})\frac{m+\mu_n}{\sqrt{n}}  +o_p(1)\\
&=&o_p(1).
\ese
According to Assumption \ref{condmoment}, there exists $\tau_n$ such that $\tau_n=o(\tau^{-1})$ and $\mu_n/\sqrt{n}=o(\tau_n)$. Let $L_n=\{\blambda: \|\blambda\|\leq\tau_n\}$. For $\blambda\in L_n$, we have
\be\label{suplambdag}
\sup_{1\leq i\leq n,  \beta\in \calB}|\blambda\tp g(\beta,\hat\boldeta;\O_i)|\leq \tau_n\tau=o_p(1).
\ee
So $ L_n\subset L(\beta,\hat\boldeta) $ and there exists a positive constant $c$ such that $-c^{-1}\leq\rho^{''}(\blambda\tp g(\beta,\hat\boldeta;\O_i))\leq -c$ and $|\rho^{'''}(\blambda g(\beta,\hat\boldeta;\O_i))|\leq c$. Using Taylor expansion around $\blambda=0$, we have
$$
\rho(\blambda\tp g(\beta,\hat\boldeta;\O_i))=\rho(0)+\rho^{'}(0)\blambda\tp g(\beta,\hat\boldeta;\O_i)+\frac{1}{2}\rho^{''}(\bar x)\blambda\tp g(\beta,\hat\boldeta;\O_i)g(\beta,\hat\boldeta;\O_i)\tp\blambda,
$$
where $\bar x$ lies between 0 and $\blambda(\beta,\hat\boldeta)\tp g(\beta,\hat\boldeta;\O_i).$ Because $\rho(0)=0$, $\rho^{'}(0)=-1$, $P(\blambda,\beta,\hat\boldeta)$ becomes
$$
P(\blambda,\beta,\hat\boldeta)=- \blambda\tp\bar g(\beta,\hat\boldeta)+\frac{1}{2}\rho^{''}(\bar x)\blambda\tp \bar\Omega(\beta,\hat\boldeta)\blambda.
$$
Let $\blambda(\beta,\hat\boldeta)=\argmax_{\blambda\in L_n} P(\blambda,\beta,\hat\boldeta)$. Based on Assumption \ref{condeigen} and $\rho^{''}(\blambda\tp g(\beta,\hat\boldeta;\O_i))>-c^{-1}$, we have 
$$
0=P(0,\beta,\hat\boldeta)\leq P(\blambda(\beta,\hat\boldeta),\beta,\hat\boldeta)\leq\|\blambda(\beta,\hat\boldeta)\|\|\bar g(\beta,\hat\boldeta)\|-c\|\blambda(\beta,\hat\boldeta)\|^2,
$$
which implies $\|\blambda(\beta,\hat\boldeta)\|\leq c\|\bar g(\beta,\hat\boldeta)\|$. Then Lemma \ref{lemmag} gives that 
\be\label{eqlambda}
\sup_{\beta\in \calB}\|\blambda(\beta,\hat\boldeta)\|\leq c\sup_{\beta\in \calB}\|\bar g(\beta,\hat\boldeta)\|=O_p(\frac{\mu_n}{\sqrt{n}})=o_p(\tau_n).
\ee
Thus $\blambda(\beta,\hat\boldeta)\in L_n$. Since $\blambda(\beta,\hat\boldeta)$ is a local maximum of $P(\blambda,\beta,\hat\boldeta)$ and $P$ is convex function, $\blambda(\beta,\hat\boldeta)$ is also the global maximum value, i.e., $\blambda(\beta,\hat\boldeta)=\argmax_{\blambda \in L(\beta,\hat\boldeta)} P(\blambda,\beta,\hat\boldeta)$.

Similarly, we can prove that $\sup_{\beta\in \calB}\|\blambda(\beta,\boldeta_0)\|=O_p(\mu_n/\sqrt{n})$. Moreover, because $\bar g(\beta_0,\eta_0)=O_p(\sqrt{{m}/{n}})$, \eqref{eqlambda} yields $\|\blambda(\beta_0,\boldeta_0)\|=O_p(\sqrt{{m}/{n}})$.

(ii) Since $\blambda(\beta_0,\hat\boldeta)=\argmax_{\blambda\in L(\beta_0,\hat\boldeta)}1/n\sum_{i=1}^n\rho(\blambda\tp g(\beta_0,\hat\boldeta;\O_i))$, its first order derivative is 0:
\be \label{rho1}
\frac{1}{n}\sum_{i=1}^n \rho^{'}(\blambda(\beta_0,\hat\boldeta)\tp g(\beta,\hat\boldeta;\O_i))g(\beta,\hat\boldeta;\O_i)=0.
\ee
Applying Taylor expansion on $\rho^{'}$ around 0, we have
$$
\rho^{'}(\blambda(\beta_0,\hat\boldeta)\tp g(\beta_0,\hat\boldeta;\O_i))=\rho^{'}(0)+\rho^{''}(0)\blambda(\beta_0,\hat\boldeta)\tp g(\beta_0,\hat\boldeta;\O_i)+1/2\rho^{'''}(\bar x)(\blambda(\beta_0,\hat\boldeta)\tp g(\beta_0,\hat\boldeta;\O_i))^2 ,
$$
where $\bar x$ lies between 0 and $\blambda(\beta_0,\hat\boldeta)$. So \eqref{rho1} becomes 
\be \label{Taylor_lambda}
-\bar g(\beta_0,\hat\boldeta)-\bar\Omega(\beta_0,\widehat\boldeta)\blambda(\beta_0,\hat\boldeta)+R=0,
\ee
where $R=1/(2n)\sum_{i=1}^n\rho^{'''}(\bar x)(\blambda(\beta_0,\hat\boldeta)\tp g(\beta_0,\hat\boldeta;\O_i))^2g(\beta_0,\hat\boldeta;\O_i)$. According to Assumptions \ref{condeigen} and \ref{condgel}, we have
\be\label{R-order}
&&\|R\|\n\\
&\leq& c\sup_{1\leq i\leq n}\|g(\beta_0,\hat\boldeta;\O_i)\|\|\blambda(\beta_0,\hat\boldeta)\|^2\|\bar\Omega(\beta_0,\hat\boldeta)\|\n\\
&=&O_p(n^{1/\gamma}\{\bbE(\|g(\beta_0,\hat\boldeta;\O)\|^\gamma)\}^{1/\gamma})O_p(\mu_n^2/n)\n\\
&=&o_p(\mu_n/\sqrt{n})
\ee
is a small order term, which means 
\be \label{lambda1}
\blambda(\beta_0,\hat\boldeta)=-\bar\Omega(\beta_0,\hat\boldeta)^{-1}\bar g(\beta_0,\hat\boldeta)+o_p(\mu_n/\sqrt{n}),
\ee
and similarly $\blambda(\beta_0,\boldeta_0)=-\bar\Omega(\beta_0,\boldeta_0)^{-1}\bar g(\beta_0,\boldeta_0)+o_p(\mu_n/\sqrt{n})$. According to Lemma \ref{lemmag} and Lemma \ref{lemmaomega}, we have
\be  \label{lambda-lambda}
&&\|\blambda(\beta_0,\hat\boldeta)-\blambda(\beta_0,\boldeta_0)\|\n\\
&\leq & \| \bar\Omega(\beta_0,\boldeta_0)^{-1}\bar g(\beta_0,\boldeta_0)-\bar\Omega(\beta_0,\hat\boldeta)^{-1}\bar g(\beta_0,\hat\boldeta)\|+o_p(\mu_n/\sqrt{n})\n\\
&\leq& \|(\bar\Omega(\beta_0,\boldeta_0)^{-1}-\bar\Omega(\beta_0,\hat\boldeta)^{-1})\bar g(\beta_0,\boldeta_0)\|\n\\
&&+\|\bar\Omega(\beta_0,\hat\boldeta)^{-1}(\bar g(\beta_0,\boldeta_0)-\bar g(\beta_0,\hat\boldeta))\|+o_p(\mu_n/\sqrt{n})\n\\
&=& O_p(\sqrt{m\log (m)/n})O_p(\sqrt{m/n})+O_p(1)o_p(\mu_n/\sqrt{n})+o_p(\mu_n/\sqrt{n})\n\\
&=&o_p(\mu_n/\sqrt{n}).
\ee
The last step comes from that $m^3/n\to 0$ and $m/\mu_n^2 $ is bounded.
Similarly, by replacing $\blambda(\beta_0,\hat\boldeta)$ with $\blambda(\beta,\hat\boldeta)$ in \eqref{lambda-lambda}, we can obtian $\sup_{\beta \in \calB}|\blambda(\beta,\hat\boldeta)-\blambda(\beta,\boldeta_0)|=o_p(\sqrt{\mu_n/n})$.
\end{proof}

\begin{lemma} \label{lemmaq}
Let $\tilde Q(\beta,\boldeta) = \tilde g(\beta,\boldeta)\tp \Omega(\beta,\boldeta)^{-1}\tilde g(\beta,\boldeta)\tp/2+m/(2n)$. Under Assumptions \ref{condnuispara}-\ref{condgel}, $$\sup_{\beta\in \calB}|\tilde Q(\beta,\boldeta_0)-\hat Q(\beta,\hat\boldeta)|=o_p(\frac{\mu_n^2}{n}).$$
\end{lemma}
\begin{proof}
According to the Taylor expansion at 0, we have
\be \label{qhat}
\hat Q(\beta,\hat\boldeta) &=& \frac{1}{n}\sum_{i=1}^n \rho( \blambda(\beta,\hat\boldeta)\tp g(\beta,\hat\boldeta;\O_i))\n\\
&=&\frac{1}{n}\sum_{i=1}^n \rho(0)+\frac{1}{n}\sum_{i=1}^n \rho^{'}(0) \blambda(\beta,\hat\boldeta)\tp g(\beta,\hat\boldeta;\O_i)\n\\
&&+\quad \frac{1}{2n}\sum_{i=1}^n \rho^{''}(0) \blambda(\beta,\hat\boldeta)\tp g(\beta,\hat\boldeta;\O_i)g(\beta,\hat\boldeta;\O_i)\tp \blambda(\beta,\hat\boldeta)+r,
\ee
where $r=1/(6n)\sum_{i=1}^n\rho^{'''}(\bar x)( \blambda(\beta,\hat\boldeta)\tp g(\beta,\hat\boldeta;\O_i))^3$ is a remainder term and $\bar x$ lies between $0$ and $ \blambda(\beta,\hat\boldeta)\tp g(\beta,\hat\boldeta;\O_i)$. It can be proved that 
\be\label{rorder}
r&=&\frac{1}{6n}\sum_{i=1}^n\rho^{'''}(\bar x)( \blambda(\beta,\hat\boldeta)\tp g(\beta,\hat\boldeta;\O_i))^3\n\\
&\leq&c\|\blambda(\beta,\hat\boldeta)\|\max_{i}\|g(\beta,\hat\boldeta;\O_i)\|\blambda(\beta,\hat\boldeta)\tp\bar\Omega(\beta,\hat\boldeta)\blambda(\beta,\hat\boldeta)   \n\\
&=&O_p(\frac{\mu_n}{\sqrt{n}}\tau_n)O_p(\mu_n^2/n)\n\\
&=&o_p(\mu_n^2/n).
\ee
Based on \eqref{R-order} and \eqref{rorder}, \eqref{qhat} becomes
\be
&&\hat Q(\beta,\hat\boldeta) \n\\
&=& -\frac{1}{n}\sum_{i=1}^n  (-\bar g(\beta,\hat\boldeta)+R)\tp\bar\Omega(\beta,\hat\boldeta)^{-1} g(\beta,\hat\boldeta;\O_i)\n\\
&&- \frac{1}{2n}\sum_{i=1}^n(-\bar g(\beta,\hat\boldeta)+R)\tp\bar\Omega(\beta,\hat\boldeta)^{-1} g(\beta,\hat\boldeta;\O_i)g(\beta,\hat\boldeta;\O_i)\tp\bar\Omega(\beta,\hat\boldeta)^{-1}(-\bar g(\beta,\hat\boldeta)+R)+r\n\\
&=& \frac{1}{2} \bar  g(\beta,\hat\boldeta)\tp\bar\Omega(\beta,\hat\boldeta)^{-1} \bar g(\beta,\hat\boldeta) +o_p(\mu_n^2/n).
\ee
Lemma 7 in \cite{ye2024} gives that $\sup_{\beta\in\calB}|1/2g(\beta,\hat\boldeta)\tp\bar\Omega(\beta,\hat\boldeta)^{-1} \bar g(\beta,\hat\boldeta)-\tilde Q(\beta,\boldeta_0)|=o_p(\mu_n^2/n)$, thus we have
$$
\sup_{\beta\in\calB}|\hat Q(\beta,\hat\boldeta)-\tilde Q(\beta,\boldeta_0)  |=\sup_{\beta\in\calB}|\frac{1}{2}g(\beta,\hat\boldeta)\tp\bar\Omega(\beta,\hat\boldeta)^{-1} \bar g(\beta,\hat\boldeta)-\tilde Q(\beta,\boldeta_0)|+o_p(\mu_n^2/n)=o_p(\mu_n^2/n).
$$
\end{proof}

\begin{lemma} \label{lemmapsi}
Under Assumptions \ref{condnuispara}-\ref{condeigen} and \ref{condkernel}, 
(i). $\left\|\bar{\psi}\left(\beta_0,G_0, \hat{\boldeta}\right)-\bar{\psi}\left(\beta_0,G_0, \boldeta_0\right)\right\|=o_p(n^{-1/2})$.

(ii). $\left\|\bar{\psi^{'}}(G_0,\hat{\boldeta})-\bar{\psi^{'}}\left(G_0,\boldeta_0\right)\right\|=o_p(n^{-1/2})$.

(iii). $\sup _{\beta \in \calB}\left\|\bar{\psi}(\beta,G_0, \hat{\boldeta})-\bar{\psi}\left(\beta,G_0, \boldeta_0\right)\right\|=o_p(n^{-1/2})$.   

(iv). $\sup_{\beta\in\calB}\|\bar \psi(\beta,G_0,\hat\boldeta)\|=O_p(\mu_n/\sqrt{n})$ and $\|\bar \psi(\beta_0,G_0,\hat\boldeta)\|=O_p(\sqrt{m/n})$.     
\end{lemma}
\begin{proof}
(i). Define $\Phi$ as follows:
\be
&&\Phi(\O,t)\n\\
&=& {\psi}\left(\beta_0,G_0, \boldeta+t(\hat{\boldeta}-\boldeta_0)\right) \n\\
&=& \frac{\delta}{G_0(Y|X)}[\Z-\{f+t(\hat f-f)\}]\Big([A-\{h_1+t(\hat h_1-h_1)\}][Y-\{h_2+t(\hat h_2-h_2)\}\n\\
&&\quad -\beta_0(A-\{h_1+t(\hat h_1-h_1)\})]-[h_3+t(\hat h_3-h_3)-\beta_0\{h_4+t(\hat h_4-h_4)\}]\Big)\n\\
&&+(1-\frac{\delta}{G_0(Y|X)})[\xi_0(\beta_0;\O_{A})+t\{\hat\xi(\beta_0;\O_{A})-\xi_0(\beta_0,\hat\boldeta;\O_{A})\}].\n
\ee
Then $\bar{\psi}\left(\beta_0,G_0, \hat{\boldeta}\right)-\bar{\psi}\left(\beta_0,G_0, \boldeta_0\right)= 1/n\sum_{i=1}^n \{\Phi(1;\O_i)-\Phi(0;\O_i)\}=1/n\sum_{i=1}^n \Phi^{'}(0;\O_i)+r$, where $r=O_p(\|\hat\boldeta -\boldeta_0\|^2)=o_p(n^{-1/2})$ is the second order term. Moreover,
\be
&&\Phi^{'}(0;\O)\n\\
&=&\frac{\delta}{G_0(Y|\O_A)}\Big((\hat f-f)[(A-h_1)\{Y-h_2-\beta_0(A-h_1)\}-h_3+\beta_0 h_4]\Big)\n\\
&&+ \frac{\delta}{G_0(Y|\O_A)}(\Z-f(\X))(A-h_1)\{\beta_0(\hat h_1-h_1)-(\hat h_2 -h_2)\}\n \\
&&+ \frac{\delta}{G_0(Y|\O_A)}(\Z-f(\X))(\hat h_1- h_1)\{Y-h_2-\beta_0(A-h_1)\}\n\\
&&+ (1-\frac{\delta}{G_0(Y|\O_A)})\{\hat\xi(\beta_0;\O_{A})-\xi_0(\beta_0;\O_{A})\}\n\\
&\equiv&\Phi_1+\Phi_2+\Phi_3+\Phi_4.\n
\ee
For $\Phi_1$, the empirical process gives that $(\bbE_n-\bbE)(\Phi_i)=o_p(n^{-1/2})$ for $i=1,2,3,4$, which can be proved similarly to that of Lemma \ref{lemmag}. Besides, we have
\be
&&\bbE \Phi_1=\bbE(\bbE(\Phi_1|\O_T))\n\\
&=&\bbE\left\{\frac{\bbE(\delta|\O_T)}{G_0(T|\O_A)}\Big((\hat f-f)[(A-h_1)\{Y-h_2-\beta_0(A-h_1)\}-h_3+\beta_0 h_4]\Big)\right\}\n\\
&=&\bbE\Big((\hat f-f)[(A-h_1)\{Y-h_2-\beta_0(A-h_1)\}-h_3+\beta_0 h_4] \Big)\n\\
&=&\bbE\left\{\bbE\Big((\hat f-f)[(A-h_1)\{Y-h_2-\beta_0(A-h_1)\}-h_3+\beta_0 h_4]|\X \Big)\right\}\n\\
&=&\bbE\left\{(\hat f-f)\bbE\{R_AR_Y-h_3-\beta_0(R_AR_A-h_4)|\X \}  \right\}\n\\
&=&0.\n
\ee
Thus, we obtain $\bbE \Phi_1=o_p(n^{-1/2})$. Similarly, by taking the conditional expectation concerning $\O_T$, $\delta / G_0$ can be eliminated, and then by taking the conditional expectation concerning $\X$, we can obtain $\bbE \Phi_2 =  \bbE \Phi_3 =\bbE \Phi_4=0$. Consequently, we have $P_n \Phi^{'}(0;\O) = o_p(n^{-1/2})$, and $\left\|\bar{\psi}\left(\beta_0,G_0, \hat{\boldeta}\right)-\bar{\psi}\left(\beta_0,G_0, \boldeta_0\right)\right\|= \|1/n\sum_{i=1}^n \Phi^{'}(0;\O_i)+r\|=o_p(n^{-1/2})$. In fact, a similar proof as that in Lemma \ref{lemmag} can be applied by decomposing $\bar{\psi}\left(\beta_0,G_0, \hat{\boldeta}\right)-\bar{\psi}\left(\beta_0,G_0, \boldeta_0\right)$ into several parts. Each part can be controlled using empirical process techniques. The terms corresponding to $P=0$ or higher-order remainders are $o_p(1)$, and hence the same conclusion follows.

(ii). Define $\xi^{'}(\beta;\O_A) = \partial\xi(\beta;\O_A)/\partial \beta$ and $\Phi$ as follows:
\be
&&\Phi(t;\O)\n\\
&=&\psi^{'}(G_0,\boldeta_0+t(\hat{\boldeta}-\boldeta_0);\O)\n\\
&=& \frac{\delta}{G_0(Y|\O_A)}[\Z-\{f+t(\hat f-f)\}]\left(h_4+t(\hat h_4-h_4)-[A-\{h_1+t(\hat h_1-h_1)\}]^2\right)\n\\
&&+(1-\frac{\delta}{G_0(Y|\O_A)})[\xi_0^{'}(\beta_0;\O_{A})+t\{\hat\xi^{'}(\beta_0;\O_{A})-\xi_0^{'}(\beta_0;\O_{A})\}].\n
\ee
Then we can obtain $\left\|\bar{\psi^{'}}(G_0,\hat{\boldeta})-\bar{\psi^{'}}\left(G_0,\boldeta_0\right)\right\|=1/n \sum_{i=1}^n\{\Phi(1;\O_i)-\Phi(0;\O_i)\}+r = o_p(n^{-1/2})$ using the similar steps in (i). 

(iii). Since $\psi$ is a linear function of $\beta$, we have
\be
&&\sup _{\beta \in \calB}\left\|\bar{\psi}(\beta,G_0, \hat{\boldeta})-\bar{\psi}\left(\beta,G_0, \boldeta_0\right)\right\|\n\\
&=&\sup _{\beta \in \calB}\left\|\bar{\psi}(\beta,G_0, \hat{\boldeta})-\bar{\psi}(\beta_0,G_0, \hat{\boldeta})+\bar{\psi}(\beta_0,G_0, \hat{\boldeta})-\bar{\psi}(\beta_0,G_0, \boldeta_0)+\bar{\psi}(\beta_0,G_0, \boldeta_0)-\bar{\psi}\left(\beta,G_0, \boldeta_0\right)\right\|\n\\
&\leq&\sup _{\beta \in \calB}|\beta-\beta_0|\left\|\bar{\psi^{'}}(G_0, \hat{\boldeta})-\bar{\psi^{'}}\left(G_0, \boldeta_0\right)\right\|+\|\bar{\psi}(\beta_0,G_0, \hat{\boldeta})-\bar{\psi}(\beta_0,G_0, \boldeta_0)\|\n\\
&=&o_p(n^{-1/2})\n.
\ee

(iv). By Assumption \ref{condeigen}, we know that $n\bbE(\|\bar \psi(\beta_0,G_0,\boldeta_0)\|^2)/m\leq \tr(\Sigma(\beta_0,G_0,\boldeta_0))/m\leq c$. Combining the Markov inequality, we can obtain that $\|\bar \psi(\beta_0,G_0,\boldeta_0)\|=O_p(\sqrt{m/n})$. Moreover, we have
\be
&&\bbE(\psi^{'}(G_0,\boldeta_0;\O))\n\\
&=& \bbE\{\frac{\delta}{G_0(Y|\O_A)}(g^{'}(\boldeta_0;\O)-\xi_0^{'}(G_0;\O_A))+\xi_0^{'}(G_0;\O_A)\}\n\\
&=&\bbE\{\frac{\bbE(I(T\leq C)|\O_T)}{G_0(T|\O_A)}g^{'}(\boldeta_0;\O_T)\}+\bbE\{(1-\frac{\bbE(I(T\leq C)|\O_T)}{G_0(T|\O_A)})\xi_0^{'}(G_0;\O_A)  \}\n\\
&=& \bbE(g^{'}(\boldeta_0;\O_T)).\n
\ee
By the law of large numbers, we have
\be\label{psi'}
\|\bar\psi^{'}(G_0,\boldeta_0)\|&\leq& \|\bar\psi^{'}(G_0,\boldeta_0)-\bbE(\psi^{'}(G_0,\boldeta_0;\O))\|+\|\bbE(\psi^{'}(G_0,\boldeta_0;\O))\|\n\\
& = &O_p (\sqrt{m/n} )+\|\bbE(g^{'}(\boldeta_0;\O_T))\|\n\\
&=& O_p(\mu_n/\sqrt{n}).
\ee
Because $\psi$ is the linear function of $\beta$, we have 
\be\label{psibar}
\sup_{\beta\in\calB}\|\bar \psi(\beta,G_0,\boldeta_0)-\bar \psi(\beta_0,G_0,\boldeta_0)\| = \|\bar\psi^{'}(G_0,\boldeta_0)\|\sup_{\beta\in\calB}(\beta-\beta_0) = O_p(\mu_n/\sqrt{n}).
\ee
Then Lemma \ref{lemmapsi}(i), \eqref{psi'} and \eqref{psibar} gives
\bse 
&&\sup_{\beta\in\calB}\|\bar \psi(\beta,G_0,\hat\boldeta)\|\\
&=&\sup_{\beta\in\calB}\|\bar \psi(\beta,G_0,\hat\boldeta)-\bar \psi(\beta,G_0,\boldeta_0)+\bar \psi(\beta,G_0,\boldeta_0)-\bar \psi(\beta_0,G_0,\boldeta_0)+\bar \psi(\beta_0,G_0,\boldeta_0)\|\\
&\leq &\sup_{\beta\in\calB}\|\bar \psi(\beta,G_0,\hat\boldeta)-\bar \psi(\beta,G_0,\boldeta_0)\|+\sup_{\beta\in\calB}\|\bar \psi(\beta,G_0,\boldeta_0)-\bar \psi(\beta_0,G_0,\boldeta_0)\|+\|\bar \psi(\beta_0,G_0,\boldeta_0)\|\\
&\leq& o_p(n^{-1/2})+O_p(\mu_n/\sqrt{n})+O_p(\sqrt{m/n})\\
&=&O_p(\mu_n/\sqrt{n}).
\ese
For $\sup_{\beta\in\calB}\|\bar \psi(\beta_0,G_0,\hat\boldeta)\|$, the second term in the above decomposition is zero, so we have $\sup_{\beta\in\calB}\|\bar \psi(\beta_0,G_0,\hat\boldeta)\|=O_p(\sqrt{m/n})$. This completes the proof of Lemma \ref{lemmapsi}.
\end{proof}

\begin{lemma}\label{lemmapsi-G} Under Assumptions \ref{condnuispara}-\ref{condkernel}
(i). $\left\|\bar{\psi}\left(\beta_0,\hat G , \hat{\boldeta}\right)-\bar{\psi}\left(\beta_0,G_0, \hat{\boldeta}\right)\right\|=O_p(\sqrt{m}\kappa)$.

(ii). $\left\|\bar{\psi^{'}}(\hat G ,\hat{\boldeta})-\bar{\psi^{'}}\left(G_0,\hat{\boldeta}\right)\right\|=O_p(\sqrt{m}\kappa)$.

(iii). $\sup _{\beta \in \calB}\left\|\bar{\psi}(\beta,\hat G , \hat{\boldeta})-\bar{\psi}\left(\beta,G_0, \hat{\boldeta}\right)\right\|=O_p(\sqrt{m}\kappa)$.   

(iv). $\sup_{\beta\in\calB}\|\bar \psi(\beta,\hat G ,\hat\boldeta)\|=O_p(\mu_n/\sqrt{n})$.    
\end{lemma}
\begin{proof}
(i) Since $G(Y|\O_A)\geq c>0$ for some positive constant $c$, and the function class $\boldeta$ is bounded, the moment function $\psi$ is Lipschitz continuous with respect to $G$. Because $\psi$ is a $m$-dimensional, we have 
$$
\left\|\bar{\psi}(\beta_0,\hat G , \hat{\boldeta})-\bar{\psi}\left(\beta_0,G_0, \hat{\boldeta}\right)\right\|\leq c \|\hat G-G \|_{L^2}=O_p(\sqrt{m}\kappa). 
$$
The difference between Lemma \ref{lemmapsi-G} and Lemma \ref{lemmapsi} is that $\psi$ is Neyman orthogonal to $\boldeta$ but not to $G$. As a result, the rate is not $o_p(n^{-1/2})$, but rather $O_p(\sqrt{m}\kappa)$. 

(ii), (iii) The proofs are similar to (i). 

(iv) According to Lemma \ref{lemmapsi} (iv) and Lemma \ref{lemmapsi-G} (iii), we have
\bse 
&&\sup_{\beta\in\calB}\|\bar \psi(\beta,\hat G ,\hat\boldeta)\|\\
&=&\sup_{\beta\in\calB}\|\bar \psi(\beta,\hat G ,\hat\boldeta)-\bar \psi(\beta,G_0,\hat\boldeta)+\bar \psi(\beta,G_0,\hat\boldeta)\|\\
&\leq &\sup_{\beta\in\calB}\|\bar \psi(\beta,\hat G ,\hat\boldeta)-\bar \psi(\beta,G_0,\hat\boldeta)\|+\sup_{\beta\in\calB}\|\bar \psi(\beta,G_0,\hat\boldeta)\|\\
&\leq& O_p(\sqrt{m}\kappa)+O_p(\mu_n/\sqrt{n})\\
&=&O_p(\mu_n/\sqrt{n}).
\ese
the last line comes from Assumption \ref{condweak} (ii) : $\sqrt{nm}\kappa/\mu_n=o(1)$, which implies $\sqrt{m}\kappa=o(\mu_n/\sqrt{n})$ and $\mu_n/\sqrt{n}$ dominates $\sqrt{m}\kappa$.
\end{proof}
\begin{lemma}\label{lemmasigma}
Under Assumptions \ref{condnuispara}-\ref{condeigen},

(i). $\left\|\bar{\Sigma}\left(\beta_0, G_0,\hat{\boldeta}\right)-\Sigma_0\right\|=O_p(\sqrt{m\log(m)/n})$;

(ii). $\left\|n^{-1} \sum_{i=1}^n \psi(\beta_0, G_0, \hat{\boldeta};\O_i) \psi^{'}(G_0,\hat{\boldeta};\O_i)\tp-\bbE\left\{\psi(\beta_0,G_0,\boldeta_0;\O) \psi^{'}(G_0,\boldeta_0;\O)\tp\right\}\right\|\\=O_p(\sqrt{m\log(m)/n})$;

(iii). $\left\|n^{-1} \sum_{i=1}^n \psi^{'}(G_0,\hat{\boldeta};\O_i) \psi^{'}(G_0,\hat{\boldeta};\O_i)\tp-\bbE\left\{\psi^{'}(G_0,\boldeta_0;\O) \psi^{'}(G_0,\boldeta_0;\O)\tp\right\}\right\|=\\
O_p(\sqrt{m\log(m)/n})$;

(iv). $\sup _{\beta \in \calB}\left\|\bar{\Sigma}(\beta,G_0, \hat{\boldeta})-\Sigma\left(\beta,G_0, \boldeta_0\right)\right\|=O_p(\sqrt{m\log(m)/n})$;

(v). $\sup _{\beta \in \calB}\left\|n^{-1} \sum_{i=1}^n \psi(\beta,G_0, \hat{\boldeta};\O_i)\psi^{'}(G_0,\hat{\boldeta};\O_i)\tp-\bbE\left\{\psi(\beta,G_0, \boldeta_0;\O) \psi^{'}(G_0,\boldeta_0;\O)\tp\right\}\right\|=O_p(\sqrt{m\log(m)/n})$.    
\end{lemma}
\begin{proof}
(i). Let $\psi_j(\beta,G,\boldeta;\O_i)$ and $\xi_{0j}(\beta;\O_{A,i})$ be the $j$th element of $\psi(\beta,G,\boldeta;\O_i)$ and $\xi_0(\beta;\O_A)$. The matrix norm is smaller than the matrix trace:
\be
&&\left\|\bar{\Sigma}\left(\beta_0,G_0, \hat{\boldeta}\right)-\Sigma_0\right\|\n\\
&\leq& \left\|\bar{\Sigma}\left(\beta_0,G_0, \hat{\boldeta}\right)-\bar{\Sigma}\left(\beta_0,G_0, {\boldeta}_0\right)\right\|+\left\|\bar{\Sigma}\left(\beta_0,G_0, {\boldeta}_0\right)-\Sigma_0\right\|\n\\
&\leq&\tr(\bar{\Sigma}\left(\beta_0,G_0, \hat{\boldeta}\right)-\bar{\Sigma}\left(\beta_0,G_0, {\boldeta}_0\right))+\left\|\bar{\Sigma}\left(\beta_0,G_0, {\boldeta}_0\right)-\Sigma_0\right\|\n\\
&=&\sum_{j=1}^{m}\frac{1}{n}\sum_{i=1}^n\{\psi_j(\beta_0,G_0,\hat\boldeta;\O_i)^2-\psi_j(\beta_0,G_0,\boldeta_0;\O_i)^2\}+\left\|\bar{\Sigma}\left(\beta_0,G_0, {\boldeta}_0\right)-\Sigma_0\right\|\n.
\ee 
Theorem 1 in \cite{tropp2016} gives that 
\be \label{eqtropp}
\left\|\bar{\Sigma}\left(\beta_0,G_0, {\boldeta}_0\right)-\Sigma_0\right\|=O_p(\sqrt{m\log(m)/n}).
\ee

Next, we will prove $1/n\sum_{i=1}^n\{\psi_j(\beta_0,G_0,\hat\boldeta;\O_i)^2-\psi_j(\beta_0,G_0,\boldeta_0;\O_i)^2\}=o_p(\sqrt{m/n})$.
\be     
&&\frac{1}{n}\sum_{i=1}^n\{\psi_j(\beta_0,G_0,\hat\boldeta;\O_i)^2-\psi_j(\beta_0,G_0,\boldeta_0;\O_i)^2\}\n\\
&=&\frac{1}{n}\sum_{i=1}^n \frac{\delta_i}{G_0(Y_i|\O_{A,i})}\{g_j(\beta_0,\hat\boldeta;\O_i)^2-g_j(\beta_0,\boldeta_0;\O_i)^2 \}  \n\\
&&+\frac{1}{n}\sum_{i=1}^n (1-\frac{\delta_i}{G_0(Y_i|\O_{A,i})})\{\xi_{0j}(\beta_0;\O_{A,i})^2-\xi_{0j}(\beta_0;\O_i)^2 \} \n\\
&&-\frac{2}{n}\sum_{i=1}^n \frac{\delta_i}{G_0(Y_i|\O_{A,i})}(1-\frac{\delta_i}{G_0(Y_i|\O_{A,i})})\{g_j(\beta_0,\hat\boldeta;\O_{A,i})\xi_{0j}(\beta_0;\O_{A,i})\n\\
&&\quad\quad\quad\quad-g_j(\beta_0,\boldeta_0;\O_i)\xi_{0j}(\beta_0;\O_{A,i}) \} \n\\
&\equiv&E_1+E_2+E_3\n.
\ee
For $E_1$, Lemma \ref{lemmaomega} (i) gives that $\bbE\{g_j(\beta_0,\hat\boldeta;\O_i)^2-g_j(\beta_0,\boldeta_0;\O_i)^2 \}=o_p(\sqrt{m/n})$, so
\bse
&&\frac{1}{n}\sum_{i=1}^n \frac{\delta_i}{G_0(Y_i|\O_{A,i})}\{g_j(\beta_0,\hat\boldeta;\O_i)^2-g_j(\beta_0,\boldeta_0;\O_i)^2 \}\\
&=&(\bbE_n-\bbE)\left[\frac{\delta}{G_0(Y|\O_A)}\{g_j(\beta_0,\hat\boldeta;\O)^2-g_j(\beta_0,\boldeta_0;\O)^2 \}\right]\\
&&+\bbE\left[\frac{\delta}{G_0(Y|\O_A)}\{g_j(\beta_0,\hat\boldeta;\O)^2-g_j(\beta_0,\boldeta_0;\O)^2 \}\right]\\
&=&\bbE\left(\bbE\left[\frac{\delta}{G_0(T|\O_A)}\{g_j(\beta_0,\hat\boldeta;\O_T)^2-g_j(\beta_0,\boldeta_0;\O_T)^2 \}|\O_T\right]\right)+o_p(\sqrt{m/n})\\
&=&\bbE\left(\{g_j(\beta_0,\hat\boldeta;\O_T)^2-g_j(\beta_0,\boldeta_0;\O_T)^2 \}\bbE\left[\frac{\delta}{G_0(T|X)}|\O_T\right]\right)+o_p(\sqrt{m/n})\\
&=&o_p(\sqrt{m/n}).
\ese
For $E_2$, the empirical process gives that 
$$
(\bbE_n-\bbE)\left[(1-\frac{\delta}{G_0(Y|\O_A)})\{\hat\xi_j(\beta_0;\O_A)^2-\xi_{0j}(\beta_0;\O_A)^2 \}\right]=o_p(\sqrt{m/n}).
$$
Also we have
\bse
&&\bbE\left[(1-\frac{\delta}{G_0(Y|\O_A)})\{\hat\xi_j(\beta_0;\O_A)^2-\xi_{0j}(\beta_0;\O_A)^2 \}\right]\\
&=&\bbE\left(\bbE\left[(1-\frac{\delta}{G_0(T|\O_A)})\{\hat\xi_j(\beta_0;\O_A)^2-\xi_{0j}(\beta_0;\O_A)^2 \}|\O_T\right]\right)\\
&=&\bbE\left(\{\hat\xi_j(\beta_0;\O_A)^2-\xi_{0j}(\beta_0;\O_A)^2 \}\bbE\left[1-\frac{\delta}{G_0(T|\O_A)}|\O_T\right]\right)\\
&=&0.
\ese
So $E_2=o_p(\sqrt{m/n})$. Similarly, we can prove $E_3=o_p(\sqrt{m/n})$. Thus $1/n\sum_{i=1}^n\{\psi_j(\beta_0,G_0,\hat\boldeta;\O_i)^2-\psi_j(\beta_0,G_0,\boldeta_0;\O_i)^2\}=o_p(\sqrt{m/n})$, and together with \eqref{eqtropp}, we obtain $\left\|\bar{\Sigma}\left(\beta_0,G_0, \hat{\boldeta}\right)-\Sigma_0\right\|=O_p(\sqrt{m\log(m)/n})$.

(ii)-(v). The proof is similar to (i) and Lemma \ref{lemmaomega}, and thus the details are omitted here.
\end{proof}

\begin{lemma} \label{lemmalambdacensor}
Under Assumptions \ref{condnuispara}-\ref{condkernel}, we have 
$\sup_{\beta\in \calB}\|\blambda(\beta,\hat G,\hat\boldeta)\|=O_p(\mu_n/\sqrt{n})$,\\ $\sup_{\beta\in \calB}\|\blambda(\beta,G_0,\boldeta_0)\|=O_p(\mu_n/\sqrt{n})$, $\|\blambda(\beta_0,G_0,\boldeta_0)\|=O_p(\sqrt{m/n})$. 
\end{lemma}

\begin{proof}
The proof is similar to that of Lemma \ref{lemmalambda} (i). Just by replace $g(\beta,\boldeta)$ to $\psi(\beta,G,\boldeta)$, we can obtain an inequality analogous to \eqref{eqlambda}:
$$\sup_{\beta\in \calB}\|\blambda(\beta,\hat G ,\hat\boldeta)\|\leq c\sup_{\beta\in \calB}\|\bar \psi(\beta,\hat G ,\hat\boldeta)\|=O_p(\frac{\mu_n}{\sqrt{n}}).$$
Similarly, we can obtain $\sup_{\beta\in \calB}\|\blambda(\beta,G_0,\boldeta_0)\|=O_p(\mu_n/\sqrt{n})$, $\|\blambda(\beta_0,G_0,\boldeta_0)\|=O_p(\sqrt{m/n})$ by $\|\bar\psi(\beta,G_0,\boldeta_0)\|=O_p(\mu_n/\sqrt{n})$ and $\|\bar\psi(\beta_0,G_0,\boldeta_0)\|=O_p(\sqrt{m/n})$, and $\sup_{\beta\in \calB}\|\blambda(\beta,\hat G ,\hat\boldeta)-\blambda(\beta,G_0 ,\hat\boldeta_0) \|=o_p(\mu_n/\sqrt{n})$.

Also, we can derive the expansion of $\blambda$ by following similar steps as in equation \eqref{Taylor_lambda}. 
\be
-\bar \psi(\beta_0,\hat G ,\hat\boldeta)-\bar\Sigma(\beta_0,\hat G ,\boldeta)\blambda(\beta_0,\hat G ,\hat\boldeta)+R=0,\n
\ee
where the remainder term $R=1/(2n)\sum_{i=1}^n\rho^{'''}(\bar x)(\blambda(\beta_0,\hat G ,\hat\boldeta)\tp \psi(\beta_0,\hat G ,\hat\boldeta;\O_i))^2\psi(\beta_0,\hat G ,\hat\boldeta;\O_i)$. According to Assumption \ref{condmoment} and \ref{condgel}, we have
\be
&&\|R\|\n\\
&\leq& c\sup_{1\leq i\leq n}\|\psi(\beta_0,\hat G ,\hat\boldeta;\O_i)\|\|\blambda(\beta_0,\hat G ,\hat\boldeta)\|^2\|\bar\Sigma(\beta_0,\hat G ,\hat\boldeta)\|\n\\
&=&O_p(n^{1/\gamma}\{\bbE(\|\psi(\beta_0,\hat G ,\hat\boldeta;\O)\|^\gamma)\}^{1/\gamma})O_p(\mu_n^2/n)\n\\
&=&o_p(\mu_n/\sqrt{n}).
\ee
So we have
\be  \label{R1}
\blambda(\beta_0,\hat G ,\hat\boldeta)=(-\bar \psi(\beta_0,\hat G ,\hat\boldeta)+R)\bar\Sigma(\beta_0,\hat G ,\hat\boldeta)^{-1}.
\ee  
Then based on Lemmas \ref{lemmapsi}-\ref{lemmapsi-G} and Assumption \ref{condweak}, $\|\blambda(\beta_0,\hat G ,\hat\boldeta)-\blambda(\beta_0, G_0 ,\boldeta_0)\|$ becomes
\be
&&\|\blambda(\beta_0,\hat G ,\hat\boldeta)-\blambda(\beta_0, G_0 ,\boldeta_0)\|\n\\
&=&\| \bar\psi(\beta_0,G_0 ,\boldeta_0)\bar\Sigma(\beta_0,G_0,\boldeta_0)-\bar\psi(\beta_0,\hat G ,\hat\boldeta)\bar\Sigma(\beta_0,\hat G,\hat\boldeta)+o_p(\mu_n/\sqrt{n})\|\n\\
&=&\|\bar\psi(\beta_0,G_0 ,\boldeta_0)-\bar\psi(\beta_0,\hat G ,\hat\boldeta)\|\|\bar\Sigma(\beta_0,G_0,\boldeta_0)\|\n\\
&&+\|\bar\psi(\beta_0,\hat G ,\hat\boldeta)\|\|\bar\Sigma(\beta_0,\hat G,\hat\boldeta)-\bar\Sigma(\beta_0,G_0,\boldeta_0)\|\n\\
&=&O_p(\sqrt{m}\kappa)O_p(1)+O_p(\mu_n\sqrt{n})o_p(1)\n\\
&=&o_p(\mu_n\sqrt{n}).\n
\ee
\end{proof}
\begin{lemma} \label{lemmaqc}
Let $\tilde Q(\beta,G,\boldeta) = \tilde \psi(\beta,G,\boldeta)\tp \Sigma(\beta,G,\boldeta)^{-1}\tilde \psi(\beta,G,\boldeta)\tp/2+m/(2n)$. Under Assumptions \ref{condnuispara}-\ref{condgel}, $$\sup_{\beta\in \calB}|\tilde Q(\beta,G_0,\boldeta_0)-\hat Q(\beta,\hat G ,\hat\boldeta)|=o_p(\frac{\mu_n^2}{n}).$$
\end{lemma}
\begin{proof}
According to the Taylor expansion at $\blambda\tp g= 0$, we have
\be \label{qhat11}
&&\hat Q(\beta,\hat G ,\hat\boldeta) \n\\
&=& \frac{1}{n}\sum_{i=1}^n \rho( \blambda(\beta,\hat G ,\hat\boldeta)\tp \psi(\beta,\hat G ,\hat\boldeta;\O_i))\n\\
&=&\frac{1}{n}\sum_{i=1}^n \rho(0)+\frac{1}{n}\sum_{i=1}^n \rho^{'}(0) \blambda(\beta,\hat G ,\hat\boldeta)\tp \psi(\beta,\hat G ,\hat\boldeta;\O_i)\n\\
&&+ \frac{1}{2n}\sum_{i=1}^n \rho^{''}(0) \blambda(\beta,\hat G ,\hat\boldeta)\tp \psi(\beta,\hat G ,\hat\boldeta;\O_i)\psi(\beta,\hat G ,\hat\boldeta;\O_i)\tp \blambda(\beta,\hat G ,\hat\boldeta)+r,
\ee
where $r=1/(6n)\sum_{i=1}^n\rho^{'''}(\bar x)( \blambda(\beta,\hat G ,\hat\boldeta)\tp \psi(\beta,\hat G ,\hat\boldeta;\O_i))^3$ is a remainder term and $\bar x$ lies between $0$ and $ \blambda(\beta,\hat G ,\hat\boldeta)\tp \psi(\beta,\hat G ,\hat\boldeta;\O_i)$. It can be proved that 
\be\label{rorder1}
r&=&\frac{1}{6n}\sum_{i=1}^n\rho^{'''}(\bar x)( \blambda(\beta,\hat G ,\hat\boldeta)\tp \psi(\beta,\hat G ,\hat\boldeta;\O_i))^3\n\\
&\leq&c\|\blambda(\beta,\hat G ,\hat\boldeta)\|\max_{i}\|\psi(\beta,\hat G ,\hat\boldeta;\O_i)\|\blambda(\beta,\hat G ,\hat\boldeta)\tp\bar\Sigma(\beta,\hat G ,\hat\boldeta)\blambda(\beta,\hat G ,\hat\boldeta)   \n\\
&=&O_p(\frac{\mu_n}{\sqrt{n}}\tau_n)O_p(\mu_n^2/n)\n\\
&=&o_p(\mu_n^2/n).
\ee
Based on \eqref{R1} and \eqref{rorder1}, \eqref{qhat11} becomes
\be 
&&\hat Q(\beta,\hat G ,\hat\boldeta) \n\\
&=& -\frac{1}{n}\sum_{i=1}^n  (-\bar \psi(\beta,\hat G ,\hat\boldeta)+R)\tp\bar\Sigma(\beta,\hat G ,\hat\boldeta)^{-1} \psi(\beta,\hat G ,\hat\boldeta;\O_i)\n\\
&&- \frac{1}{2n}\sum_{i=1}^n(-\bar \psi(\beta,\hat G ,\hat\boldeta)+R)\tp\bar\Sigma(\beta,\hat G ,\hat\boldeta)^{-1} \psi(\beta,\hat G ,\hat\boldeta;\O_i)\psi(\beta,\hat G ,\hat\boldeta;\O_i)\tp\n\\
&&\quad\times\bar\Sigma(\beta,\hat G ,\hat\boldeta)^{-1}(-\bar \psi(\beta,\hat G ,\hat\boldeta)+R)+r\n\\
&=& \frac{1}{2} \bar  \psi(\beta,\hat G ,\hat\boldeta)\tp\bar\Sigma(\beta,\hat G ,\hat\boldeta)^{-1} \bar \psi(\beta,\hat G ,\hat\boldeta) +o_p(\mu_n^2/n).
\ee
Next we will prove $\sup_{\beta\in \calB}| \bar \psi(\beta,\hat G ,\hat\boldeta)\tp\bar\Sigma(\beta,\hat G ,\hat\boldeta)^{-1} \bar \psi(\beta,\hat G ,\hat\boldeta)-\tilde Q(\beta,G_0,\boldeta_0) |=o_p(\mu_n^2/n)$.According to Lemmas \ref{lemmapsi}-\ref{lemmasigma}, we have 
\be\label{eq33}
&&\sup_{\beta\in \calB}| \bar \psi(\beta,\hat G \hat\boldeta)\tp\bar\Sigma(\beta,\hat G ,\hat\boldeta)^{-1} \bar \psi(\beta,\hat G ,\hat\boldeta)-\tilde \psi(\beta,\hat G ,\boldeta_0)\tp\Sigma(\beta,G_0,\boldeta_0)^{-1} \tilde \psi(\beta,G_0,\boldeta_0)-m/(2n) |\n\\
&\leq &\sup_{\beta\in \calB}| \bar \psi(\beta,\hat G ,\hat\boldeta)\tp\bar\Sigma(\beta,\hat G ,\hat\boldeta)^{-1} \bar \psi(\beta,\hat G ,\hat\boldeta)-\bar \psi(\beta,G_0,\boldeta_0)\tp\bar\Sigma(\beta,\hat G ,\hat\boldeta)^{-1} \bar \psi(\beta,G_0,\boldeta_0) | \n\\
&&+\sup_{\beta\in \calB}|\bar \psi(\beta,G_0,\boldeta_0)\tp\bar\Sigma(\beta,\hat G ,\hat\boldeta)^{-1} \bar \psi(\beta,G_0,\boldeta_0)-\bar \psi(\beta,G_0,\boldeta_0)\tp\Sigma(\beta,G_0,\boldeta_0)^{-1} \bar \psi(\beta,G_0,\boldeta_0)| \n\\
&&+\sup_{\beta\in \calB}|\bar \psi(\beta,G_0,\boldeta_0)\tp\Sigma(\beta,G_0,\boldeta_0)^{-1} \bar \psi(\beta,G_0,\boldeta_0)-\tilde Q(\beta,G_0,\boldeta_0)|.\n
\ee
Since the nuisance functions in the third term are all at their true value, Lemma A.2 in \cite{newey2009} gives that the third term has order $o_p(\mu_n^2/n)$. Then \ref{eq33} becomes
\be
&&\sup_{\beta\in \calB}\| \bar \psi(\beta,\hat G ,\hat\boldeta)-\bar \psi(\beta,G_0,\boldeta_0)\|\sup_{\beta\in \calB}\| \bar \psi(\beta,\hat G ,\hat\boldeta)+\bar \psi(\beta,G_0,\boldeta_0)\|\sup_{\beta\in \calB}\|\bar\Sigma(\beta,\hat G ,\hat\boldeta)^{-1}\|\n\\
&&+ \sup_{\beta\in \calB}\| \bar \psi(\beta,G_0,\boldeta_0)\|^2\sup_{\beta\in \calB}\|\bar\Sigma(\beta,\hat G ,\hat\boldeta)^{-1}-\Sigma(\beta,G_0,\boldeta_0)^{-1}\|+o_p(\mu_n^2/n)\n\\
&=&O_p(\sqrt{m}\kappa)O_p(\mu_n/\sqrt{n})O_p(1)+O_p(\mu_n^2/n)o_p(1)+o_p(\mu_n^2/n)\n\\
&=&o_p(\mu_n^2/n).
\ee
Combining \eqref{qhat11} and \eqref{eq33}, we obtain $$\sup_{\beta\in \calB}|\tilde Q(\beta,G_0,\boldeta_0)-\hat Q(\beta,\hat G ,\hat\boldeta)|=o_p(\frac{\mu_n^2}{n}).$$
\end{proof}

\subsection{Corollary \ref{coro1}}
\begin{corollary}\label{coro1}
Suppose Assumptions \ref{condcondind}-\ref{condgel} hold, and the censoring rate is 0. $\hat{\beta}_{\textrm{GEL2}}$ is consistent and asymptotical normal, i.e., as $n \rightarrow \infty$, and $m^3/n \to 0$, $\hat{\beta}_{\textrm{GEL2}} \xrightarrow{p} \beta_0$ and
$$
\frac{\mu_n\left(\hat{\beta}_{\textrm{GEL2}}-\beta_0\right)}{\sqrt{H^{-1}(H+V_4)H^{-1}}} \xrightarrow{d} N(0,1),
$$
where  
$V_4 = \mu_n^{-2} \bbE\left[U_i \Omega_0^{-1} U_i\right]$,
$H = n g^{*\top} \Omega_0^{-1} g^{*}/\mu_n^2$,
and $U_i=g^{'}(\boldeta_0;\O_i)-g^{*}-\{\Omega_0^{-1} \bbE(g(\beta_0,\boldeta_0;\O)\allowbreak g^{'}(\boldeta_0;\O)\tp)\}\tp g(\beta_0,\boldeta_0;\O_i)$ is the population residual from least squares regression of $g^{'}(\boldeta_0;\O_i)-g^{*}$ on $g(\beta_0,\boldeta_0;\O_i)$. 
\end{corollary}

\begin{proof}
According to Lemma \ref{lemma1} (i), it is sufficient to prove $\sqrt{n}\|\tilde{g}(\hat\beta,\boldeta_0)\|/\mu_n = o_p(1)$ to prove the consistency. Recall that $\tilde g(\beta_0,\boldeta_0)=\bbE(g(\beta_0,\boldeta_0;\O_i))=0$ from the robustness of the influence functon, which implies $\tilde Q(\beta_0,\boldeta_0) = \tilde g(\beta_0,\boldeta_0)\tp\tilde \Omega(\beta_0,\boldeta_0)^{-1}\tilde g(\beta_0,\boldeta_0)/2+m/(2n)=0+m/(2n)=m/(2n)$. From Lemma \ref{lemma1}, and Assumption \ref{condeigen}, we have
\bse
&&\|\hat\beta-\beta_0\|^2\\
&\leq & \frac{n}{\mu_n^2}\|\tilde g(\hat\beta,\boldeta_0) \|^2\\
&\leq&c \frac{n}{\mu_n^2}\tilde g(\hat\beta,\boldeta_0)\tp \Omega(\hat\beta,\boldeta_0)^{-1} \tilde g(\hat\beta,\boldeta_0)+m/(2n)-m/(2n)\\
&\leq&c \frac{n}{\mu_n^2}\{ \tilde Q(\hat\beta,\boldeta_0)-\tilde Q(\beta_0,\boldeta_0)\}\\
&\leq& c \frac{n}{\mu_n^2}\{\tilde Q(\hat\beta,\boldeta_0)-\hat Q(\hat\beta,\hat\boldeta)+\hat Q(\hat\beta,\hat\boldeta)-\hat Q(\beta_0,\hat\boldeta)+\hat Q(\beta_0,\hat\boldeta)-\tilde Q(\beta_0,\boldeta_0)\}\\
&=&c\frac{n}{\mu_n^2}o_p(\frac{\mu_n^2}{n})\\
&=&o_p(1),
\ese
where the second to the last step comes from Lemma \ref{lemmaq} and $\hat Q(\hat\beta,\hat\boldeta)\leq \hat Q(\beta_0,\hat\boldeta)$ because $\hat\beta$ minimizes $\hat Q(\beta,\hat\boldeta)$.
This completes the proof of the consistency.

Next, we will prove the asymptotic normality of $\hat\beta$. From Taylor expansion of the first order condition $\partial \hat{Q}(\beta, \hat{\boldeta}) /\partial \beta|_{\beta=\hat{\beta}}=0$, we have that
\be \label{Taylor}
0=\frac{n}{\mu_n}\left.\frac{\partial \hat{Q}(\beta, \hat{\boldeta})}{\partial \beta}\right|_{\beta=\hat{\beta}}=\frac{n}{\mu_n}\left.\frac{\partial \hat{Q}(\beta, \hat{\boldeta})}{\partial \beta}\right|_{\beta=\beta_0}+\frac{n}{\mu_n^2}\left.\frac{\partial^2 \hat{Q}(\beta, \hat{\boldeta})}{\partial \beta^2}\right|_{\beta=\bar{\beta}}\mu_n\left(\hat{\beta}-\beta_0\right),
\ee
where $\bar{\beta}$ lies between $\beta_0$ and $\hat{\beta}$. First we analyze $\partial \hat{Q}(\beta, \hat{\boldeta})/{\partial \beta}|_{\beta=\beta_0}$. The first order condition of $\blambda$ gives that
\be \label{lambdafirstorder}
\frac{1}{n}\sum_{i=1}^n\rho(\blambda(\beta_0,\hat\boldeta)\tp g(\beta_0,\hat\boldeta;\O_i))g(\beta_0,\hat\boldeta;\O_i)=0.
\ee
Define $\hat U_i=g^{'}(\boldeta_0;\O_i)-g^{*}-1/n\sum_{j=1}^n\{g(\beta_0,\hat\boldeta;\O_j)g^{'}(\hat\boldeta;\O_j)\tp\}\bar \Omega(\beta_0,\hat\boldeta)^{-1}g(\beta_0,\hat\boldeta;\O_i)$. Combining with \eqref{Taylor_lambda}, \eqref{lambdafirstorder}, and Lemma \ref{lemmalambda}, we have
\be \label{qda123}
&&\left.\frac{\partial \hat{Q}(\beta, \hat{\boldeta})}{\partial \beta}\right|_{\beta=\beta_0}\n\\
&=&\frac{1}{n}\sum_{i=1}^n\rho^{'}(\blambda(\beta_0,\hat\boldeta)\tp g(\beta_0,\hat\boldeta;\O_i))\blambda(\beta_0,\hat\boldeta)\tp g^{'}(\hat\boldeta;\O_i)\n\\
&=&\frac{1}{n}\sum_{i=1}^n\{-1-\blambda(\beta_0,\hat\boldeta)\tp g(\beta_0,\hat\boldeta;\O_i)+r\}\{-\bar g(\beta_0,\hat\boldeta)\tp\bar\Omega(\beta_0,\hat\boldeta)^{-1}+R\}g^{'}(\hat\boldeta;\O_i) \n\\
&=&-\bar g(\beta_0,\hat\boldeta)\tp\bar\Omega(\beta_0,\hat\boldeta)^{-1}\left\{\frac{1}{n}\sum_{i=1}^{n}g(\beta_0,\hat\boldeta;\O_i)g^{'}(\hat\boldeta;\O_i)\tp\right\}\bar\Omega(\beta_0,\hat\boldeta)^{-1}\bar g(\beta_0,\hat\boldeta) \n\\
&&+\bar {g^{'}}(\hat\boldeta)\tp\bar\Omega(\beta_0,\hat\boldeta)^{-1}\bar g(\beta_0,\hat\boldeta)+o_p({\mu_n}/{n})\n\\
&=& g^{*\top} \bar\Omega(\beta_0,\hat\boldeta)^{-1}\bar g(\beta_0,\hat\boldeta)+ \frac{1}{n}\sum_{i=1}^n\hat U_i\tp\bar \Omega(\beta_0,\hat\boldeta)^{-1}\bar g(\beta_0,\hat\boldeta)+o_p({\mu_n}/{n}).
\ee   
Lemma A12 in \cite{newey2009} gives that 
\be\label{eqlemmaa12}
\frac{n}{\mu_n}\{g^{*\top} \Omega(\beta_0,\boldeta_0)^{-1}\bar g(\beta_0,\boldeta_0)+ \frac{1}{n}\sum_{i=1}^n U_i\tp \Omega(\beta_0,\boldeta_0)^{-1}\bar g(\beta_0,\boldeta_0)+o_p(1)\}\xrightarrow{d} N(0,H+V_4).
\ee
The main difference between \eqref{qda123} and \eqref{eqlemmaa12} is that the nuisance parameter is $\widehat\boldeta$ in the former and $\boldeta_0$ in the latter. Thus, next we will prove
\be \label{normal1}
g^{*\top} \Omega(\beta_0,\boldeta_0)^{-1}\bar g(\beta_0,\boldeta_0)-g^{*\top} \bar\Omega(\beta_0,\hat\boldeta)^{-1}\bar g(\beta_0,\hat\boldeta)=o_p(\mu_n/n)
\ee
and 
\be \label{normal2}
\frac{1}{n}\sum_{i=1}^n U_i\tp \Omega(\beta_0,\boldeta_0)^{-1}\bar g(\beta_0,\boldeta_0)-\frac{1}{n}\sum_{i=1}^n\hat U_i\tp\bar \Omega(\beta_0,\hat\boldeta)^{-1}\bar g(\beta_0,\hat\boldeta)=o_p({\mu_n}/{n}).
\ee
For \eqref{normal1}, based on Assumption \ref{condeigen}, Lemmas \ref{lemmag} and \ref{lemmaomega}, we have
\be    
&& g^{*\top} \Omega(\beta_0,\boldeta_0)^{-1}\bar g(\beta_0,\boldeta_0)-g^{*\top} \bar\Omega(\beta_0,\hat\boldeta)^{-1}\bar g(\beta_0,\hat\boldeta)\n\\
&=&g^{*\top}\Omega(\beta_0,\boldeta_0)^{-1}\{\bar g(\beta_0,\boldeta_0)-\bar g(\beta_0,\hat\boldeta)\}+g^{*\top}\{\Omega(\beta_0,\boldeta_0)^{-1}-\bar\Omega(\beta_0,\hat\boldeta)^{-1}\}\bar g(\beta_0,\hat\boldeta)\n\\
&\leq&\|g^{*}\|\|\Omega(\beta_0,\boldeta_0)^{-1}\|\|\bar g(\beta_0,\boldeta_0)-\bar g(\beta_0,\hat\boldeta)\|\n\\
&&+\|g^{*}\|\|\bar g(\beta_0,\hat\boldeta)\|\|\Omega(\beta_0,\boldeta_0)^{-1}-\bar\Omega(\beta_0,\hat\boldeta)^{-1}\|\n\\
&=&O_p(\mu_n/\sqrt{n})O_p(1)o_p(n^{-1/2})+O_p(\mu_n/\sqrt{n})O_p(\sqrt{m/n})O_p(\sqrt{m\log(m)/n})\n\\
&=&o_p(\mu_n/n).\n
\ee
The last step comes from $O_p(\sqrt{m^2\log(m)/n})=\sqrt{m^3/n}=o_p(1)$. Before calculate the order of \eqref{normal2}, we will first derive the order of $\|1/n\sum_{i=1}^n(\hat U_i-U_i)\|$ by Lemmas \ref{lemmag}-\ref{lemmaomega} and Assumption \ref{condeigen}:
\be\label{ui}
&&\|\frac{1}{n}\sum_{i=1}^n(\hat U_i-U_i)\|\n\\
&\leq&\|\bar {g^{'}}(\hat\boldeta)-\bar {g^{'}}(\boldeta_0)\|+\Big\|\{\frac{1}{n}\sum_{i=1}^ng(\beta,\hat\boldeta;\O_i)g^{'}(\hat\boldeta;\O_i)\tp\}\bar \Omega(\beta,\hat\boldeta)^{-1}\bar g(\beta,\hat\boldeta)\n\\
&&-\bbE \{g(\beta,\boldeta_0;\O)g^{'}(\boldeta_0;\O)\tp\}\Omega(\beta,\boldeta_0)^{-1}\bar g(\beta,\boldeta_0)\Big\|\n\\
&=&\|\bar {g^{'}}(\hat\boldeta)-\bar {g^{'}}(\boldeta_0)\|\n\\
&&+\|[\frac{1}{n}\sum_{i=1}^ng(\beta,\hat\boldeta;\O_i)g^{'}(\hat\boldeta;\O_i)\tp-\bbE \{g(\beta,\boldeta_0;\O)g^{'}(\boldeta_0;\O)\tp\}]\Omega(\beta,\hat\boldeta)^{-1}\bar g(\beta,\hat\boldeta)   \|\n\\
&&+\|\bbE \{g(\beta,\boldeta_0;\O)g^{'}(\boldeta_0;\O)\tp\}\{\bar \Omega(\beta,\hat\boldeta)^{-1}-\Omega(\beta,\boldeta_0)^{-1}\}\bar g(\beta,\hat\boldeta)\|\n\\
&&+\|\bbE \{g(\beta,\boldeta_0;\O)g^{'}(\boldeta_0;\O)\tp\}\Omega(\beta,\boldeta_0)^{-1}\{\bar g(\beta,\hat\boldeta)-\bar g(\beta,\boldeta_0)\}\|\n\\
&=&o_p(n^{-1/2})+O_p(\sqrt{m\log(m)/n})O_p(\mu_n/\sqrt{n})+O_p(\sqrt{m\log(m)/n})O_p(\mu_n/\sqrt{n})+o_p(n^{-1/2})  \n\\
&=&O_p(\mu_n\sqrt{m\log (m)}/n)+o_p(n^{-1/2}).
\ee
For \eqref{normal2}, based on Assumption \ref{condeigen}, Lemmas \ref{lemmag}, \ref{lemmaomega} and \eqref{ui}, we have
\be
&&\frac{1}{n}\sum_{i=1}^n U_i\tp \Omega(\beta_0,\boldeta_0)^{-1}\bar g(\beta_0,\boldeta_0)-\frac{1}{n}\sum_{i=1}^n\hat U_i\tp\bar \Omega(\beta_0,\hat\boldeta)^{-1}\bar g(\beta_0,\hat\boldeta) \n\\
&=&\frac{1}{n}\sum_{i=1}^n U_i\tp \Omega(\beta_0,\boldeta_0)^{-1}\{\bar g(\beta_0,\boldeta_0)-\bar g(\beta_0,\hat\boldeta)\} \n\\
&&+\frac{1}{n}\sum_{i=1}^n U_i\tp \{\Omega(\beta_0,\boldeta_0)^{-1}-\bar \Omega(\beta_0,\hat\boldeta)^{-1}\}\bar g(\beta_0,\hat\boldeta)   \n\\
&&+\frac{1}{n}\sum_{i=1}^n(U_i-\hat U_i)\tp\bar \Omega(\beta_0,\hat\boldeta)^{-1}\bar g(\beta_0,\hat\boldeta)  \n\\
&\leq& O_p(\frac{\mu_n}{\sqrt{n}})O_p(1)o_p(\frac{1}{\sqrt{n}})+O_p(\frac{\mu_n}{\sqrt{n}})O_p(\frac{\sqrt{m\log (m)}}{\sqrt{n}})O_p(\frac{\sqrt{m}}{\sqrt{n}})\n\\
&&+O_p(\frac{\mu_n\sqrt{m\log (m)}}{n})O_p(1)O_p(\frac{\sqrt{m}}{\sqrt{n}})+ o_p(n^{-1/2})O_p(1)O_p(\frac{\sqrt{m}}{\sqrt{n}}) \n\\
&=&o_p(\mu_n/n).\n
\ee
where the second to last step comes from that $m/\mu_n^2$ is bounded and 
\bse
&&\frac{\mu_n}{\sqrt{n}}\|\frac{1}{n}\sum_{i=1}^n U_i\|\\
&\leq&\|\bar {g^{'}}(\boldeta_0)-g^{*}\|+\|\bbE\{g(\beta_0,\boldeta_0;\O)g^{*\top}\}\Omega_0^{-1}\bar g(\beta_0,\boldeta_0)\|\\
&=&O_p(\frac{\mu_n}{\sqrt{n}}).
\ese Combining \eqref{qda123}-\eqref{normal2}, we have  
\be\label{Taylor1}
\frac{n}{\mu_n}\left.\frac{\partial \hat{Q}(\beta, \hat{\boldeta})}{\partial \beta}\right|_{\beta=\beta_0}\xrightarrow{d} N(0,H+\Lambda).
\ee
Next, we will prove $\sup_{\beta\in \calB} n/\mu_n^2\partial^2 \hat{Q}(\beta, \hat\boldeta)/{\partial \beta^2}\to H$. Let $\calV_{gg}(\beta,\boldeta)=\sum_{i=1}^n\rho^{''}(\blambda(\beta,\boldeta)\tp g(\beta,\boldeta;\O_i))\\g(\beta,\boldeta;\O_i)g(\beta,\boldeta;\O_i)\tp/n$, $\calV_{gg^{'}}(\beta,\boldeta)=\sum_{i=1}^n\rho^{''}(\blambda(\beta,\boldeta)\tp g(\beta,\boldeta;\O_i))g(\beta,\boldeta;\O_i)g^{'}(\boldeta;\O_i)\tp/n$, $\calV_{gg^{'}}^{v}(\beta,\boldeta)=\sum_{i=1}^n\rho^{''}(v)g^{'}(\boldeta;\O_i)g^{'}(\boldeta;\O_i)\tp/n$, and $\calV_{g^{'}g^{'}}(\beta,\boldeta)=\sum_{i=1}^n\rho^{''}(\blambda(\beta,\boldeta)\tp g(\beta,\boldeta;\O_i))\\g^{'}(\boldeta;\O_i)g^{'}(\boldeta;\O_i)\tp/n$. 
Then, we have 
\be\label{partialqbb}
&&\frac{\partial^2 \hat{Q}(\beta, \boldeta)}{\partial \beta^2}\n\\
&=&\frac{1}{n}\sum_{i=1}^n \rho^{''}(\blambda(\beta,\boldeta)\tp g(\beta,\boldeta;\O_i))\{\blambda(\beta,\boldeta)\tp g^{'}(\boldeta;\O_i)\}\{\blambda(\beta,\boldeta)\tp g^{'}(\boldeta;\O_i)+\frac{\partial\blambda(\beta,\boldeta)\tp}{\partial\beta}g(\beta,\boldeta;\O_i)\}\n\\
&&+ \frac{1}{n}\sum_{i=1}^n \rho^{'}(\blambda(\beta,\boldeta)\tp g(\beta,\boldeta;\O_i))\frac{\partial\blambda(\beta,\boldeta)\tp}{\partial\beta}g^{'}(\boldeta;\O_i)\n\\
&=&\blambda(\beta,\boldeta)\tp\calV_{g^{'}g^{'}}(\beta,\boldeta)\blambda(\beta,\boldeta)+\blambda(\beta,\boldeta)\tp\calV_{gg^{'}}(\beta,\boldeta)\frac{\partial\blambda(\beta,\boldeta)}{\partial\beta}\n\\
&&+\frac{1}{n}\sum_{i=1}^n\{-1+\rho^{''}(\bar v_1)\blambda(\beta,\boldeta)\tp g(\beta,\boldeta;\O_i)\}\frac{\partial\blambda(\beta,\boldeta)\tp}{\partial\beta}g^{'}(\boldeta;\O_i)\n\\
&=&\blambda(\beta,\boldeta)\tp\calV_{g^{'}g^{'}}(\beta,\boldeta)\blambda(\beta,\boldeta)+\blambda(\beta,\boldeta)\tp\calV_{gg^{'}}(\beta,\boldeta)\frac{\partial\blambda(\beta,\boldeta)}{\partial\beta}\n\\
&&+\blambda(\beta,\boldeta)\tp\calV_{gg^{'}}^{\bar v_1}(\beta,\boldeta)\frac{\partial\blambda(\beta,\boldeta)}{\partial\beta}-\frac{\partial\blambda(\beta,\boldeta)\tp}{\partial\beta}\bar {g^{'}} (\boldeta)\n\\
&\equiv& D_1(\boldeta)+D_2(\boldeta)+D_3(\boldeta)+D_4(\boldeta),
\ee
where $|\bar v_1|\leq |\blambda(\beta,\hat\boldeta)\tp g(\beta,\hat\boldeta;\O_i)|$, and the derivative of $\blambda$ can be calculated by the implicit function theorem:
\bse
\frac{\partial\blambda(\beta,\boldeta)}{\partial \beta}&=&-\calV_{gg}(\beta,\boldeta)^{-1}\{\calV_{gg^{'}}(\beta,\boldeta)\blambda(\beta,\boldeta)+\frac{1}{n}\sum_{i=1}^n\rho^{'}(\blambda(\beta,\boldeta)\tp g(\beta,\boldeta;\O_i))g^{'}(\boldeta;\O_i)\}\\
&=&-\calV_{gg}(\beta,\boldeta)^{-1}[\{\calV_{gg^{'}}(\beta,\boldeta)+\calV_{gg^{'}}^{\bar v_2}(\beta,\boldeta)\}\blambda(\beta,\boldeta)-g^{'}(\boldeta;\O_i)].
\ese
Now we will prove $\sup_{\beta\in\calB}|D_i(\hat\boldeta)-D_i(\boldeta_0)|=o_p(\mu_n^2/n)$ for $i=1,2,3,4$. According to Taylor expansion and \eqref{suplambdag}, we have
\bse
&&\|\calV_{gg}(\beta,\boldeta_0)+\bar\Omega(\beta,\boldeta_0)\|\\
&=&\|\frac{1}{n}\sum_{i=1}^n\{\rho^{''}(\blambda(\beta,\boldeta)\tp g(\beta,\boldeta;\O_i))-\rho^{''}(0)\}g(\beta,\boldeta;\O_i)g(\beta,\boldeta;\O_i)\tp\|\\
&=&\|\frac{1}{n}\sum_{i=1}^n\rho^{'''}(x)\blambda(\beta,\boldeta_0)\tp g(\beta,\boldeta_0;\O_i) g(\beta,\boldeta_0;\O_i)g(\beta,\boldeta_0;\O_i)\tp\|\\
&\leq&c\max_{i\leq n}|\blambda(\beta,\boldeta_0)\tp g(\beta,\boldeta_0;\O_i)|\frac{1}{n}\sum_{i=1}^n\|g(\beta,\boldeta_0;\O_i)g(\beta,\boldeta_0;\O_i)\tp\| \\
&=&o_p(1)O_p(1)\\
&=&o_p(1),
\ese
where $|x|\leq |\blambda(\beta,\boldeta_0)\tp g(\beta,\boldeta_0;\O_i)|$. Similarly, we can prove that $ \|\calV_{gg}(\beta,\hat\boldeta)+\bar\Omega(\beta,\hat\boldeta)\|=o_p(1)$. Combining with Lemma \ref{lemmaomega}, it can be proved that
\bse
&&\|\calV_{gg}(\beta,\hat\boldeta)-\calV_{gg}(\beta,\boldeta_0)\|\\
&\leq&\|\calV_{gg}(\beta,\hat\boldeta)+\bar\Omega(\beta,\hat\boldeta)\|+ \|\bar\Omega(\beta,\hat\boldeta)-\bar\Omega(\beta,\boldeta_0)\| +\|\calV_{gg}(\beta,\boldeta_0)+\bar\Omega(\beta,\boldeta_0)\|\\
&=&o_p(1).
\ese
Similarly, we can prove $\|\calV_{gg^{'}}(\beta,\hat\boldeta)-\calV_{gg^{'}}(\beta,\boldeta_0)\|=o_p(1)$, $\|\calV_{gg^{'}}^v(\beta,\hat\boldeta)-\calV_{gg^{'}}^v(\beta,\boldeta_0)\|=o_p(1)$, and $\|\calV_{g^{'}g^{'}}(\beta,\hat\boldeta)-\calV_{g^{'}g^{'}}(\beta,\boldeta_0)\|=o_p(1)$.
So we have
\bse
&&\sup_{\beta\in\calB}|D_1(\hat\boldeta)-D_1(\boldeta_0)|\\
&=&\sup_{\beta\in\calB}|\{\blambda(\beta,\hat\boldeta)-\blambda(\beta,\boldeta_0)\}\tp\calV_{g^{'}g^{'}}(\beta,\hat\boldeta)\{\blambda(\beta,\hat\boldeta)+\blambda(\beta,\boldeta_0)\}|\\
&&+\sup_{\beta\in\calB}|\blambda(\beta,\hat\boldeta)\tp\{\calV_{g^{'}g^{'}}(\beta,\hat\boldeta)-\calV_{g^{'}g^{'}}(\beta,\boldeta_0)\}\blambda(\beta,\boldeta_0)|\\
&=&o_p(\mu_n/\sqrt{n})O_p(1)O_p(\mu_n/\sqrt{n})+O_p(\mu_n/\sqrt{n})o_p(1)O_p(\mu_n/\sqrt{n})\\
&=&o_p(\mu_n^2/n).
\ese
For $D_2$ and $D_3$, if we can prove $\sup_{\beta\in\calB}\|\partial \blambda(\beta,\hat\boldeta)/\partial \beta-\partial \blambda(\beta,\boldeta_0)/\partial \beta\|=o_p(\mu_n/\sqrt{n})$ and $ \sup_{\beta\in\calB}\|\partial \blambda(\beta,\boldeta)/\partial \beta\|=O_p(\mu_n/\sqrt{n})$, then we can just replace $\blambda$ with $\partial \blambda/\partial \beta$ in the proof of $D_1$ and obtain the desired result. Next, we will prove the above two equations. Based on Lemma \ref{lemmag} and \ref{lemmalambda}, we have
\bse
&&\sup_{\beta\in\calB}\|\frac{\partial \blambda(\beta,\boldeta)}{\partial \beta}\|\\
&\leq& \sup_{\beta\in\calB}\|\calV_{gg}(\beta,\boldeta)^{-1}\{\calV_{gg^{'}}(\beta,\boldeta)+\calV_{gg^{'}}^{\bar v_2}(\beta,\boldeta)\}\blambda(\beta,\boldeta)\|+\sup_{\beta\in\calB}\|\calV_{gg}(\beta,\boldeta)^{-1}g^{'}(\boldeta;\O_i)\|\\
&=&O_p(1)O_p(\mu_n/\sqrt{n})+O_p(1)O_p(\mu_n/\sqrt{n})\\
&=&O_p(\mu_n/\sqrt{n}).
\ese
Similarly, we can also prove $\sup_{\beta\in\calB}\|\partial \blambda(\beta,\hat\boldeta)/\partial \beta-\partial \blambda(\beta,\boldeta_0)/\partial \beta\|=o_p(\mu_n/\sqrt{n})$, which implies $\sup_{\beta\in\calB}|D_i(\hat\boldeta)-D_i(\boldeta_0)|=o_p(\mu_n^2/n)$ for $i=2,3$. For $D_4$, we have
\bse
&&\sup_{\beta\in\calB}|D_4(\hat\boldeta)-D_4(\boldeta_0)|\\
&\leq&\sup_{\beta\in\calB}|\frac{\partial\blambda(\beta,\hat\boldeta)\tp}{\partial\beta}\{g^{'}(\hat\boldeta)-g^{'}(\boldeta_0)\}|+\sup_{\beta\in\calB}|\{\frac{\partial\blambda(\beta,\hat\boldeta)\tp}{\partial\beta}-\frac{\partial\blambda(\beta,\boldeta_0)\tp}{\partial\beta}\} g^{'}(\boldeta_0)|\\
&=&o_p(\mu_n^2/n).
\ese
Combining $D_1$ to $D_4$, we have $\sup_{\beta\in\calB}|{\partial^2 \hat{Q}(\beta, \hat\boldeta)}/{\partial \beta^2}-{\partial^2 \hat{Q}(\beta, \boldeta_0)}/{\partial \beta^2}|=o_p(\mu_n^2/n)$. Note that Lemma 13 in \cite{newey2009} gives that $\sup_{\beta\in\calB}n/\mu_n^2|{\partial^2 \hat{Q}(\beta, \boldeta_0)}/{\partial \beta^2}|\to H$, which implies
\be\label{Taylor2}
\sup_{\beta\in\calB}\frac{n}{\mu_n^2}|\frac{\partial^2 \hat{Q}(\beta, \hat\boldeta)}{\partial \beta^2}|\to H.
\ee 
Finally, based on \eqref{Taylor}, \eqref{Taylor1}, and \eqref{Taylor2}, the asymptotic normality can be established as  
$$
\frac{\mu_n\left(\hat{\beta}-\beta_0\right)}{\sqrt{H^{-2}(H+V_4)}}\xrightarrow{d} N(0,1).
$$
\end{proof}

\subsection{Proof of Theorem \ref{thmcensor}} 

First, we will prove the consistency of $\hat\beta$. Since $\psi$ is a linear function of $\beta$, we have
$$
\|\tilde\psi(\beta,G_0,\boldeta_0)-\tilde\psi(\beta_0,G_0,\boldeta_0)\|=|\beta-\beta_0|\|\psi^{'}(G_0,\boldeta_0)\|=|\beta-\beta_0|\|g^{'}(\boldeta_0)\|,
$$
which implies $|\beta-\beta_0|\leq c\mu_n/\sqrt{n}\|\tilde\psi(\beta,G_0,\boldeta_0)\|$. Then it is sufficient to prove $\mu_n/\sqrt{n}\|\tilde\psi(\beta,G_0,\boldeta_0)\|=o_p(1)$. Note that $\tilde Q(\beta_0,G_0,\boldeta_0)=m/(2n)$ and $\hat Q(\hat\beta,\hat G ,\hat\boldeta)\leq\hat Q(\beta_0,\hat G ,\hat\boldeta)$, then according to Lemma \ref{lemmaqc}
\bse 
&&\frac{n}{\mu_n^2}\|\tilde\psi(\beta,G_0,\boldeta_0)\|^2\\
&\leq& \frac{n}{\mu_n^2}\|\tilde \psi(\hat\beta,G_0,\boldeta_0) \|^2\\
&\leq&c \frac{n}{\mu_n^2}\tilde \psi(\hat\beta,G_0,\boldeta_0)\tp \Sigma(\hat\beta,G_0,\boldeta_0)^{-1} \tilde \psi(\hat\beta,G_0,\boldeta_0)+m/(2n)-m/(2n)\\
&\leq&c \frac{n}{\mu_n^2}\{ \tilde Q(\hat\beta,G_0,\boldeta_0)-\tilde Q(\beta_0,G_0,\boldeta_0)\}\\
&\leq& c \frac{n}{\mu_n^2}\{\tilde Q(\hat\beta,G_0,\boldeta_0)-\hat Q(\hat\beta,\hat G ,\hat\boldeta)+\hat Q(\hat\beta,\hat G ,\hat\boldeta)-\hat Q(\beta_0,\hat G ,\hat\boldeta)+\hat Q(\beta_0,\hat G ,\hat\boldeta)-\tilde Q(\beta_0,G_0,\boldeta_0)\}\\
&=&c\frac{n}{\mu_n^2}o_p(\frac{\mu_n^2}{n})\\
&=&o_p(1).
\ese
Then we will prove the asymptotic normality. According to the Taylor expansion, we have 
\be \label{Taylor_censor}
0&=& \frac{n}{\mu_n}\frac{\partial \hat Q(\hat\beta,\hat G ,\hat\boldeta)}{\partial \beta}=\frac{n}{\mu_n}\frac{\partial\hat Q(\beta_0, \hat G ,\hat{\boldeta})}{\partial\beta}+\frac{n}{\mu_n^2}\frac{\partial^2\hat Q(\bar\beta,\hat G ,\hat\boldeta)}{\partial \beta^2}\mu_n(\hat{\beta}-\beta_0).
\ee  
Similar to the proof of Corollary \ref{coro1}, we can obtain the result with the nuisance functions being their true value:
$$
\frac{n}{\mu_n}\left.\frac{\partial \hat{Q}(\beta, G_0,{\boldeta_0})}{\partial \beta}\right|_{\beta={\beta_0}} \to N(0,H+V_1),
$$
and 
$$
\frac{n}{\mu_n^2}\left.\frac{\partial^2 \hat{Q}(\beta,G_0,{\boldeta_0})}{\partial \beta^2}\right|_{\beta=\beta_{0}}\to H.
$$
The difference between the first order term can be decomposed as
\be \label{Qfirstorder}
&&\frac{n}{\mu_n}\frac{\partial\hat Q(\beta_0, \hat G ,\hat{\boldeta})}{\partial\beta}-\frac{n}{\mu_n}\frac{\partial\hat Q(\beta_0, G_0 ,{\boldeta}_0)}{\partial\beta}\n\\
&=&\frac{n}{\mu_n}D_{\boldeta}\frac{\partial\hat Q(\beta_0,G_0,\boldeta_0)}{\partial \beta}(\hat{\boldeta}-\boldeta_0)+\frac{n}{\mu_n}D_{G}\frac{\partial\hat Q(\beta_0,G_0,\boldeta_0)}{\partial\beta}(\hat G -G_0)+o_p(1).
\ee
The reminder term is $o_p(1)$ because in the second order expansion, $\partial \psi/\partial \beta=O_p(\mu_n/\sqrt{n}) $ by Assumption \ref{condweak} and $\|\widehat\boldeta-\boldeta_0\|^2_{L^2}=o_p(n^{-1/2})$ and $\|\widehat G-G_0\|^2_{L^2}=o_p(n^{-1/2})$.
For the first term in \eqref{Qfirstorder}, we can prove $$\frac{n}{\mu_n}D_{\boldeta}\frac{\partial\hat Q(\beta_0,G_0,\boldeta_0)}{\partial\beta}(\hat\boldeta-\boldeta_0)=\frac{n}{\mu_n}\frac{\partial\hat Q(\beta_0,G_0,\boldeta_0)}{\partial \beta}-\frac{n}{\mu_n}\frac{\partial\hat Q(\beta_0,G_0,\hat\boldeta)}{\partial \beta}+o_p(1)=o_p(1).$$ 
Equation \eqref{qda123} gives the decomposition of ${\partial\hat Q(\beta_0,G_0,\boldeta_0)}/{\partial \beta}$:
\be \label{partialQbetacensor}
&&\left.\frac{\partial \hat{Q}(\beta,G,{\boldeta})}{\partial \beta}\right|_{\beta=\beta_0}\n\\
&=&\bar \psi^{'}(G,\boldeta)\tp\bar\Sigma(\beta_0,G,\boldeta)^{-1}\bar \psi(\beta_0,G,\boldeta)\n\\
&&-\bar \psi(\beta_0,G,\boldeta)\tp\bar\Sigma(\beta_0,G,\boldeta)^{-1}\left\{\frac{1}{n}\sum_{i=1}^{n}\psi(\beta_0,G,\boldeta;\O_i)\psi^{'}(G,\boldeta;\O_i)\tp\right\}\bar\Sigma(\beta_0,G,\boldeta)^{-1}\bar \psi(\beta_0,G,\boldeta)\n\\
&&+o_p({\mu_n}/{n})\n.
\ee
The decomposition has two terms. for the first term, based on Lemma \ref{lemmapsi} and \ref{lemmasigma}, we have
\be
&&\bar \psi^{'}(G_0,\boldeta_0)\tp\bar\Sigma(\beta_0,G_0,\boldeta_0)^{-1}\bar \psi(\beta_0,G_0,\boldeta_0)-\bar \psi^{'}(G_0,\hat\boldeta)\tp\bar\Sigma(\beta_0,G_0,\hat\boldeta)^{-1}\bar \psi(\beta_0,G_0,\hat\boldeta)\n\\
&=&\bar\psi^{'}(G_0,\boldeta_0)\tp\bar\Sigma(\beta_0,G_0,\boldeta_0)^{-1}\{\bar \psi(\beta_0,G_0,\boldeta_0)-\bar \psi(\beta_0,G_0,\hat\boldeta)\}\n\\
&&+ \bar\psi^{'}(G_0,\boldeta_0)\tp\{\bar\Sigma(\beta_0,G_0,\boldeta_0)^{-1}-\bar\Sigma(\beta_0,G_0,\hat\boldeta)^{-1}\}\bar \psi(\beta_0,G_0,\hat\boldeta)\n\\
&&+\{\bar\psi^{'}(G_0,\boldeta_0)-\bar \psi^{'}(G_0,\hat\boldeta)\}\tp\bar\Sigma(\beta_0,G_0,\hat\boldeta)^{-1}\bar \psi(\beta_0,G_0,\hat\boldeta)\n\\
&=&o_p(\mu_n/n).\n
\ee
Similar to the above proof, we can prove 
\small{\bse
&&\bar \psi(\beta_0,G_0,\boldeta_0)\tp\bar\Sigma(\beta_0,G_0,\boldeta_0)^{-1}\left\{\frac{1}{n}\sum_{i=1}^{n}\psi(\beta_0,G_0,\boldeta_0;\O_i)\psi^{'}(G_0,\boldeta_0;\O_i)\tp\right\}\bar\Sigma(\beta_0,G_0,\boldeta_0)^{-1}\bar \psi(\beta_0,G_0,\boldeta_0)\\
&-&\bar \psi(\beta_0,G_0,\hat\boldeta)\tp\bar\Sigma(\beta_0,G_0,\hat\boldeta)^{-1}\left\{\frac{1}{n}\sum_{i=1}^{n}\psi(\beta_0,G_0,\hat\boldeta;\O_i)\psi^{'}(G_0,\hat\boldeta;\O_i)\tp\right\}\bar\Sigma(\beta_0,G_0,\hat\boldeta)^{-1}\bar \psi(\beta_0,G_0,\hat\boldeta)\\
&&=o_p(\mu_n/n).
\ese}
Thus, $n/\mu_n{\hat Q(\beta_0,G_0,\boldeta_0)}/{\partial\beta}-n/\mu_n\hat Q(\beta_0,G_0,\hat\boldeta)/{\partial\beta} = o_p(1)$. 

Next, we will find the asymptotic property of $n/\mu_nD_G \partial\hat Q(\beta_0,G_0,\boldeta_0)/\partial \beta(\hat G -G_0)$. Different from $\boldeta$, $G$ is a non-Neyman orthogonal nuisance, thus special attention must be paid to the asymptotic behavior of this term. Theorem 2.3 in \cite{gonzalez1994} gives
\be\label{eq:ghatg0}
\hat G (T|\O_A)-G_0(T|\O_A) = \frac{1}{n}\sum_{j=1}^n  \phi(\O_{T,j}^{'},\O_T)+O_p((\frac{\log(n)}{nh^{m+d_x+1}})^{3/4}+h^{\lfloor 3(m+d_x+1)/2\rfloor+1}),
\ee 
where 
$$
\phi(\O_{T,j}^{'},\O_T) =B_{nj}(\O_T)G_0(T|\O_A)\left\{\frac{I(T^{'}\leq T,\delta^{'}=0)}{P(T^{'}<T|\O_A)}-\int_0^{\min\{T^{'},T\}}\frac{d P(T^{'}\leq s, \delta^{'}=0)}{P(T^{'}<s|\O_A)^2}\right\},
$$
and 
$$
\bbE \phi(\O_{T,j}^{'},\O_T) =0.
$$
Note that Theorem 2.3 in \cite{gonzalez1994} is about the distribution of $T$, but here we need the distribution of $C$. Due to the symmetry between $T$ and $C$, by replacing $\delta = 1$ in the original theorem with $\delta = 0$, the distribution of $C$ can be obtained. In addition, since we use the higher order kernel $K$ in Assumption \ref{condkernel}, the last term in \eqref{eq:ghatg0} is $h^{\lfloor 3(m+d_x+1)/2\rfloor+1}$ instead of $h^2$. This ensures that the reminder term is $o_p(n^{-1/2})$ by assigning proper $h$, such as $h=n^{-1/{3(m+d_x+1)+1}}$. 
Define 
$$
\calL_{\delta}=\{D_G \frac{\sqrt{n}}{\mu_n}\frac{\partial\rho(\beta_0,G_0,\boldeta_0;\O)}{\partial \beta}(G-G_0): \|G-G_0\|_{L^2}\leq\delta\}.
$$
Then there is a constant $c$ such that $\sigma^2=c\kappa^2$ and  $\sup_{l\in\calL_\delta}\bbE l^2\leq  n/\mu_n^2 \bbE\|\rho^{'}(\blambda\tp g(\beta_0,\boldeta_0;\O))\blambda\tp$ $\delta(g^{'}(\boldeta_0;\O)-\xi^{'}(\boldeta_0;\O))/G_0^2 \|^2\|G-G_0\|_{L^2}^2 \leq c n/\mu_n^2\|\blambda\|^2\kappa^2 \lesssim c\kappa^2 = \sigma ^2$ due to Assumptions \ref{condweak}, \ref{condkernel} and Lemma \ref{lemmalambdacensor}. In addition, Lemma 14 in \cite{kim2019} ensures that the covering number of a kernel density estimation can be bounded by $\sup_{Q_n}\calN(\epsilon\|F\|,\calL_1,\|\cdot\|_{2,Q_n}) \leq c(n/\epsilon)^m$. Similar to the proof of \eqref{pn-p} in Lemma \ref{lemmag}, we can obtain $\sqrt{n}(\bbE_n-\bbE)D_G \sqrt{n}/\mu_n\partial \rho/\partial \beta(\hat G-G_0)=O_p(\kappa\sqrt{m\log(n)}+m\log(n)/\sqrt{n})=o_p(1)$.
Then we have  
\be \label{normalG}
&&D_G\frac{n}{\mu_n}\frac{\partial\hat Q(\beta_0,G_0,\boldeta_0)}{\partial\beta}(\hat G -G_0)\n\\
&=&(\bbE_n-\bbE_{P_{\O_T}}+\bbE_{P_{\O_T}}) D_G\frac{n}{\mu_n}\frac{\partial\rho(\blambda\tp\psi(\beta_0,G_0,\boldeta_0;\O_{T}))}{\partial\beta}(\hat G -G_0)\n\\
&=&\bbE_{P_{\O_T}} D_G\frac{n}{\mu_n}\frac{\partial\rho(\blambda\tp\psi(\beta_0,G_0,\boldeta_0;\O_{T}))}{\partial\beta}(\hat G -G_0)+o_p(1)\n\\
&=&\bbE_{P_{\O_T}} D_G\frac{n}{\mu_n}\frac{\partial\rho(\blambda\tp\psi(\beta_0,G_0,\boldeta_0;\O_{T}))}{\partial\beta}\{\frac{1}{n}\sum_{i=1}^n\phi(\O_{T,i}^{'},\O_{T})+o_p(n^{-1/2})\}+o_p(1)\n\\
&=&\frac{1}{n}\sum_{i=1}^n \bbE_{P_{\O_T}} [D_G\frac{n}{\mu_n}\frac{\partial\rho(\blambda\tp\psi(\beta_0,G_0,\boldeta_0;\O_{T}))}{\partial\beta}\phi(\O^{'}_{T,i},\O_{T})]+o_p(1).
\ee 
Note that $\O_T$ is independent of $\O_T^{'}$, and is integrated out by $\bbE_{P_{\O_T}}$, thus the last line is an independent average of $n$ random variables with mean 0. As shown in \cite{stute1995,gerds2017}, such doubly integrated Kaplan-Meier quantity has $\sqrt{n}$ normality with zero expectation under Assumption \ref{condkernel}. This allows us to treat the summand in \eqref{normalG} as the $X_i$ term in Lemma A10 of \cite{newey2009}. 

Lemma A12 in \cite{newey2009} gives that 
\be\label{eqQ0}
&&\frac{n}{\mu_n}\left.\frac{\partial \hat{Q}(\beta, G_0,{\boldeta_0})}{\partial \beta}\right|_{\beta={\beta_0}}\n\\
&=&\frac{1}{n}\sum_{i=1}^n \frac{n}{\mu_n}\{\psi^{*}\Sigma_0^{-1}\psi(\beta_0,G_0,{\boldeta_0};\O_i)+\frac{1}{n}\sum_{i=1}^n U_{i}^{C\top}\Sigma_0^{-1}\psi(\beta_0,G_0,{\boldeta_0};\O_i) \}+o_p(1)\n\\
&\to& N(0,H+V_1).
\ee
Combining \eqref{normalG}, \eqref{eqQ0}, and Lemma A10 in \cite{newey2009}, \eqref{Qfirstorder} becomes
\be
&&\frac{n}{\mu_n}\frac{\partial\hat Q(\beta_0, \hat G ,\hat{\boldeta})}{\partial\beta}\n\\
&=& \frac{1}{n}\sum_{i=1}^n \frac{n}{\mu_n}\{\psi^{*}\Sigma_0^{-1}\psi(\beta_0,G_0,{\boldeta_0};\O_i)+\frac{1}{n}\sum_{i=1}^n U_{i}^{C\top}\Sigma_0^{-1}\psi(\beta_0,G_0,{\boldeta_0};\O_i) \}\n\\
&&+\frac{1}{n}\sum_{i=1}^n \bbE_{P_{\O_T}} [D_G\frac{n}{\mu_n}\frac{\partial\rho(\blambda\tp\psi(\beta_0,G_0,\boldeta_0;\O_{T}))}{\partial\beta}\phi(\O_{T,i}^{'},\O_{T})]+o_p(1)\n\\
&\to & N(0,H+V_1+V_2),
\ee
where $V_2 = \bbE_{P_{\O_T^{'}}}(\bbE_{P_{\O_T}}D_G\{n/\mu_n\partial \rho(\blambda\tp\psi(\beta_0,G_0,\boldeta_0;\O_{T}))/\partial \beta\allowbreak\phi(\O_{T}^{'},\O_{T})\})^2 $\\ $+2\bbE_{P_{\O_T^{'}}}(\bbE_{P_{\O_T}}\{D_Gn^2/\mu_n^2\partial\rho(\blambda\tp\psi(\beta_0,G_0,\boldeta_0;\O_{T}))/\partial\beta\phi(\O_{T}^{'},\O_{T})\}\{\psi^{*}\Sigma_0^{-1}\psi(\beta_0,G_0,{\boldeta_0};\O_{T})\})$ includes two parts: a variance term and a covariace term between $\psi^{*}\Sigma_0^{-1}\psi$ and $\phi$. Note that the covariace term related to $U_{i}^{C\top}\Sigma_0^{-1}\psi$ is 0 from Lemma A10 \cite{newey2009}.
Next, we will prove
$$
|\frac{n}{\mu_n^2}\frac{\partial^2\hat Q(\bar\beta,\hat G ,\hat\boldeta)}{\partial \beta^2}-\frac{n}{\mu_n^2}\frac{\partial^2\hat Q(\bar\beta,G_0,\boldeta_0)}{\partial \beta^2}|=o_p(1).
$$
Similar to \eqref{partialqbb}, we have
\bse
&&\frac{\partial^2 \hat{Q}(\bar\beta, G, \boldeta)}{\partial \beta^2}\\
&=&\frac{1}{n}\sum_{i=1}^n \rho^{''}(\blambda(\bar\beta ,G,\boldeta)\tp \psi(\bar\beta ,G,\boldeta;\O_i))\blambda(\bar\beta ,G,\boldeta)\tp \psi^{'}(G,\boldeta;\O_i)\psi^{'}(G,\boldeta;\O_i)\tp\blambda(\bar\beta ,G,\boldeta) \\
&&+\frac{1}{n}\sum_{i=1}^n \rho^{''}(\blambda(\bar\beta ,G,\boldeta)\tp \psi(\bar\beta ,G,\boldeta;\O_i))\blambda(\bar\beta ,G,\boldeta)\tp \psi^{'}(G,\boldeta;\O_i)\frac{\partial\blambda(\bar\beta ,G,\boldeta)\tp}{\partial\beta}\psi(\bar\beta ,G,\boldeta;\O_i)\\
&&+ \frac{1}{n}\sum_{i=1}^n \rho^{'}(\blambda(G,\bar\beta ,\boldeta)\tp \psi(\bar\beta ,G,\boldeta;\O_i))\frac{\partial\blambda(\bar\beta ,G,\boldeta)\tp}{\partial\beta}\psi^{'}(G,\boldeta;\O_i)\\
&=&(I_1+I_2+I_3)(\bar\beta , G, \boldeta).
\ese
For $I_1$, we have
\bse    
&&I_1(\bar\beta , \hat G,\hat\boldeta)\\
&=& \frac{1}{n}\sum_{i=1}^n(\rho^{''}(0)+ \rho^{'''}(\bar x)\blambda(\bar\beta ,\hat G,\hat\boldeta)\tp \psi(\hat G,\hat\boldeta;\O_i))\blambda(\bar\beta ,\hat G,\hat\boldeta)\tp \psi^{'}(\hat G,\hat\boldeta;\O_i)\psi^{'}(\hat G,\hat\boldeta;\O_i)\tp\blambda(\bar\beta ,\hat G,\hat\boldeta)\\
&=&-\frac{1}{n}\sum_{i=1}^n\blambda(\bar\beta ,\hat G,\hat\boldeta)\tp \psi^{'}(\hat G,\hat\boldeta;\O_i)\psi^{'}(\hat G,\hat\boldeta;\O_i)\tp\blambda(\bar\beta ,\hat G,\hat\boldeta)+R,
\ese
where the remainder term $R$ satisfies 
\be
R&=&\frac{1}{n}\sum_{i=1}^n\rho^{'''}(\bar x)\blambda(\bar\beta ,\hat G,\hat\boldeta)\tp \psi(\hat G,\hat\boldeta;\O_i)\blambda(\bar\beta ,\hat G,\hat\boldeta)\tp \psi^{'}(\hat G,\hat\boldeta;\O_i)\psi^{'}(\hat G,\hat\boldeta;\O_i)\tp\blambda(\bar\beta ,\hat G,\hat\boldeta)\n\\
&\leq& c\max_{i\leq n}\sup_{\beta\in\mathcal{B}}\|\psi(\beta ,\hat G,\hat\boldeta;\O_i)\|\|\blambda(\bar\beta ,\hat G,\hat\boldeta)\|^3\|\frac{1}{n}\sum_{i=1}^n \psi^{'}(\hat G,\hat\boldeta;\O_i)\psi^{'}(\hat G,\hat\boldeta;\O_i)\tp  \|\n\\
&=&o_p(\mu_n^2/n).
\ee
Then Lemmas \ref{lemmapsi}, \ref{lemmapsi-G}, and \ref{lemmalambdacensor} give that
\be
&&I_1(\bar\beta, \hat G, \hat\boldeta)-I_1(\bar\beta, G_0, \boldeta_0) \n\\
&=&\blambda(\bar\beta ,\hat G,\hat\boldeta)\tp\left(\frac{1}{n}\sum_{i=1}^n \psi^{'}( G_0,\boldeta_0;\O_i)\psi^{'}(G_0,\boldeta_0;\O_i)\tp-\frac{1}{n}\sum_{i=1}^n\psi^{'}(\hat G,\hat\boldeta;\O_i)\psi^{'}(\hat G,\hat\boldeta;\O_i)\tp\right)\blambda(\bar\beta ,\hat G,\hat\boldeta)\n\\
&&+\left(\blambda(\bar\beta ,G_0,\boldeta_0)-\blambda(\bar\beta ,\hat G,\hat\boldeta)\right)\tp \frac{1}{n}\sum_{i=1}^n \psi^{'}( G_0,\boldeta_0;\O_i)\psi^{'}(G_0,\boldeta_0;\O_i)\tp\left(\blambda(\bar\beta ,G_0,\boldeta_0)+\blambda(\bar\beta ,\hat G,\hat\boldeta)\right)\n\\
&=&O_p(\mu_n^2/n)o_p(1)+o_p(\mu_n/\sqrt{n})O_p(1)O_p(\mu_n/\sqrt{n})\n\\
&=&o_p(\mu_n^2/n).
\ee
Therefore, we prove $I_1(\bar\beta, \hat G, \hat\boldeta)-I_1(\bar\beta, G_0, \boldeta_0) =o_p(\mu_n^2/n)$, and the other two terms can be proved in the same way. Thus 
$$
\frac{n}{\mu_n^2}\frac{\partial^2\hat Q(\bar\beta,\hat G ,\hat\boldeta)}{\partial \beta^2}=\frac{n}{\mu_n^2}\frac{\partial^2\hat Q(\bar\beta,G_0 ,\boldeta_0)}{\partial \beta^2}+o_p(1)\to H,
$$
and \eqref{Taylor_censor} becomes 
\be 
&&\frac{n}{\mu_n^2}\frac{\partial^2\hat Q(\bar\beta,\hat G ,\hat\boldeta)}{\partial \beta^2}\mu_n(\hat{\beta}-\beta_0)=-\frac{n}{\mu_n}\frac{\partial\hat Q(\beta_0, \hat G,{\hat\boldeta})}{\partial\beta}+o_p(1)\n.
\ee  
Thus we obtain the asymptotic normality of $\widehat\beta$:
$$
\mu_n(\hat{\beta}-\beta_0)\xrightarrow{d} N(0,H^{-2}(H+V_1+V_2)).
$$

\subsection{Remarks on Theorem \ref{thmcensor}: Variance decomposition and impact of non-Neyman orthogonal nuisance estimation}\label{secvardecom}

We now analyze the asymptotic variance of the GEL estimator  $\hat{\beta}_{\textrm{GEL2}}^{\textrm{C}}$. Specifically, the variance of  $\hat{\beta}_{\textrm{GEL2}}^{\textrm{C}}$ can be decomposed into three main components: (i) a classical GMM-type variance term $H$ \citep{hansen1982}; (ii) an additional term $V_1$ that captures the contribution of weak identification \citep{newey2009}; and (iii) a non-Neyman orthogonal induced term $V_2$ that arises from estimating non-Neyman orthogonal nuisance functions. The $V_1$ captures the additional uncertainty from weak moment functions. GEL accounts for this component, which is negligible under strong identification but remains non-vanishing with many weak moments, making its inclusion essential for valid inference.

Figure~\ref{fig:moment_convergence} illustrates the main differences between our method and the classical GMM \citep{hansen1982}. The left panel shows that our estimator accommodates both weak moment conditions $(\mu_n < \sqrt{n})$ and strong moment conditions $(\mu_n=\sqrt{n})$. Thus, our convergence rate is adaptively determined by the moment strength rather than being fixed. The right panel shows that our method also allows the number of moments $m$ to diverge while preserving the convergence rate, while classical two-stage least squares estimators typically rely on a fixed or finite set of valid moments \citep{hansen1982}. This flexibility helps mitigate weak moment problems by incorporating many candidate instruments.

\begin{figure}[htbp]
\centering
\includegraphics[width=0.8\textwidth]{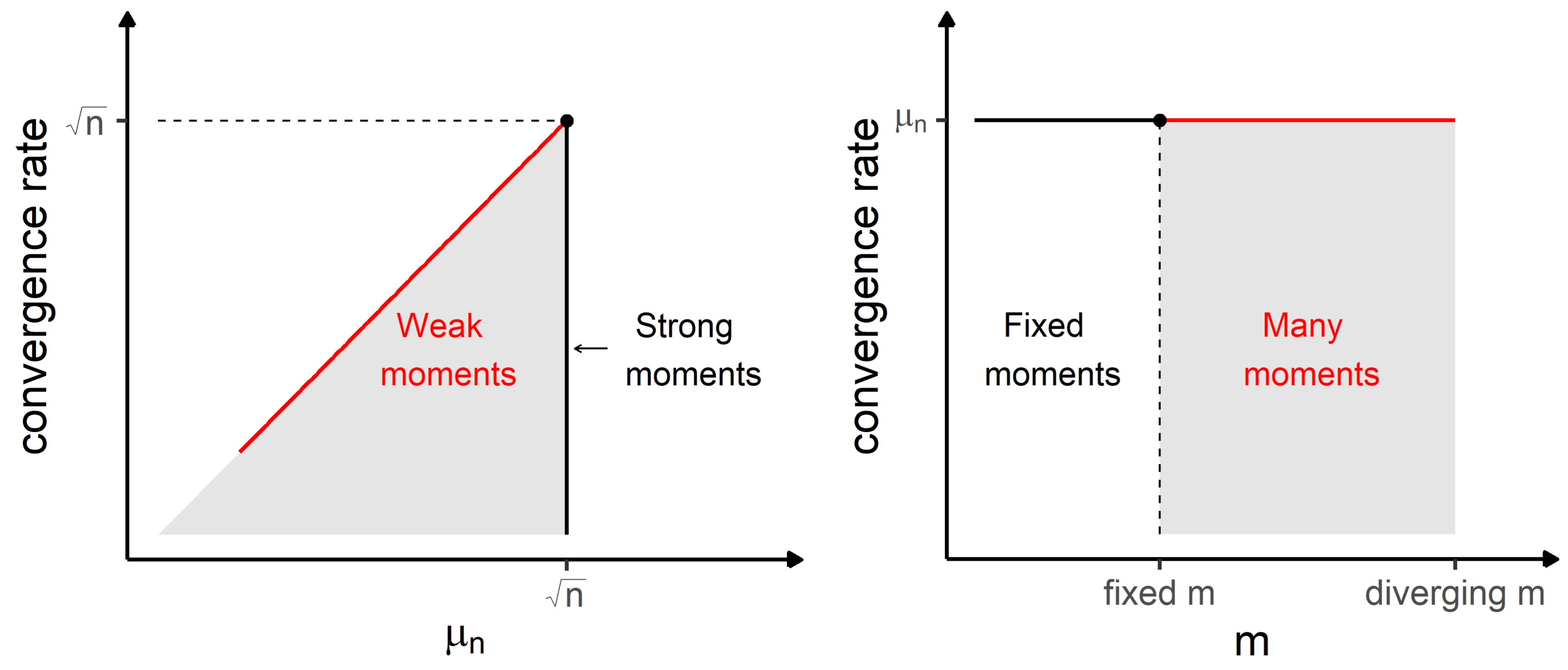}
\caption{Comparison of our proposed method to the classical GMM with strong moments \citep{hansen1982}. The left figure shows that the proposed estimator (red) attains the convergence rate over both the weak moment condition $(\mu_{n}<\sqrt{n})$ and the strong moment condition $\mu_{n}=\sqrt{n}$, while the classical GMM estimator (black) works only at $(\mu_{n}=\sqrt{n})$. The right figure shows that our method allows the number of moment conditions $m$ to diverge and at the same time preserves the convergence rate beyond the fixed $m$ regime (black).}
\label{fig:moment_convergence}
\end{figure} 

The new variance term $V_2$ consists of two parts: the variance from the non-Neyman orthogonal nuisance function $G$, and the covariance between this non-Neyman orthogonal term and the main empirical process term $(\bbE_n-\bbE)[n/\mu_n \partial \rho(\blambda\tp\psi(\beta_0, G_0, \boldeta_0;\O))/ \partial \beta]$. It must be taken into account in our framework where the conditional censoring distribution is a non-Neyman orthogonal nuisance function. When the censoring rate is non-zero, this extra variability cannot be ignored and is critical for ensuring accurate and valid inference. When censoring is absent, $V_2 = 0$, and our result reduces to the uncensored case, as shown in Corollary~\ref{coro1} in the Supplementary Material. 

In practice, $H$ and $V_1$ can be consistently estimated by
$$
\hat H = \frac{\partial^2\hat Q(\hat\beta_{\textrm{GEL2}}^{\textrm{C}},\hat G, \hat\boldeta)}{\partial\beta^2 },\quad \hat V_1= \hat H^{-1} \hat D\tp\bar\Sigma(\hat\beta_{\textrm{GEL2}}^{\textrm{C}},\hat G,\hat\boldeta)^{-1}\hat D\hat H^{-1}, 
$$
respectively, where
$$
\hat D = \sum_{i=1}^n \frac{\rho^{'}(\hat\blambda(\hat\beta_{\textrm{GEL2}}^{\textrm{C}},\hat G,\hat\boldeta)\tp \psi^{'}(\hat\beta_{\textrm{GEL2}}^{\textrm{C}},\hat G,\hat\boldeta;\O_i))\psi^{'}(\hat\beta_{\textrm{GEL2}}^{\textrm{C}},\hat G,\hat\boldeta;\O_i)}{\sum_{j=1}^n \rho^{'}(\hat\blambda(\hat\beta_{\textrm{GEL2}}^{\textrm{C}},\hat G,\hat\boldeta)\tp \psi^{'}(\hat\beta_{\textrm{GEL2}}^{\textrm{C}},\hat G,\hat\boldeta;\O_j))}. 
$$
In addition, $V_2$ can be consistently estimated by
\bse
&&\hat H^{-2} \frac{1}{n}\sum_{j=1}^{n}\left(\frac{1}{n}\sum_{i=1}^{n}D_G\frac{\partial \rho(\blambda(\hat\beta_{\textrm{GEL2}}^{\textrm{C}},\hat G,\hat\boldeta)\tp\psi(\hat\beta_{\textrm{GEL2}}^{\textrm{C}},\hat G,\hat\boldeta;\O_i))}{\partial \beta}\phi(\O_j,\O_i)\right)^2\n\\
&&+2\hat H^{-2}  \frac{1}{n}\sum_{j=1}^{n}\left(\frac{1}{n}\sum_{i=1}^{n}\{D_G\frac{\partial\rho(\blambda\tp\psi(\hat\beta,\hat G,\hat \boldeta;\O_{i}))}{\partial\beta}\phi(\O_{j},\O_{i})\bar\psi^{'}\bar\Sigma^{-1}\psi(\hat\beta_{\textrm{GEL2}}^{\textrm{C}},\hat G,{\hat\boldeta};\O_{i})\}\right). 
\ese

\section{Proof of Theorem \ref{thmtest}}
Following the proof of Theorem 4 in \cite{newey2009} and Theorem 3 in \cite{ye2024}, we have
\be\label{2nqhat}
\frac{2n\hat Q(\beta_0,G_0,\boldeta_0)-(m-1)}{\sqrt{2(m-1)}}\to N(0,1),
\ee
and 
$$
\frac{2n\hat Q(\hat\beta,\hat G,\hat \boldeta)-2n\hat Q(\beta_0,\hat G,\hat \boldeta)}{\sqrt{2(m-1)}}=o_p(1).
$$
Then we only have to calculate $2n\hat Q(\beta_0, \hat G ,\hat\boldeta)-2n\hat Q(\beta_0,G_0,\boldeta_0)$. 
By Taylor's expansion, $2n\hat Q(\hat\beta, \hat G ,\hat\boldeta)-2n\hat Q(\hat\beta,G_0,\boldeta_0)$ becomes
\be\label{test-Gn}
&&\frac{2n\hat Q(\beta_0, \hat G ,\hat\boldeta)-2n\hat Q(\beta_0,G_0,\boldeta_0)}{\sqrt{2(m-1)}}\n\\
&=&D_G\frac{2n}{\sqrt{2(m-1)}}{\hat Q(\beta_0,G_0,\boldeta_0)}(\hat G -G_0)+D_{\boldeta}\frac{2n}{\sqrt{2(m-1)}}{\hat Q(\beta_0,G_0,\boldeta_0)}(\hat\boldeta -\boldeta_0) +o_p(1).  \n
\ee
For the first term, we have 
\be
&&D_G\frac{2n}{\sqrt{2(m-1)}}{\hat Q(\hat\beta,G_0,\boldeta_0)}(\hat G -G_0)\n\\
&=&(\bbE_n-\bbE_{P_{\O_T}}+\bbE_{P_{\O_T}})D_G\frac{2n}{\sqrt{2(m-1)}}\rho(\blambda\tp\psi(\beta_0,G_0,\boldeta_0;\O_T))(\hat G -G_0)\n\\
&=&\bbE_{P_{\O_T}} D_G\frac{2n}{\sqrt{2(m-1)}}\rho(\blambda\tp\psi(\beta_0,G_0,\boldeta_0;\O_T))(\frac{1}{n}\sum_{i=1}^n\phi(\O^{'}_{T,i},\O_T)+o_p(n^{-1/2}))\n\\
&=&\frac{1}{n}\sum_{i=1}^n\bbE_{P_{\O_T}} D_G\frac{2n}{\sqrt{2(m-1)}}\rho(\blambda\tp\psi(\beta_0,G_0,\boldeta_0;\O_T))\phi(\O^{'}_{T,i},\O_T)+o_p(1)\n\\
&\to &N(0,V_3)\n,
\ee
where $V_3 = \bbE_{P_{\O_T}^{'}}(\bbE_{P_{\O_T}}D_G 2n/\sqrt{2(m-1)} \rho(\blambda\tp\psi(\beta_0,G_0,\boldeta_0;\O_T))\phi(\O_{T}^{'},\O_{T}))^2$. 

For the second term, since $\psi$ is Neyman orthogonal to $\boldeta$, this term is $o_p(1)$. %For the remainder term $R$, it contains several second order term. We provide a detailed proof for the order of one representative term below. The remaining terms can be analyzed using similar arguments and are omitted. The remainder term with regard to $G$ is
Then according to Lemma A10 in \cite{newey2009}, we have
\be
&&\frac{{2n\hat Q(\hat\beta,\hat G ,\hat\boldeta)-(m-1)}}{\sqrt{2(m-1)}}\n\\
&=&\frac{2n\hat Q(\hat\beta,\hat G ,\hat\boldeta)-2n\hat Q(\beta_0,\hat G ,\hat\boldeta)}{\sqrt{2(m-1)}}+\frac{2n\hat Q(\beta_0,\hat G ,\hat\boldeta)-2n\hat Q(\beta_0,G_0 ,\boldeta_0)}{\sqrt{2(m-1)}}\n\\
&&+\frac{2n\hat Q(\beta_0,G_0 ,\boldeta_0)-(m-1)}{\sqrt{2(m-1)}} \n\\
&\to& N(0,1+V_3).
\ee
This completes the proof of Theorem \ref{thmtest}.
\section{Higher-order Kernel}\label{suppsec:kernel}

Higher-order kernel functions are designed to reduce the bias in nonparametric estimation; see \citet{parzen1962, fan1992}. A kernel function $K$ of order $r$ satisfies the moment conditions
\[
\int u^i K(u)\, du = 0 \quad \text{for } i = 1, \dots, r-1, \quad \text{and} \quad \int K(u)\, du = 1.
\]
These assumptions ensure that the leading bias term vanishes up to order $h^r$ in kernel-based estimators. Specifically, for the kernel density estimator $\hat{f}(x)$ of a smooth function $f(x)$, the bias admits the expansion
\[
\mathbb{E}[\hat{f}(x)] - f(x) = \sum_{i=1}^r\frac{h^i}{i!} f^{(i)}(x) \int u^i K(u) \, du + o(h^r),
\]
where $f^{(r)}$ denotes the $r$th derivative of $f$. If $K$ is of order $r$, all lower-order terms vanish, and the leading bias becomes $O(h^r)$. Since kernel functions are typically chosen to be even, thus all odd-order moments vanish automatically, only even-order moments need to be constructed explicitly.

Following \citet{fan1992}, kernels of arbitrary order can be constructed recursively. In particular, starting from the standard Cauchy density $K(x) = {2}/{\pi(1 + x^2)^2}$, higher-order kernels $K^{(2r-1)}(x)$ satisfy the recursion
\[
(1 + x^2)K^{(2r-1)}(x) + 2(2r - 1)x K^{(2r-2)}(x) + (2r - 1)(2r - 2)K^{(2r-3)}(x) = 0.
\]
Closed-form expressions for $K^{(2r-1)}(x)$ up to order $8$ are given in Table~\ref{tab:highorderkernel}.

\begin{table}
\center
\caption{\label{tab:highorderkernel}Higher-order kernels derived from the standard Cauchy density}
\begin{tabular}{c c}
\hline
Kernel order & Kernel function \\ 
\hline
2 & $\displaystyle \frac{2}{\pi(1 + x^2)^2}$ \\
4 & $\displaystyle \frac{16(x^2 - 1)}{\pi(1 + x^2)^4}$ \\
6 & $\displaystyle \frac{90.54(3x^4 - 10x^2 + 3)}{(1 + x^2)^6}$ \\
8 & $\displaystyle \frac{130.38(x^6 - 7x^4 + 7x^2 - 1)}{(1 + x^2)^8}$ \\
\hline
\end{tabular}
\end{table}

Because higher-order kernels often involve high-degree polynomials, the resulting kernel functions may take negative values over parts of their domain. To ensure nonnegativity of the estimated density, we apply a positive truncation to the estimator. Specifically, we define the modified density estimator $\tilde{f}(x)$ as
\[
\tilde{f}(x) = \frac{ \hat{f}(x) \cdot I\{\hat{f}(x)>0\}}{\int \hat{f}(u) \cdot I\{\hat{f}(u)>0\} \, du},
\]
where $\hat{f}(x)$ is the original, possibly nonpositive, estimator. This transformation guarantees that $\tilde{f}(x)$ remains strictly positive and integrates to one. Importantly, this truncation and renormalization step preserves the asymptotic properties of the estimator; see, for example, \citet{hall1993}.

\section{Additional simulation results}
\subsection{Right-Censored data simulation}\label{suppsec:censorsimu}

In the main text, we present simulation results under nonlinear nuisance functions. For completeness, we also conducted simulations with linear nuisance functions. The only difference between the two settings lies in the data-generating process of the nuisance functions $\alpha$ and $\gamma$; all other components of the simulation design remain identical. Specifically, in the linear nuisance setting, $Z_j\sim \textrm{Unif}[-2,2]$ for $j=1,\dots,m$, $\alpha_1(\Z)=\sum_{j=1}^m \xi_{A,j}Z_j$, $\gamma_1(\Z)=\sum_{j=1}^m\xi_{Y,j}Z_j$, $\alpha_2(\X)=0$, and $\gamma_2(\X)=\sum_{j=1}^5 X_j$, where $\xi_{A,j}\sim N(0,0.4(1-h^2))$.

We use the ReLU activation function and the Adam optimizer to train the DNN, with a maximum of 1000 epochs. The remaining hyperparameters are selected via grid search, with the specific search ranges summarized in Table \ref{grid-search}. To prevent overfitting and reduce unnecessary computation, we adopt an early stopping strategy: 5\% of the training data is used as a validation set, and training is terminated early if the MSE on validation set does not improve for five consecutive epochs. The model corresponding to the lowest validation MSE during training is selected as the final model. In addition, the hyperparameters of RF and XGB are also chosen by cross validation.

\begin{table}
\center
\caption{\label{grid-search}Grid search ranges for DNN hyperparameters.}
\begin{tabular}{l l}
\hline
Hyperparameter & Search Range \\
\hline
Depth          & \{2, 3, 4, 5\} \\
Width & \{20, 50, 100, 200\} \\
Learning Rate  & \{0.001, 0.0005, 0.0001\} \\
Dropout Rate   & \{0.0, 0.1, 0.2, 0.3\} \\
Batch size & \{64, 128, 256, 512\}\\
\hline
\end{tabular}
\end{table}

Table \ref{censor-simu-linear} shows the under the linear nuisance function, all methods except AFT achieve low bias and coverage probabilities close to the normal level. This suggests that when the model structure aligns with the data generating process, both parametric and nonparametric methods can perform adequately. The bias of the AFT model in the presence of confounding demonstrates the need for specialized methods to perform valid causal inference.

\begin{table}
\center
\caption{\label{censor-simu-linear}Simulation results with censored outcomes under linear nuisance functions. The results are based on 200 repeated experiments with $n=10,000$ and $m=20$. BIAS means the percentage bias between $\hat\beta$ and $\beta_0$, SD is the sample standard deviation of $\hat\beta$ over 200 replications, SE is the averaged standard error, and CP is the 95\% coverage probability.}
\small{\begin{tabular}{lrrrrlrrrr}
\hline
\multicolumn{5}{c}{CASE 1 (No invalid IVs)}                                                                      & \multicolumn{5}{c}{CASE 2 (40\% invalid   IVs)}                                                                  \\ \hline
Model    & \multicolumn{1}{c}{BIAS} & \multicolumn{1}{c}{SD} & \multicolumn{1}{c}{SE} & \multicolumn{1}{c}{CP} & Model    & \multicolumn{1}{c}{BIAS} & \multicolumn{1}{c}{SD} & \multicolumn{1}{c}{SE} & \multicolumn{1}{c}{CP} \\ \hline
DNN\_ET  & 0.865\%                  & 0.115                   & 0.157                   & 0.950                  & DNN\_ET  & 0.911\%                  & 0.108                   & 0.158                   & 0.955                  \\
DNN\_EL  & 0.663\%                  & 0.115                   & 0.171                   & 0.955                  & DNN\_EL  & 0.361\%                  & 0.108                   & 0.175                   & 0.955                  \\
DNN\_CUE & 1.406\%                  & 0.115                   & 0.165                   & 0.970                  & DNN\_CUE & 1.526\%                  & 0.108                   & 0.167                   & 0.970                  \\
LR\_ET   & 0.616\%                  & 0.107                   & 0.154                   & 0.965                  & LR\_ET   & 0.034\%                  & 0.111                   & 0.159                   & 0.955                  \\
LR\_EL   & 0.416\%                  & 0.107                   & 0.169                   & 0.975                  & LR\_EL   & -0.310\%                 & 0.111                   & 0.178                   & 0.955                  \\
LR\_CUE  & 1.108\%                  & 0.107                   & 0.162                   & 0.975                  & LR\_CUE  & 0.596\%                  & 0.111                   & 0.168                   & 0.970                  \\
RF\_ET   & -1.472\%                 & 0.154                   & 0.234                   & 0.955                  & RF\_ET   & 6.307\%                  & 0.143                   & 0.232                   & 0.995                  \\
RF\_EL   & -5.655\%                 & 0.155                   & 0.435                   & 0.980                  & RF\_EL   & 2.947\%                  & 0.144                   & 0.370                   & 1.000                  \\
RF\_CUE  & 2.633\%                  & 0.150                   & 0.247                   & 0.965                  & RF\_CUE  & 6.248\%                  & 0.136                   & 0.249                   & 0.995                  \\
XGB\_ET  & -5.137\%                 & 0.110                   & 0.169                   & 0.965                  & XGB\_ET  & -4.847\%                 & 0.116                   & 0.169                   & 0.970                  \\
XGB\_EL  & -5.549\%                 & 0.110                   & 0.185                   & 0.975                  & XGB\_EL  & -5.459\%                 & 0.116                   & 0.195                   & 0.990                  \\
XGB\_CUE & -4.504\%                 & 0.111                   & 0.178                   & 0.965                  & XGB\_CUE & -4.175\%                 & 0.116                   & 0.177                   & 0.985                  \\
AFT      & -123.668\%               & 0.011                   & 0.008                   & 0.000                  & AFT      & -106.013\%               & 0.018                   & 0.010                   & 0.000                  \\ \hline
\multicolumn{5}{c}{CASE 3 (90\%   invalid IVs with InSIDE)}                                                      & \multicolumn{5}{c}{CASE 4 (90\% invalid IVs without InSIDE)}                                                     \\ \hline
Model    & \multicolumn{1}{c}{BIAS} & \multicolumn{1}{c}{SD} & \multicolumn{1}{c}{SE} & \multicolumn{1}{c}{CP} & Model    & \multicolumn{1}{c}{BIAS} & \multicolumn{1}{c}{SD} & \multicolumn{1}{c}{SE} & \multicolumn{1}{c}{CP} \\ \hline
DNN\_ET  & 0.611\%                  & 0.114                   & 0.161                   & 0.950                  & DNN\_ET  & 0.072\%                  & 0.119                   & 0.173                   & 0.930                  \\
DNN\_EL  & 0.186\%                  & 0.115                   & 0.182                   & 0.965                  & DNN\_EL  & -0.241\%                 & 0.121                   & 0.196                   & 0.930                  \\
DNN\_CUE & 1.160\%                  & 0.115                   & 0.169                   & 0.970                  & DNN\_CUE & 0.611\%                  & 0.119                   & 0.181                   & 0.940                  \\
LR\_ET   & 0.310\%                  & 0.110                   & 0.157                   & 0.950                  & LR\_ET   & -1.509\%                 & 0.123                   & 0.165                   & 0.935                  \\
LR\_EL   & -0.056\%                 & 0.110                   & 0.179                   & 0.965                  & LR\_EL   & -1.752\%                 & 0.124                   & 0.188                   & 0.945                  \\
LR\_CUE  & 0.832\%                  & 0.110                   & 0.166                   & 0.960                  & LR\_CUE  & -0.994\%                 & 0.124                   & 0.174                   & 0.950                  \\
RF\_ET   & 8.462\%                  & 0.149                   & 0.233                   & 0.980                  & RF\_ET   & 3.615\%                  & 0.147                   & 0.234                   & 0.960                  \\
RF\_EL   & 5.296\%                  & 0.151                   & 0.655                   & 1.000                  & RF\_EL   & -0.419\%                 & 0.148                   & 0.473                   & 0.990                  \\
RF\_CUE  & 8.099\%                  & 0.145                   & 0.249                   & 0.985                  & RF\_CUE  & 7.616\%                  & 0.145                   & 0.249                   & 0.970                  \\
XGB\_ET  & -5.306\%                 & 0.111                   & 0.172                   & 0.950                  & XGB\_ET  & -11.363\%                & 0.112                   & 0.166                   & 0.925                  \\
XGB\_EL  & -5.459\%                 & 0.112                   & 0.214                   & 0.970                  & XGB\_EL  & -11.695\%                & 0.113                   & 0.219                   & 0.960                  \\
XGB\_CUE & -4.639\%                 & 0.112                   & 0.181                   & 0.960                  & XGB\_CUE & -10.784\%                & 0.112                   & 0.174                   & 0.945                  \\
AFT      & -114.957\%               & 0.024                   & 0.009                   & 0.000                  & AFT      & -104.292\%               & 0.015                   & 0.009                   & 0.000                  \\ \hline
\end{tabular}}
\end{table}

\subsection{Uncensored data simulation}\label{suppsec:uncensorsimu}
We also conduct some simulations with uncensored outcomes. The data-generating process is the same as in the censored case, except that we do not generate censoring time $C$ and observed outcomes $Y$. Instead, we treat the true event time $T$ as fully observed.

Tables \ref{uncensor-linear} and \ref{uncensor-nonlinear} present results for different methods when sample size is $n=10,000$ and $n=100,000$. In Table \ref{uncensor-linear}, under the linear nuisance functions, all methods achieve low bias and satisfactory coverage probability. This is expected, as the model specification aligns with the DGP. The weak identification test statistic $F_{\textrm{MAWII}} = 5.58$ for $n = 10,000$ and 51.26 for $n = 100,000$, both  above the threshold of 2. This confirms that the model is identifiable via heteroscedasticity, with stronger evidence as the sample size increases. In contrast, Table \ref{uncensor-nonlinear} highlights the performance differences under nonlinear nuisance function. Traditional machine learning methods such as LR, RF, and XGB exhibit substantial bias and small coverage probabilities across all cases. These methods fail to account for the nonlinear relationships embedded in the data, which results in inaccurate estimation of the treatment effect. By comparison, DNN maintains stable and accurate performance across both linear and nonlinear settings. They consistently yield low bias, small standard deviation, and coverage probabilities close to the nominal level, indicating their ability to flexibly approximate complex functional forms. 

In the comparison of different GEL functions, the results with a larger sample size show that EL, ET, and CUE yield nearly identical performance, indicating their asymptotic equivalence. However, in the smaller sample setting, CUE tends to produce larger bias compared to EL and ET. This suggests that while all three methods perform similarly when the sample size is sufficiently large, CUE is relatively suboptimal in finite samples. EL and ET demonstrate greater stability and efficiency under limited data, making them more reliable choices in practice.

\begin{table}
\center
\caption{\label{uncensor-linear}Simulation results of uncensored data under linear nuisance functions. The results are based on 200 repeated experiments with $n=10,000$ and $m=20$. BIAS means the percentage bias between $\hat\beta$ and $\beta_0$, SD is the sample standard deviation of $\hat\beta$ over 200 replications, SE is the averaged standard error, and CP is the 95\% coverage probability.}
\footnotesize{\begin{tabular}{clrrrrrrrr}
\hline
&          & \multicolumn{4}{c}{$n=10,000$, $F_{\textrm{MAWII}}=5.58$}                                                                         & \multicolumn{4}{c}{$n=100,000$, $F_{\textrm{MAWII}}=51.26$}                                                                        \\ \cline{2-10} 
& Model    & \multicolumn{1}{c}{BIAS} & \multicolumn{1}{c}{SD} & \multicolumn{1}{c}{SE} & \multicolumn{1}{c}{CP} & \multicolumn{1}{c}{BIAS} & \multicolumn{1}{c}{SD} & \multicolumn{1}{c}{SE} & \multicolumn{1}{c}{CP} \\ \hline
\multirow{12}{*}{\begin{tabular}[c]{@{}c@{}}CASE 1\\ (No invalid IVs)\end{tabular}}                     & DNN\_ET  & 2.086\%                  & 0.071                   & 0.070                   & 0.955                  & 0.502\%                  & 0.021                   & 0.021                   & 0.955                  \\
& DNN\_EL  & 1.599\%                  & 0.072                   & 0.079                   & 0.950                  & 0.444\%                  & 0.021                   & 0.021                   & 0.955                  \\
& DNN\_CUE & 2.574\%                  & 0.071                   & 0.072                   & 0.965                  & 0.561\%                  & 0.021                   & 0.021                   & 0.955                  \\
& LR\_ET   & 1.927\%                  & 0.070                   & 0.069                   & 0.945                  & 0.355\%                  & 0.021                   & 0.021                   & 0.955                  \\
& LR\_EL   & 1.404\%                  & 0.070                   & 0.073                   & 0.945                  & 0.296\%                  & 0.021                   & 0.021                   & 0.955                  \\
& LR\_CUE  & 2.421\%                  & 0.071                   & 0.070                   & 0.960                  & 0.413\%                  & 0.021                   & 0.021                   & 0.955                  \\
& RF\_ET   & -1.219\%                 & 0.103                   & 0.099                   & 0.925                  & -3.275\%                 & 0.035                   & 0.033                   & 0.915                  \\
& RF\_EL   & -1.776\%                 & 0.104                   & 0.106                   & 0.930                  & -3.159\%                 & 0.035                   & 0.034                   & 0.920                  \\
& RF\_CUE  & -0.474\%                 & 0.103                   & 0.103                   & 0.940                  & -3.410\%                 & 0.035                   & 0.033                   & 0.920                  \\
& XGB\_ET  & -3.958\%                 & 0.088                   & 0.085                   & 0.945                  & -2.185\%                 & 0.023                   & 0.023                   & 0.925                  \\
& XGB\_EL  & -4.544\%                 & 0.088                   & 0.086                   & 0.940                  & -2.238\%                 & 0.023                   & 0.023                   & 0.925                  \\
& XGB\_CUE & -3.402\%                 & 0.088                   & 0.087                   & 0.955                  & -2.130\%                 & 0.023                   & 0.023                   & 0.930                  \\ \hline
\multirow{12}{*}{\begin{tabular}[c]{@{}c@{}}CASE 2\\ (40\% invalid IVs)\end{tabular}}                   & DNN\_ET  & 2.069\%                  & 0.072                   & 0.070                   & 0.945                  & 0.500\%                  & 0.021                   & 0.021                   & 0.955                  \\
& DNN\_EL  & 1.614\%                  & 0.072                   & 0.083                   & 0.940                  & 0.441\%                  & 0.021                   & 0.021                   & 0.955                  \\
& DNN\_CUE & 2.556\%                  & 0.072                   & 0.072                   & 0.955                  & 0.558\%                  & 0.021                   & 0.021                   & 0.955                  \\
& LR\_ET   & 1.926\%                  & 0.070                   & 0.069                   & 0.945                  & 0.355\%                  & 0.021                   & 0.021                   & 0.955                  \\
& LR\_EL   & 1.404\%                  & 0.070                   & 0.073                   & 0.945                  & 0.296\%                  & 0.021                   & 0.021                   & 0.955                  \\
& LR\_CUE  & 2.421\%                  & 0.071                   & 0.070                   & 0.960                  & 0.413\%                  & 0.021                   & 0.021                   & 0.955                  \\
& RF\_ET   & -0.769\%                 & 0.105                   & 0.101                   & 0.950                  & -3.478\%                 & 0.037                   & 0.034                   & 0.910                  \\
& RF\_EL   & -1.400\%                 & 0.107                   & 0.106                   & 0.950                  & -3.364\%                 & 0.037                   & 0.034                   & 0.915                  \\
& RF\_CUE  & 0.008\%                  & 0.104                   & 0.104                   & 0.960                  & -3.613\%                 & 0.037                   & 0.034                   & 0.910                  \\
& XGB\_ET  & -5.695\%                 & 0.079                   & 0.084                   & 0.945                  & -2.626\%                 & 0.021                   & 0.023                   & 0.935                  \\
& XGB\_EL  & -6.289\%                 & 0.080                   & 0.086                   & 0.930                  & -2.682\%                 & 0.021                   & 0.023                   & 0.935                  \\
& XGB\_CUE & -5.125\%                 & 0.079                   & 0.086                   & 0.945                  & -2.568\%                 & 0.021                   & 0.023                   & 0.935                  \\ \hline
\multirow{12}{*}{\begin{tabular}[c]{@{}c@{}}CASE 3\\ (90\% invalid IVs\\  with InSIDE  )\end{tabular}}  & DNN\_ET  & 2.140\%                  & 0.072                   & 0.070                   & 0.945                  & 0.500\%                  & 0.021                   & 0.021                   & 0.955                  \\
& DNN\_EL  & 1.623\%                  & 0.072                   & 0.076                   & 0.950                  & 0.441\%                  & 0.021                   & 0.021                   & 0.955                  \\
& DNN\_CUE & 2.628\%                  & 0.072                   & 0.072                   & 0.960                  & 0.558\%                  & 0.021                   & 0.021                   & 0.955                  \\
& LR\_ET   & 1.926\%                  & 0.070                   & 0.069                   & 0.945                  & 0.355\%                  & 0.021                   & 0.021                   & 0.955                  \\
& LR\_EL   & 1.404\%                  & 0.070                   & 0.073                   & 0.945                  & 0.296\%                  & 0.021                   & 0.021                   & 0.955                  \\
& LR\_CUE  & 2.421\%                  & 0.071                   & 0.070                   & 0.960                  & 0.413\%                  & 0.021                   & 0.021                   & 0.955                  \\
& RF\_ET   & -2.134\%                 & 0.114                   & 0.104                   & 0.925                  & -3.210\%                 & 0.037                   & 0.035                   & 0.930                  \\
& RF\_EL   & -2.566\%                 & 0.115                   & 0.111                   & 0.920                  & -3.085\%                 & 0.037                   & 0.035                   & 0.930                  \\
& RF\_CUE  & -1.394\%                 & 0.114                   & 0.108                   & 0.935                  & -3.348\%                 & 0.037                   & 0.035                   & 0.935                  \\
& XGB\_ET  & -6.294\%                 & 0.082                   & 0.085                   & 0.955                  & -2.670\%                 & 0.022                   & 0.023                   & 0.930                  \\
& XGB\_EL  & -6.808\%                 & 0.083                   & 0.088                   & 0.955                  & -2.726\%                 & 0.022                   & 0.023                   & 0.925                  \\
& XGB\_CUE & -5.738\%                 & 0.083                   & 0.087                   & 0.965                  & -2.612\%                 & 0.022                   & 0.023                   & 0.935                  \\ \hline
\multirow{12}{*}{\begin{tabular}[c]{@{}c@{}}CASE 4\\ (90\% invalid IVs\\  without InSIDE)\end{tabular}} & DNN\_ET  & 2.071\%                  & 0.071                   & 0.070                   & 0.955                  & 0.489\%                  & 0.021                   & 0.021                   & 0.955                  \\
& DNN\_EL  & 1.590\%                  & 0.072                   & 0.080                   & 0.950                  & 0.431\%                  & 0.021                   & 0.021                   & 0.955                  \\
& DNN\_CUE & 2.559\%                  & 0.072                   & 0.072                   & 0.960                  & 0.548\%                  & 0.021                   & 0.021                   & 0.955                  \\
& LR\_ET   & 1.927\%                  & 0.070                   & 0.069                   & 0.945                  & 0.355\%                  & 0.021                   & 0.021                   & 0.955                  \\
& LR\_EL   & 1.404\%                  & 0.070                   & 0.073                   & 0.945                  & 0.296\%                  & 0.021                   & 0.021                   & 0.955                  \\
& LR\_CUE  & 2.421\%                  & 0.071                   & 0.070                   & 0.960                  & 0.413\%                  & 0.021                   & 0.021                   & 0.955                  \\
& RF\_ET   & -2.874\%                 & 0.105                   & 0.101                   & 0.935                  & -3.796\%                 & 0.036                   & 0.034                   & 0.920                  \\
& RF\_EL   & -3.478\%                 & 0.105                   & 0.116                   & 0.940                  & -3.641\%                 & 0.036                   & 0.035                   & 0.920                  \\
& RF\_CUE  & -2.155\%                 & 0.105                   & 0.104                   & 0.945                  & -3.940\%                 & 0.036                   & 0.035                   & 0.920                  \\
& XGB\_ET  & -5.224\%                 & 0.080                   & 0.085                   & 0.935                  & -2.772\%                 & 0.021                   & 0.023                   & 0.945                  \\
& XGB\_EL  & -5.684\%                 & 0.080                   & 0.087                   & 0.940                  & -2.828\%                 & 0.021                   & 0.023                   & 0.945                  \\
& XGB\_CUE & -4.663\%                 & 0.080                   & 0.087                   & 0.955                  & -2.715\%                 & 0.021                   & 0.023                   & 0.945                  \\ \hline
\end{tabular}}
\end{table}

\begin{table}
\center
\caption{\label{uncensor-nonlinear}Simulation results of uncensored data under nonlinear nuisance functions. The results are based on 200 repeated experiments with $n=10,000$ and $m=20$. BIAS means the percentage bias between $\hat\beta$ and $\beta_0$, SD is the sample standard deviation of $\hat\beta$ over 200 replications, SE is the averaged standard error, and CP is the 95\% coverage probability.}
\footnotesize{\begin{tabular}{clrrrrrrrr}
\hline
&          & \multicolumn{4}{c}{$n=10,000$, $F_{\textrm{MAWII}}=2.75$}      & \multicolumn{4}{c}{$n=100,000$, $F_{\textrm{MAWII}}=24.92$}    \\ \cline{2-10} 
& Model    & BIAS      & SD   & SE    & CP    & BIAS      & SD   & SE   & CP    \\ \hline
\multirow{12}{*}{\begin{tabular}[c]{@{}c@{}}CASE 1\\ (No invalid IVs)\end{tabular}}                     & DNN\_ET  & 1.117\%   & 0.124 & 0.112  & 0.950 & -0.092\%  & 0.026 & 0.024 & 0.935 \\
& DNN\_EL  & 1.150\%   & 0.127 & 0.176  & 0.945 & -0.154\%  & 0.026 & 0.024 & 0.935 \\
& DNN\_CUE & 1.753\%   & 0.124 & 0.116  & 0.955 & -0.030\%  & 0.026 & 0.024 & 0.940 \\
& LR\_ET   & 7.466\%   & 0.540 & 0.589  & 0.730 & 0.341\%   & 0.881 & 0.288 & 0.070 \\
& LR\_EL   & -12.548\% & 0.038 & 0.096  & 0.665 & -12.304\% & 0.013 & 0.039 & 0.260 \\
& LR\_CUE  & -12.171\% & 0.035 & 0.037  & 0.730 & -11.980\% & 0.011 & 0.011 & 0.035 \\
& RF\_ET   & -8.065\%  & 1.036 & 1.881  & 0.855 & 7.619\%   & 0.073 & 0.070 & 0.975 \\
& RF\_EL   & -40.111\% & 1.021 & 2.463  & 0.845 & 7.867\%   & 0.075 & 0.073 & 0.955 \\
& RF\_CUE  & -16.776\% & 1.021 & 1.218  & 0.875 & 7.739\%   & 0.071 & 0.071 & 0.980 \\
& XGB\_ET  & -22.441\% & 0.359 & 0.309  & 0.740 & -10.220\% & 0.049 & 0.038 & 0.690 \\
& XGB\_EL  & -20.307\% & 0.336 & 0.295  & 0.745 & -10.214\% & 0.049 & 0.038 & 0.685 \\
& XGB\_CUE & -23.486\% & 0.372 & 0.317  & 0.745 & -10.205\% & 0.049 & 0.038 & 0.690 \\ \hline
\multirow{12}{*}{\begin{tabular}[c]{@{}c@{}}CASE 2\\ (40\% invalid IVs)\end{tabular}}                   & DNN\_ET  & 2.370\%   & 0.124 & 0.113  & 0.940 & -0.161\%  & 0.026 & 0.024 & 0.945 \\
& DNN\_EL  & 2.143\%   & 0.126 & 0.114  & 0.925 & -0.224\%  & 0.026 & 0.024 & 0.950 \\
& DNN\_CUE & 3.012\%   & 0.124 & 0.117  & 0.960 & -0.099\%  & 0.026 & 0.024 & 0.945 \\
& LR\_ET   & 1.766\%   & 0.488 & 0.490  & 0.695 & -0.446\%  & 0.872 & 0.260 & 0.140 \\
& LR\_EL   & -12.831\% & 0.040 & 0.071  & 0.640 & -12.784\% & 0.020 & 0.032 & 0.240 \\
& LR\_CUE  & -12.346\% & 0.038 & 0.037  & 0.700 & -12.418\% & 0.019 & 0.011 & 0.090 \\
& RF\_ET   & 8.665\%   & 0.927 & 1.282  & 0.865 & 10.611\%  & 0.079 & 0.072 & 0.930 \\
& RF\_EL   & -25.853\% & 1.088 & 2.659  & 0.825 & 10.716\%  & 0.083 & 0.081 & 0.925 \\
& RF\_CUE  & 25.758\%  & 0.957 & 1.324  & 0.870 & 10.771\%  & 0.079 & 0.073 & 0.930 \\
& XGB\_ET  & -11.994\% & 0.381 & 0.337  & 0.755 & -9.917\%  & 0.059 & 0.039 & 0.710 \\
& XGB\_EL  & -11.015\% & 0.440 & 0.426  & 0.755 & -9.900\%  & 0.059 & 0.039 & 0.710 \\
& XGB\_CUE & -14.031\% & 0.341 & 0.281  & 0.765 & -9.895\%  & 0.059 & 0.039 & 0.715 \\ \hline
\multirow{12}{*}{\begin{tabular}[c]{@{}c@{}}CASE 3\\ (90\% invalid IVs\\  with InSIDE  )\end{tabular}}  & DNN\_ET  & 3.053\%   & 0.120 & 0.116  & 0.955 & -0.276\%  & 0.026 & 0.025 & 0.950 \\
& DNN\_EL  & 2.504\%   & 0.122 & 0.123  & 0.945 & -0.341\%  & 0.026 & 0.025 & 0.945 \\
& DNN\_CUE & 3.801\%   & 0.120 & 0.120  & 0.960 & -0.212\%  & 0.026 & 0.025 & 0.950 \\
& LR\_ET   & 6.532\%   & 0.505 & 0.433  & 0.755 & 1.370\%   & 0.865 & 0.209 & 0.135 \\
& LR\_EL   & -10.488\% & 0.039 & 0.088  & 0.705 & -10.541\% & 0.015 & 0.033 & 0.320 \\
& LR\_CUE  & -10.001\% & 0.036 & 0.037  & 0.760 & -9.829\%  & 0.014 & 0.011 & 0.110 \\
& RF\_ET   & -40.107\% & 1.253 & 4.596  & 0.880 & 11.193\%  & 0.081 & 0.075 & 0.965 \\
& RF\_EL   & -51.942\% & 1.334 & 11.817 & 0.865 & 11.233\%  & 0.084 & 0.090 & 0.950 \\
& RF\_CUE  & -9.520\%  & 1.153 & 2.700  & 0.920 & 11.474\%  & 0.080 & 0.076 & 0.970 \\
& XGB\_ET  & -25.860\% & 0.494 & 0.335  & 0.715 & -13.013\% & 0.065 & 0.042 & 0.625 \\
& XGB\_EL  & -25.030\% & 0.510 & 0.463  & 0.715 & -12.955\% & 0.065 & 0.043 & 0.620 \\
& XGB\_CUE & -25.208\% & 0.513 & 0.331  & 0.720 & -12.991\% & 0.065 & 0.042 & 0.630 \\ \hline
\multirow{12}{*}{\begin{tabular}[c]{@{}c@{}}CASE 4\\ (90\% invalid IVs\\  without InSIDE)\end{tabular}} & DNN\_ET  & 0.508\%   & 0.119 & 0.111  & 0.940 & -0.256\%  & 0.026 & 0.024 & 0.945 \\
& DNN\_EL  & 0.109\%   & 0.120 & 0.113  & 0.930 & -0.319\%  & 0.026 & 0.024 & 0.945 \\
& DNN\_CUE & 1.143\%   & 0.119 & 0.114  & 0.955 & -0.195\%  & 0.026 & 0.024 & 0.945 \\
& LR\_ET   & 4.292\%   & 0.568 & 0.623  & 0.565 & 2.369\%   & 0.884 & 0.278 & 0.045 \\
& LR\_EL   & -16.671\% & 0.038 & 0.073  & 0.530 & -16.552\% & 0.014 & 0.045 & 0.165 \\
& LR\_CUE  & -16.533\% & 0.035 & 0.037  & 0.570 & -16.363\% & 0.012 & 0.011 & 0.000 \\
& RF\_ET   & -25.852\% & 1.009 & 1.243  & 0.920 & 13.543\%  & 0.081 & 0.075 & 0.950 \\
& RF\_EL   & 2.823\%   & 1.162 & 3.692  & 0.880 & 13.761\%  & 0.087 & 0.083 & 0.950 \\
& RF\_CUE  & -10.273\% & 0.987 & 1.384  & 0.920 & 13.716\%  & 0.080 & 0.076 & 0.950 \\
& XGB\_ET  & -29.868\% & 0.555 & 0.368  & 0.755 & -9.308\%  & 0.051 & 0.038 & 0.705 \\
& XGB\_EL  & -28.603\% & 0.504 & 0.301  & 0.745 & -9.298\%  & 0.051 & 0.038 & 0.710 \\
& XGB\_CUE & -30.264\% & 0.602 & 0.366  & 0.775 & -9.299\%  & 0.051 & 0.038 & 0.710 \\ \hline
\end{tabular}}

\end{table}

\section{Additional results in UK Biobank}\label{suppsec:ukb}

To clarify the disease definitions used in right-censored UK Biobank data analysis, we provide in Table \ref{diseasename} the detailed subcategories included in each disease group. Table \ref{samplesize} further summarizes the sample size, number of events, and censoring rates for the three clinical outcomes.

\begin{table}
\center
\caption{Subcategories of the three diseases in the real data analysis. (\url{https://phewas.genohub.org/phenotypes}) \label{diseasename}}
\begin{tabular}{ll}
\hline
Disease Name & Included Subcategories \\
\hline
\multirow{7}{*}{411: Ischemic heart disease} 
& Unstable angina \\
& Myocardial infarction \\
& Angina pectoris \\
& Coronary atherosclerosis \\
& Aneurysm and dissection of heart \\
& Other chronic ischemic heart disease \\
& Other acute and subacute ischemic heart disease \\
\hline
\multirow{6}{*}{\parbox{5cm}{290: Delirium dementia and\\ amnestic and other cognitive disorders}}
& Dementias \\
& Alzheimer's disease \\
& Dementia with cerebral degenerations \\
& Vascular dementia \\
& Delirium due to other conditions \\
& Persistent mental disorders due to other conditions \\
\hline
\multirow{7}{*}{250: Diabetes mellitus} 
& Type 1 diabetes \\
& Type 2 diabetes \\
& Insulin pump user \\
& Abnormal glucose \\
& Glycosuria or Acetonuria \\
& Polyneuropathy in diabetes \\
& Diabetic retinopathy \\
\hline
\end{tabular}   

\end{table}

\begin{table}
\center
\caption{Overview of sample size events and censoring in three clinical outcomes. IHD refers to ischaemic heart disease; DD refers to delirium dementia; DM refers to diabetes mellitus.\label{samplesize}}
\begin{tabular}{cccc}
\hline
\multicolumn{1}{l}{} & Sample size & Number of events & Censoring rate \\\hline
IHD & 32,810 & 3,944 & 0.880  \\
DD  & 35,722 & 1,039 & 0.971  \\
DM  & 34,213 & 2,664 & 0.922  \\ \hline
\end{tabular}
\end{table}
\begin{figure}[htbp]
\centering
\subfloat[Ischaemic heart disease]{%
\includegraphics[width=0.32\textwidth]{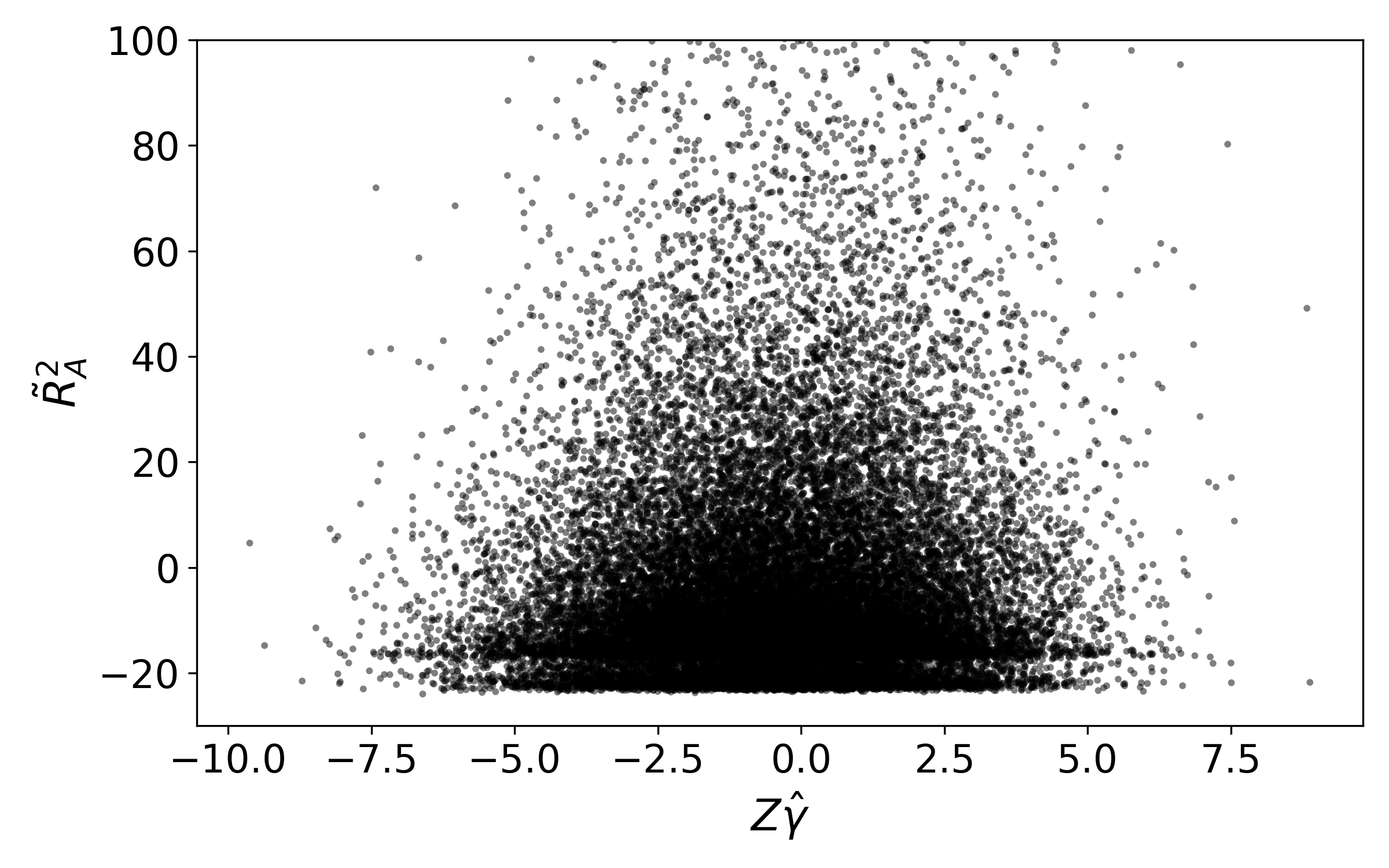}}
\hfill
\subfloat[Delirium dementia]{%
\includegraphics[width=0.32\textwidth]{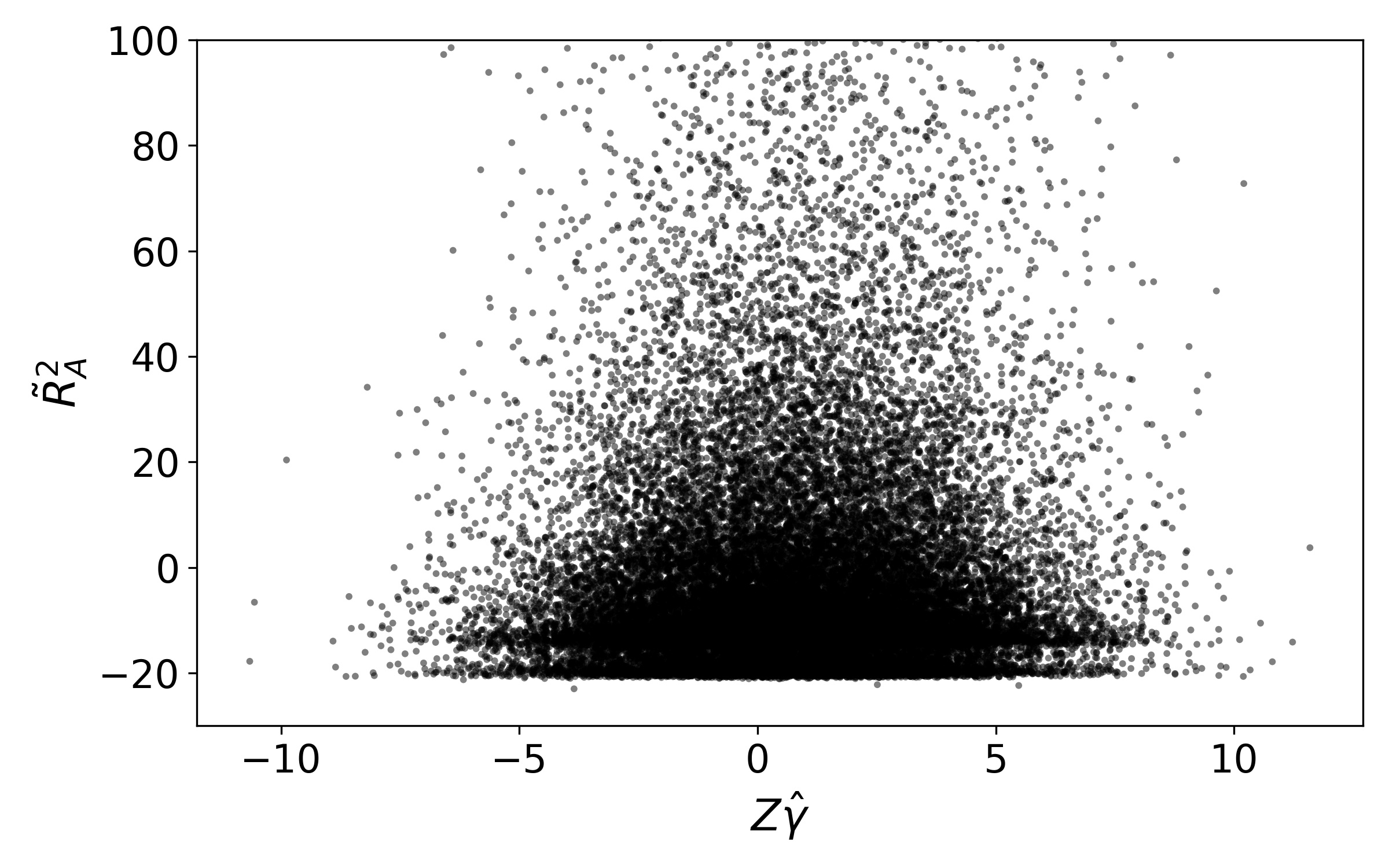}}
\hfill
\subfloat[Diabetes mellitus]{%
\includegraphics[width=0.32\textwidth]{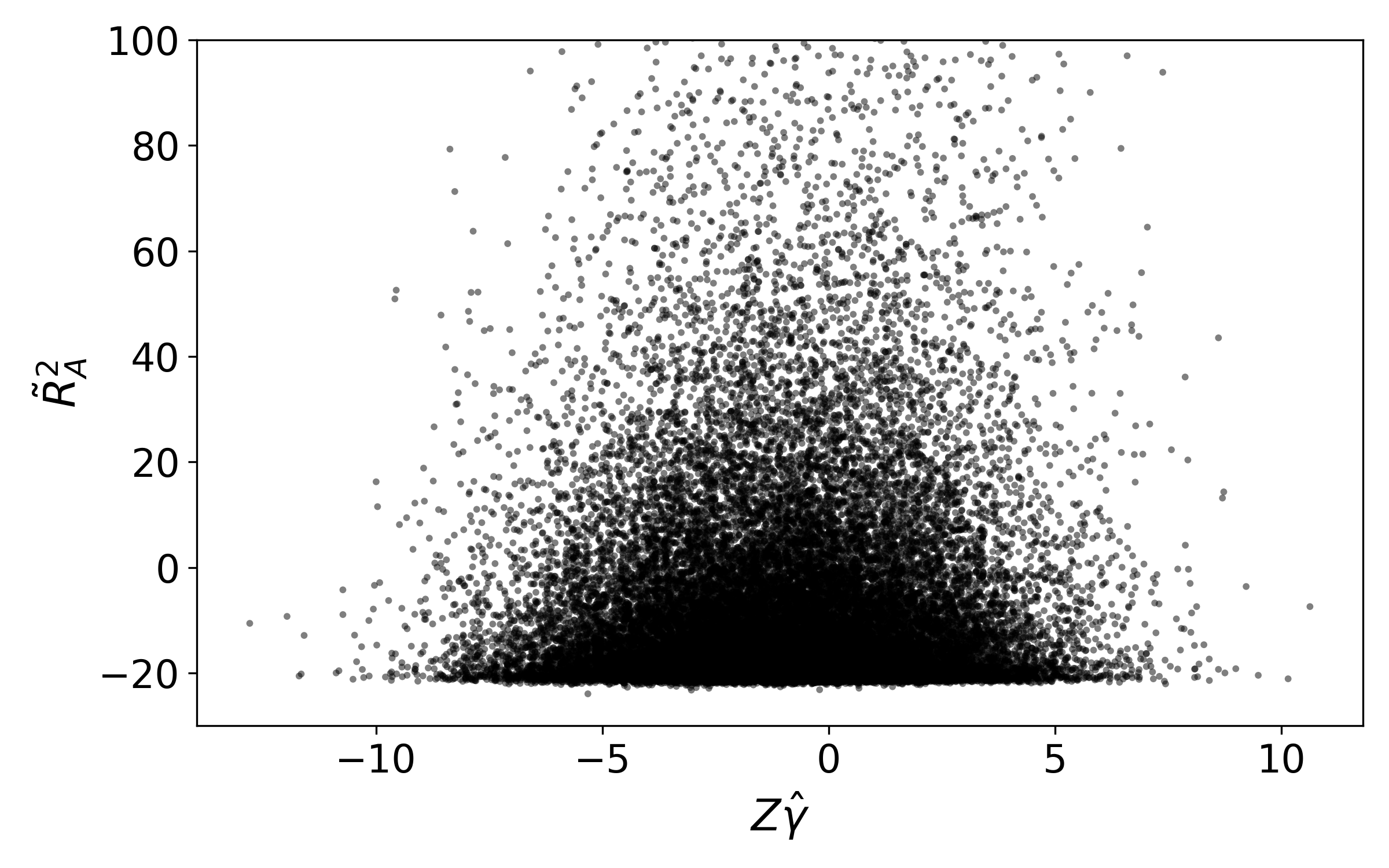}}
\caption{Residual variance plots based on the regression in \eqref{weakid}. The horizontal axis shows the fitted values $\Z\hat{\gamma}$, and the vertical axis represents $\tilde{R}_A^2$. The heteroscedastic structure visible in all three panels indicates that $\var(A\mid \Z,\X)$ varies with $\Z$, providing visual evidence for identification.\label{fig:residual_plots}}

\end{figure}
The residual variance plots in Figure~\ref{fig:residual_plots} further illustrate the non-constant variance in $\var(A | \Z,\X)$, thereby confirming that the heteroscedasticity identification assumption is satisfied. Table \ref{supptab:ihd} presents the estimated causal effect under different nuisance estimators and GEL functions. It can be seen that the results are generally similar across different nuisance estimators and GEL loss functions, except for the case of diabetes mellitus when LR is used to estimate the nuisance functions, where a substantial bias is observed. This may be due to strong nonlinearity in the nuisance components in this setting, which limits the performance of LR. It is also noteworthy that for delirium dementia, the estimated effects are consistently positive across all specifications, further supporting the presence of the obesity paradox.

\begin{table}
\caption{Estimated causal effect of BMI on three disease outcomes. Each cell shows the point estimate with standard error in parentheses, and the corresponding $p$-value is reported in the adjacent column. IHD refers to ischaemic heart disease; DD refers to delirium dementia; DM refers to diabetes mellitus.\label{supptab:ihd}}
\begin{tabular}{lcrcrcr}
\hline
Model & IHD & $p$-value & DD & $p$-value & DM & $p$-value \\
\hline
DNN\_ET  & -0.131 (0.048) & 0.006   & 0.083 (0.021) & \textless{}0.001 & -0.316 (0.055) & \textless{}0.001 \\
DNN\_EL  & -0.132 (0.051) & 0.010   & 0.086 (0.023) & \textless{}0.001 & -0.352 (0.067) & \textless{}0.001 \\
DNN\_CUE & -0.139 (0.042) & \textless{}0.001  & 0.081 (0.021) & \textless{}0.001 & -0.373 (0.032) & \textless{}0.001 \\
LR\_ET   & -0.115 (0.054) & 0.033   & 0.083 (0.025) & \textless{}0.001 &  3.620 (0.002) & \textless{}0.001 \\
LR\_EL   & -0.198 (0.073) & 0.007   & 0.068 (0.059) & 0.249  & -3.010 (0.007) & \textless{}0.001 \\
LR\_CUE  & -0.138 (0.041) & \textless{}0.001  & 0.080 (0.025) & 0.001  &  2.044 (0.001) & \textless{}0.001 \\
RF\_ET   & -0.118 (0.044) & 0.007   & 0.088 (0.024) & \textless{}0.001 & -0.586 (0.017) & \textless{}0.001 \\
RF\_EL   & -0.104 (0.024) & \textless{}0.001  & 0.086 (0.025) & \textless{}0.001 & -0.521 (0.059) & \textless{}0.001 \\
RF\_CUE  & -0.124 (0.046) & 0.007   & 0.088 (0.028) & 0.002  & -0.760 (0.006) & \textless{}0.001 \\
XGB\_ET  & -0.126 (0.042) & 0.003   & 0.094 (0.038) & 0.014  & -0.298 (0.052) & \textless{}0.001 \\
XGB\_EL  & -0.114 (0.030) & \textless{}0.001  & 0.089 (0.015) & \textless{}0.001 & -0.549 (0.083) & \textless{}0.001 \\
XGB\_CUE & -0.137 (0.043) & 0.002   & 0.091 (0.036) & 0.012  & -0.351 (0.031) & \textless{}0.001 \\
\hline
\end{tabular}

\end{table}

We next present results from the UK Biobank using uncensored data to estimate the causal effect of BMI on SBP, as shown in Table \ref{tab:sbp}. This setting does not involve censoring, and serves to demonstrate the performance of our method in a standard continuous outcome framework. The $F$-statistic for weak identification is 11.87 and the $p$-value of the over-identification test is 0.22, indicating that the causal effect is well identified and the moment conditions are compatible with the model assumptions. These diagnostics support the validity of our heteroscedasticity-based identification strategy in this uncensored-data setting. Using the DNN learner with the EL objective as an example, we obtain $\hat\beta = 0.261$ with a standard error of $0.061$, giving a 95\% confidence interval of $[0.141,0.381]$.  Hence a one-unit ($1\ \mathrm{kg/m^{2}}$) increase in BMI is associated with an average rise of about $0.261\ \mathrm{mmHg}$ in systolic blood pressure. All other estimator pairs yield positive coefficients of comparable magnitude, and the 95\% confidence intervals for all estimates exclude zero, confirming statistical significance. These findings align with the Mendelian-randomisation results of \cite{ye2024}, who reported $\hat\beta = 0.140$ (95\% CI: [0.005, 0.275]) for the effect of BMI on SBP in UK Biobank. The slight difference between our results and those of \cite{ye2024} may be attributed to covariate adjustment: while \cite{ye2024} did not include covariates, we incorporate age, which is known to have a strong effect on blood pressure \citep{franklin1997}. Nonetheless, both analyses find a statistically significant positive association between BMI and SBP, reinforcing the robustness of this causal relationship.

\begin{table}
\center
\caption{Estimated causal effect of BMI on systolic blood pressure. Each entry is point estimation of the effect with standard error in parentheses.\label{tab:sbp}}
\begin{tabular}{lccc}  
\hline
& CUE          & EL           & ET           \\ \hline
DNN          & 0.268 (0.061) & 0.261 (0.061) & 0.268 (0.060) \\
LR           & 0.314 (0.061) & 0.308 (0.061) & 0.314 (0.060) \\
XGB          & 0.311 (0.062) & 0.305 (0.062) & 0.310 (0.061) \\
RF           & 0.274 (0.067) & 0.276 (0.079) & 0.273 (0.066) \\ \hline
\end{tabular}
\end{table}

\bibliography{ref}

\end{document}